\documentclass[a4paper,twocolumn,notitlepage,superscriptaddress,floatfix]{revtex4-2}
\usepackage{amsmath}
\usepackage{amssymb}
\usepackage{amsthm}
\usepackage{graphicx}
\usepackage{xcolor}
\usepackage[english]{babel}
\usepackage{caption}
\usepackage{csquotes}
\usepackage{braket}
\usepackage[bookmarks=true,colorlinks=true,linkcolor=orange,citecolor=orange,urlcolor=orange,bookmarks]{hyperref}
\usepackage{mathtools}
\usepackage{booktabs}
\usepackage{enumitem}
\usepackage[normalem]{ulem}
\usepackage{verbatim}
\usepackage{bbm}

\usepackage[framemethod=TikZ]{mdframed}

\makeatletter
\def\maketitle{
	\@author@finish
	\title@column\titleblock@produce
	\suppressfloats[t]}
\makeatother

\mdfsetup{%
	middlelinecolor=gray,
	middlelinewidth=1pt,
	roundcorner=3pt,
	leftmargin=-\parindent,
	rightmargin=-\parindent,
	innertopmargin=0pt, 
	innerbottommargin=1em,
	skipabove=\medskipamount,
	skipbelow=\medskipamount,}

\surroundwithmdframed{putative}
\surroundwithmdframed{proposition}
\surroundwithmdframed{corollary}
\surroundwithmdframed{theorem}
\surroundwithmdframed{lemma}
\surroundwithmdframed{definition}
\surroundwithmdframed{remark}

\newcommand*{\pipe}{\ensuremath{\, | \, }}
\newcommand*{\bigpipe}{\ensuremath{\, \big| \, }}

\newcommand*{\fatpipe}{\ensuremath{\, \| \, }}
\newcommand*{\Tr}{\operatorname{Tr}}

\providecommand*{\poly}{\ensuremath{\operatorname{poly}}}
\newcommand*{\negl}{\ensuremath{\operatorname{negl}}}

\newcommand{\flimsup}{\operatornamewithlimits{\underline{lim\, sup}}}
\newcommand{\fliminf}{\operatornamewithlimits{\underline{lim\, inf}}}
\newcommand{\flim}{\operatornamewithlimits{\underline{lim}}}

\providecommand{\calA}{\ensuremath{\mathcal{A}}}

\providecommand{\calC}{\ensuremath{\mathcal{C}}}
\providecommand{\calD}{\ensuremath{\mathcal{D}}}
\providecommand{\calE}{\ensuremath{\mathcal{E}}}

\providecommand{\calH}{\ensuremath{\mathcal{H}}}

\providecommand{\calK}{\ensuremath{\mathcal{K}}}

\providecommand{\calM}{\ensuremath{\mathcal{M}}}
\providecommand{\calN}{\ensuremath{\mathcal{N}}}

\providecommand{\calQ}{\ensuremath{\mathcal{Q}}}

\providecommand{\calT}{\ensuremath{\mathcal{T}}}

\providecommand{\calV}{\ensuremath{\mathcal{V}}}
\providecommand{\calW}{\ensuremath{\mathcal{W}}}
\providecommand{\calX}{\ensuremath{\mathcal{X}}}

\providecommand{\bbC}{\ensuremath{\mathbb{C}}}

\providecommand{\bbI}{\ensuremath{\mathbb{I}}}

\providecommand{\bbN}{\ensuremath{\mathbb{N}}}

\providecommand{\bbP}{\ensuremath{\mathbb{P}}}
\providecommand{\bbQ}{\ensuremath{\mathbb{Q}}}
\providecommand{\bbR}{\ensuremath{\mathbb{R}}}

\AtBeginDocument{
	\heavyrulewidth=.08em
	\lightrulewidth=.05em
	\cmidrulewidth=.03em
	\belowrulesep=.65ex
	\belowbottomsep=0pt
	\aboverulesep=.4ex
	\abovetopsep=0pt
	\cmidrulesep=\doublerulesep
	\cmidrulekern=.5em
	\defaultaddspace=.5em
}
\newcommand{\haar}[0]{\operatorname{Haar}}

\usepackage{thmtools}
\usepackage[capitalize,nameinlink]{cleveref}
\newtheorem*{mttheorem}{Theorem}
\newtheorem*{mtlemma}{Lemma}
\newtheorem*{mtremark}{Remark}
\newtheorem*{mtdefinition}{Definition}
\newtheorem{theorem}{Theorem}
\numberwithin{theorem}{section}
\newtheorem{lemma}[theorem]{Lemma}
\newtheorem{proposition}[theorem]{Proposition}
\newtheorem{definition}[theorem]{Definition}
\newtheorem{corollary}[theorem]{Corollary}

\newtheorem{remark}[theorem]{Remark}

\newtheoremstyle{redthm}{}{}{\color{red}}{}{\color{red}\bfseries}{}{ }{}
\theoremstyle{redthm}

\newlist{propenum}{enumerate}{1} 
\setlist[propenum]{label=(\roman*), ref=\thetheorem~(\roman*)}
\crefalias{propenumi}{proposition}

\newcommand{\JJM}[1]{{\small\textcolor{blue}{[jjm: #1]}}}

\newcommand{\colonsimeq}{\hyperref[remark:notation]{:\simeq}}

\newcommand{\dccqs}{Dahlem Center for Complex Quantum Systems, Freie Universit{\"a}t Berlin, 14195 Berlin, Germany}
\newcommand{\hzb}{Helmholtz-Zentrum Berlin f{\"u}r Materialien und Energie, 14109 
	Berlin, Germany}
\newcommand{\hhi}{Fraunhofer Heinrich Hertz Institute, 10587 Berlin, Germany}

\newcommand{\nocontentsline}[3]{}

\begin{document}	
	\let\oldaddcontentsline=\addcontentsline
	\let\addcontentsline=\nocontentsline
	
	\title{Computational relative entropy}
	
	\date{\today}
	
	\author{Johannes~Jakob~Meyer}
	\affiliation{\dccqs}
	
	\author{Asad~Raza}
	\affiliation{\dccqs}
	
	\author{Jacopo~Rizzo}
	\affiliation{\dccqs}
	
	\author{Lorenzo~Leone}
	\affiliation{\dccqs}
	
	\author{Sofiene~Jerbi}
	\affiliation{\dccqs}
	
	\author{Jens~Eisert}
	\affiliation{\dccqs}
	\affiliation{\hzb}
	\affiliation{\hhi}
	
	\begin{abstract}
		Our capacity to process information depends on the computational power at our disposal. Information theory captures our ability to distinguish states or communicate messages when it is unconstrained with unrivaled beauty and elegance. For computationally bounded observers the situation is quite different -- they can, for example, 
		be fooled to believe that distributions are more random than they actually are. 
		Existing mathematical approaches in computational information theory largely follow the single-shot paradigm that, while being operationally meaningful, also gives complicated statements and is difficult to build intuition for.
		In our work, we build a new foundation for a computational quantum information theory that captures the essence of complexity-constrained information theory while retaining 
		the look and feel 
		of the unbounded asymptotic theory.
		As our fundamental quantity, we define the computational relative entropy as the optimal error exponent in asymmetric hypothesis testing when restricted to polynomially many copies and quantum gates, defined in a mathematically rigorous way.
		Building on this foundation, we prove a computational analogue of Stein's lemma, establish computational versions of fundamental inequalities like Pinsker's bound, and demonstrate a computational smoothing property showing that computationally indistinguishable states yield equivalent information measures.
		We derive a computational entropy that operationally characterizes optimal compression rates for quantum states under computational limitations and show that our quantities apply to computational entanglement theory, proving a computational version of the Rains bound. 
		Our framework reveals striking separations between computational and unbounded information measures, including quantum-classical gaps that arise from cryptographic assumptions, demonstrating that computational constraints fundamentally alter the information-theoretic landscape and open new research directions at the intersection of quantum information, complexity theory, and cryptography. 
	\end{abstract}
	
	\maketitle
	
	It has been clear since Holevo’s seminal work that Shannon's theory of information~\cite{shannon1948mathematical} takes on a fundamentally new character when quantum systems serve as carriers of information~\cite{Holevo}. 
	At first, quantum effects were regarded as obstacles -- potential hindrances to reliable communication. These concerns later proved unfounded: Quantum effects opened the door to entirely new paradigms of information processing that go beyond the capabilities of classical communication~\cite{MarkWilde}. On the conceptual side, quantum information theory became a sprawling field in its own right, with the non-commutativity of quantum theory giving rise to much richer mathematics. On the side of applications, quantum teleportation and quantum key distribution laid the groundwork for practical quantum technologies~\cite{NewQuantumCryptoReview}. Today, the mathematically rigorous framework of quantum information theory is widely recognized as one of the pillars supporting the broader field of quantum technologies~\cite{MarkWilde,tomamichel2016quantum,watrous2018theory,khatri2024principles}.
	
	In parallel with the rise of quantum communication theory, the concept of quantum computing emerged. Initially formulated as an abstract challenge to the Church–Turing thesis~\cite{Feynman-1986,QuantumTuring}, it soon evolved into a theoretically significant idea~\cite{NielsenChuang} and, ultimately, into a technological paradigm for the supercomputers of tomorrow~\cite{preskill2013quantum,MindTheGaps}. Quantum computing is now regarded worldwide as one of the leading candidates for defining the computational paradigms of the future~\cite{Roadmap}.
	
	Given the significance of both developments, it is striking how little interaction there has been between the two fields. Quantum information theory is primarily concerned with the resourcefulness of quantum states, while quantum computation focuses on computational capabilities. Yet, a principle that is second nature in quantum computing has not been equally central to quantum information theory: namely, that all physically meaningful operations must be computationally efficient. 
	
	In the classical world, steps to create a \emph{computational information theory} have been undertaken~\cite{yao1988computational}, with first steps already undertaken to generalize it to the quantum realm.
	The aim of this work is to establish a firm foundation for \emph{computational quantum information theory}. While existing works tackle questions of computational complexity related to information processing tasks, we follow a more ambitious goal: to lay the foundation for a theory of quantum information processing that has the familiar look and feel of the theory we know, but that takes into account that physically relevant applications involve both computations of polynomial complexity and also at most a polynomial amount of samples. We establish a mathematically rigorous blueprint for defining such computational analogs of information-theoretic quantities. Our main focus is the application of this blueprint to one of the most fundamental quantities in information theory, as many applications follow from this: the relative entropy.
	In analogy to unbounded information theory, where the relative entropy arises as the asymptotically optimal exponent in asymmetric hypothesis testing, we introduce the \emph{computational relative entropy} as the best error exponent by computationally feasible tests applied to a polynomial number of copies. As the relative entropy functions as a \enquote{mother quantity} for many other information-theoretic measures, we thus establish a new fundamental quantity in the computational setting. The foundational character of the computational relative entropy is underscored by the fact that the \emph{computational entropy} derived from it has the interpretation as the best rate of compression when all operations need to be computationally efficient. 
	
	We believe that our framework captures the right amount of abstraction: broad enough to enable general statements, but fine-grained enough to ensure that computational challenges are not brushed aside. We achieve this by considering \emph{any} computation to be efficient when its complexity scales polynomially -- achieving a similar degree of abstraction as the complexity class $\mathsf{P}$ does in computational complexity theory. At the same time, we introduce the concept of \emph{polynomial regularization}, where an information-theoretic quantity is averaged over an arbitrarily large polynomial number of copies of the state. This allows us to regain some of the regularity and simplicity of the familiar asymptotic quantum information theory, while always retaining that the resources at our disposal cannot be arbitrarily large.
	In this way, our approach to computational quantum information theory recovers the structure already familiar from the unbounded quantum information theory, making understanding and manipulating the involved quantities much more intuitive.
	
	In this way, we go beyond existing approaches that integrate computational constraints with information theory that quantify the achievable limits with a single copy of a state and a fixed number fo quantum gates to realize the information processing task~\cite{chen2017computational,munson2025complexity-constrained,avidan2025quantum,avidan2025fully}.
	While these \enquote{single-shot} approaches have considerable merit for their intended applications, we believe that they are too fine-grained to capture the underlying structure of computational quantum information theory.
	
	We believe that the computational relative entropy we define will also be a key quantity to analyze physical phenomena through the lens of computational complexity. The necessity of this is evidenced by recent developments in quantum information science as a whole that show that there exist quantum states that are \emph{pseudoentangled}~\cite{Pseudoentanglement,leone2025entanglementtheorylimitedcomputational,arnon2023computationalentanglementtheory}, i.e.,  that have little entanglement but are indistinguishable from maximally entangled states. Similar reasoning carries over to \emph{pseudomagic} states~\cite{PhysRevLett.132.210602} or even \emph{pseudochaotic} ensembles of states that do not feature any of the physical properties commonly associated with  quantum chaotic systems~\cite{gu2024simulating}. 
	Such considerations suggest that notions of quantum information theory, in fact -- deep considerations of what we can actually meaningfully say in physics and what not -- are substantially influenced by the computational power of the observer.
	To showcase how these questions can be tackled in our framework, we study the paradigmatic task of entanglement manipulation.
	Among other results, we establish a computational version of the Rains bound that shows that the computational distillable entanglement cannot exceed the computational relative entropy to the closest quantum state with positive partial transpose. This result beautifully showcases the merits of our approach to computational quantum information theory: the result takes the exact same \emph{form} as the original Rains bound while capturing computational limitations, it is embedded into a wider theory such that it allows us to establish new non-trivial results from properties of the computational relative entropy alone and it showcases how our systematic approach can still recover results from publications that focus on entanglement manipulation alone.
	
	\textbf{Note:} During the preparation of this manuscript, we became aware of concurrent work by Yángüez, Hahn, and Kochanowski undertaken independently from ours~\cite{Yanguez2025MinMax}. Our works share the spirit of treating hypothesis testing under complexity limitations, resulting in works complementing each other as they offer different angles on related questions.
	
	\subsection*{Related works}\label{subsec:related-works}
	We have not found computational analogues of the quantum \emph{relative entropy} in the literature. 
	Computational analogues of \emph{entropy}, however, have been studied extensively in the classical literature. Starting from Shannon, one can operationally think of his definition of entropy (Shannon entropy) as the metric that governs the limits of data compression such that the decompressed message is faithfully recovered. More than three decades later, Yao~\cite{10.5555/1382436.1382790} defined pseudoentropy with one additional constraint: the encoding and decoding processes must be computationally efficient. Håstad et al.~\cite{hastad_pseudorandom_1999} took another view on computational entropies. There, the authors define a distribution $p$ to have pseudoentropy at least $k$ if there exists a $\tilde{p}$ computationally indistinguishable from $p$ with $\tilde{p}$ having entropy at least $k$. These entropies are now commonly called HILL pseudoentropies. Interestingly, one can reverse the quantifiers in HILL pseudoentropy to obtain yet another definition of a pseudoentropy, called \enquote{metric} pseudoentropy~\cite{goos_computational_2003}. Concretely, reversing quantifiers, one gets that a distribution $p$ is said to have metric pseudoentropy $k$ if for \emph{every} distribution $\tilde{p}$ of entropy $k$, there {\em exists} computationally efficient distinguisher that has the the same advantage (up to negligible factors) over both $p$ and $\tilde{p}$. It is then natural to ask whether these three definitions are equivalent. Barak et al.~\cite{goos_computational_2003} established that HILL and metric entropies are equivalent up to small changes in the size of the distinguishing circuits -- an equivalence that was later used to prove tighter security for cryptographic tasks~\cite{fuller_unified_2012, dziembowski_leakage-resilient_2008}. On the quantum side, authors in Ref.~\cite{chen_computational_2017} defined analogous quantum HILL and metric pseudoentropies and proved an analogous equivalence in the quantum case. 
	Separation for Yao and HILL entropies was also established in the so-called shared random string model in Ref.~\cite{hutchison_conditional_2007}. Applications wise, classical pseudoentropies have been extremely useful in theoretical computer science, especially in cryptography and complexity theory: from pseudorandom generators~\cite{hastad_pseudorandom_1999}, various computational tasks in cryptography~\cite{dziembowski_leakage-resilient_2008,fuller_unified_2012,haitner_inaccessible_2021}, hardness amplification in complexity theory~\cite{492584} to, surprisingly, additive number theory~\cite{green_primes_2007}.
	
	More recently, on the quantum side, instead of measuring indistinguishability of a low entropy state relative to a high entropy state in a HILL-type fashion, Avidan et al. \cite{avidan2025fully} define quantum computational min and max entropies with the mindset of operationally understanding their properties in the setting where all operations are now computationally restricted. Indeed, the two approaches seem very complementary. However, Avidan et al. show that their quantum computational min entropy and quantum HILL entropy do not dominate each other~\cite{avidan2025fully}. 
	In a direction aligned with this work, Avidan et al.~\cite{avidan2025fully} define smoothed versions of their quantum computational min- and max-entropies, where smoothing is carried out under computational constraints, in the same spirit as our computational smoothing of relative entropy. Avidan and Arnon~\cite{avidan2025quantum} have also studied the classical {\em unpredictablility entropy}~\cite{hutchison_conditional_2007} in the quantum setting where the goal is to directly bound the computational complexity of a guessing circuit. They showed that the unpredictability entropy satisfies a leakage chain rule and give an application for pseudorandomness extraction against bounded quantum adversaries. 
	
	The central definition of our work, computational relative entropy, builds on the definition of a complexity restricted single-shot variant of hypothesis testing relative entropy defined by Munson et al.~\cite{munson2025complexity-constrained}. 
	
	\section{Technical contributions}
	Let us summarize the main technical contributions of this work. We quantify the complexity of a unitary transformation, by the number of two-local \emph{gates} needed to implement it. This definition extends to quantum channels and measurements by giving access to an unbounded amount of ancillas in the all-zero state and considering tracing out subsystems and projections onto the zero states as free operations (see \cref{section:complexity_limited_quantum_information_processing} for formal definitions). We note that the number of ancillas never exceeds the number of gates by construction. We will use the notation $C(\cdot)$ to refer to the gate complexity of a unitary, channel or measurement operator in this computational model. 
	
	\subsection{Rigorous Treatment of Polynomial Resources}
	Complexity theory has been very successful at capturing the essence of computational hardness because it describes problems at a good level of abstraction. Most notable here is the complexity class $\mathsf{P}$ of decision problems solvable in polynomial time (and its relatives). Not caring about the particular exponent of the polynomial allows one to use other algorithms of polynomial complexity to convert between problems, creating a huge equivalence class. Nevertheless, it well captures the idea that there are problems that are \emph{harder} in the sense that we do not expect that a polynomial time algorithm exists for them. 
	
	In this work, we will likewise consider computations that have polynomial complexity in some scaling parameter as efficient, giving us a similar level of coarse-graining while still capturing the essence of computational constraints, as generically, the optimal tests used in statistics would require exponential complexity. Asymptotic statistics emerges by running statistical tests on infinitely many copies, from which we obtain quantities like the relative entropy by \emph{regularizing}. Indeed, we could even \emph{define} the relative entropy as the regularized asymptotic error exponent of the type II error of hypothesis testing~\cite{hiai1991proper,khatri2024principles}
	\begin{align}\label{eqn:relative_entropy_definition_by_regularization}
		D(\rho \fatpipe \sigma) = \lim_{\epsilon \to 0}\lim_{m \to \infty}\frac{1}{m} D_h^{\epsilon}(\rho^{\otimes m}\fatpipe \sigma^{\otimes m}),
	\end{align}
	where the hypothesis testing relative entropy $D_h^{\epsilon}(\rho \fatpipe \sigma) \coloneqq -\log \min_{0\leq \Lambda \leq \bbI}\{ \Tr[\Lambda \sigma] \pipe \Tr[\Lambda \rho] \geq 1- \epsilon\}$ is a single-shot definition of the hypothesis testing type II error exponent~\cite{khatri2024principles}.
	We wish to do the same, but only allowing for \emph{polynomially} many copies to regularize over. 
	We would like to impose that the computation that is performed to test the states has polynomial complexity as well.
	To capture the ultimate limits of what can be done with polynomial resources -- be it the number of copies or the complexity of the computation -- we let the resource scale as $n^k$ and then let the exponent $k\to\infty$. Our first technical contribution is to turn this into a mathematically rigorous prescription and to ensure that this limit exists and is non-trivial.
	
	On a formal level, we let $n \in \bbN$ be the scaling variable in which we measure complexity and consider sequences of functions from $\bbN$ to $[0,\infty]$, \smash{$\{ f_k \}_{k\in\bbN} \subseteq [0,\infty]^{\bbN}$}. Our desire is to find a mathematically rigorous way to capture the limit of such a sequence as $k\to\infty$. One could, in principle consider the \emph{point-wise} limit of such functions, which is the function \smash{$n \mapsto \lim_{k\to\infty} f_k(n)$}. For our purposes, however, this would recover the setting of unbounded resources as $\lim_{k\to\infty} n^k = \infty$ for any fixed $n>1$.
	
	We overcome this issue using the theory of \emph{complete lattices}~\cite{roman2008lattices}, which are partially ordered sets for which all subsets have well-defined infima and suprema inside of the set itself. In a complete lattice, both limes inferior and superior are globally well-defined. Intuitively, our construction should capture that an experimenter with access to exponential resources $e^n$ should would eventually outperform any polynomially constrained experimenter. To do this, we build a complete lattice from the set of functions $[0,\infty]^{\bbN}$ under the \emph{eventual domination order} which lets $f \leq^* g$ if $f(n) \leq g(n)$ except for all but finitely many $n \in \bbN$. A key step in our construction is the \emph{Dedekind-MacNeille completion} which yields the smallest complete lattice that contains a given partially ordered set, while being compatible with the ordering and preserving all limits that already exist in the partially ordered set. The full details of this construction can be found in \cref{sec:limits_of_polynomial_resources}. For our purposes here it suffices that we can meaningfully define limits of polynomial resources in a way that is compatible with our approaches to bound them. We will use the notations (see \cref{def:limits})
	\begin{align}
		\fliminf_{k\to\infty} f_k \text{, } \flimsup_{k\to\infty} f_k \text{ and } \flim_{k\to\infty} f_k
	\end{align}
	to refer to our limit construction, the latter being understood to exist when limes superior and limes inferior agree, which is critically the case for any sequence of functions that is monotonic with respect to the eventual domination order.
	
	We complement this with one further concept: \emph{negligible} functions. These are functions that vanish faster than any polynomial, denoted as $f \in \negl(n)$ or $f(n) = \negl(n)$. They are ubiquitous in cryptographic applications, intuitively because negligible differences cannot be resolved by a polynomially bounded observer. As such, differences up to negligible factors can be \enquote{neglected} in a theory with computational constraints. We hence consider quantities the same if they only differ by a negligible function and show that this enables a nice and coherent theory of \emph{polynomial regularization}. Concretely, we consider two functions $f, g \in [0,\infty]^{\bbN}$ to be equivalent up to negligible functions, or $f \simeq g$, if $|f-g| \in \negl(n)$. This means that negligible functions are exactly the ones that are equivalent to the zero function, $f \in \negl(n) \ \Leftrightarrow \ f \simeq 0$.
	Additionally, we have the natural ordering relation $f \lesssim g$ if there exists a $g' \simeq g$ such that $f \leq g'$ point-wise.
	
	This notation allows us to make sense of the following lemma, which crucially allows us to \enquote{regularize away} terms that scale at most polynomially.
	\begin{mtlemma}[\cref{lemma:limit_of_inverse_polynomials_is_negligible}]
		Consider the sequence of functions $\{ n^{-k} \}_{k\in\bbN} \subseteq [0,\infty]^{\bbN}$.
		Then,
		\begin{align}
			\flim_{k \to \infty} n^{-k} \simeq 0.
		\end{align}
	\end{mtlemma}
	The above lemma is key to us recovering the look and feel of asymptotic information theory, while at the same time ensuring that limitations of computational power and the number of copies are taken into account. 
	
	In the remainder of this work, we will always use equality and ordering up to negligible functions ($\simeq, \lesssim$) when talking about computational quantities, especially those defined via the limit construction touched upon above. In many cases, the stricter notions of eventual equality and eventual dominance ($=^*, \leq^*)$ or point-wise equality and dominance ($=,\leq$) could be used, but we believe that this creates a more readable mathematical picture while retaining the spirit of taking computational limitations properly into account.
	
	\subsection{Computational relative entropy and total variation distance}
	The relative entropy $D(\rho \fatpipe \sigma)$ is arguably one of the most important quantities in (asymptotic) information theory. It quantifies the achievable errors exponents in a wide range of tasks via the asymptotic equipartition property. As error exponents are central to quantifying applications like transmission and coding, the performance quantifiers of these tasks are often obtained from the relative entropy. Prominent examples are entropy $S(\rho) = -D(\rho \fatpipe \bbI)$ -- operationally relevant as the best rate at which a source of random bits can be compressed -- and mutual information $I(A:B) = D(\rho^{AB} \fatpipe \rho^A \otimes \rho^B)$ -- which gives the communication capacity of a classical channel. It also has massive applications in physics, where it for example captures the ultimate form of the second law of thermodynamics in the form of the data processing inequality of the relative entropy.
	
	The relative entropy, however, does not take computational restrictions into account, which means that it cannot necessarily be achieved when restricted to polynomially many copies and to polynomial-time processing of said copies. There are three possible reasons for this: (i) a polynomial number of copies can be insufficient to \enquote{regularize away} finite-size corrections, (ii) a polynomial number of quantum gates is not sufficient to reach the eigenbasis of an optimal test and (iii) even if we can reach an optimal basis for performing a test, it might be \emph{computationally hard} to map measurement outcomes to the correct predictions. 
	
	Our second main contribution is a mathematically rigorous definition of a \emph{computational relative entropy} that aims to fulfill the same role in computational information theory as the relative entropy does in asymptotic information theory. To do so, we combine the polynomial regularization introduced in the previous section with the $G$-complexity relative entropy of Ref.~\cite{munson2025complexity-constrained}. The regular hypothesis testing relative entropy quantifies the optimal achievable error exponent in hypothesis testing between two states in a single-shot way. The definition of Ref.~\cite{munson2025complexity-constrained} 
	\begin{align}
		\begin{split}D_h^{\epsilon}&(\rho \fatpipe \sigma ; G) \coloneqq \\
			&-\log \min \big \{ \Tr[ \Lambda \sigma] \pipe \Tr[\Lambda \rho] \geq 1-\epsilon, C(\Lambda) \leq G\big\}
		\end{split}
	\end{align}
	adds the requirement that the gate complexity of realizing the test must not exceed a constant $G$.
	We define the computational relative entropy similar in spirit to how the relative entropy emerges (see \cref{eqn:relative_entropy_definition_by_regularization}) from the hypothesis testing relative entropy. We employ a polynomial number of copies $n^k$ and allow for $G = n^{k\ell}$ many gates. Taking both $k$ and $\ell$ to infinity and $\epsilon$ to zero yields the computational relative entropy.
	\begin{mtdefinition}[\cref{def:computational_relative_entropy}]
		Consider two quantum states $\rho_n, \sigma_n \in \calH_n$. Then, their \emph{computational relative entropy} is defined as
		\begin{align}
			\underline{D}(\rho_n \fatpipe \sigma_n) \colonsimeq \flim_{\epsilon \to 0} \flim_{\ell \to \infty} \fliminf_{k\to \infty} \frac{1}{n^k} D_h^{\epsilon}(\rho_n^{\otimes n^k} \fatpipe \sigma_n^{\otimes n^k}; n^{k\ell}),
		\end{align}
	\end{mtdefinition}
	We note that the above definition only makes sense for \emph{families} of quantum states $\rho_n$ and $\sigma_n$, because we need a scaling variable in which scalings can be polynomial. 
	The symbol $:\simeq$ indicates that the above definition is to be understood in the sense we introduced in the previous section.
	\begin{mtremark}[Notation]
		We will always use an underline to refer to computationally bounded quantities which are often also polynomially regularized.
	\end{mtremark}
	
	\paragraph*{Properties.} The computational relative entropy has many natural properties that mirror the corresponding properties of the relative entropy, as shown in \cref{proposition:properties_of_computational_relative_entropy}. Most importantly, it fulfills data-processing under channels of polynomial complexity: Let $\Phi_n$ be a quantum channel such that $C(\Phi_n) \leq O(\poly(n))$, then 
	\begin{align}
		\underline{D}(\Phi_n[\rho_n] \fatpipe \Phi_n[\sigma_n]) \lesssim\underline{D}(\rho_n \fatpipe \sigma_n). 
	\end{align}
	It also has some properties that the relative entropy does not have. While it is, like the relative entropy, super-additive with respect to tensor products of different quantum states, we show in \cref{theorem:explicit_examples_of_catalysis_and_superadditivity} that this super-additivity can be strict. More so, we show that the computational relative entropy exhibits \emph{catalysis}, where a state that is \enquote{useless} on its own allows one to gain performance. Formally, we give three quantum states $\rho_n, \sigma_n$ and $\tau_n$ such that
	\begin{align}
		\underline{D}(\rho_n \fatpipe \sigma_n) \simeq \underline{D}(\tau_n \fatpipe \tau_n) \simeq 0,
	\end{align}
	but
	\begin{align}
		\underline{D}(\rho_n\otimes \tau_n \fatpipe \sigma_n \otimes \tau_n) \gtrsim C > 0. 
	\end{align}

	\paragraph*{Computational total variation distance.}
	In standard information theory, relative entropy relates to total variation distance via the well-known Pinker's inequality. While our techniques for polynomial regularization are by definition described in multi-shot setting (in the case of computational relative entropy, for example), we can define computational trace distance between any two states in a single-copy setting by only taking the polynomial limit with respect to the number of gates. We formalize this by defining the $G$-complexity total variation distance (\cref{def:g_complexity_symmetric_hypothesis_testing_error})
	\begin{align}
		d_{\mathrm{TV}}(\rho, \sigma; G) \coloneqq \max \big\{ \tfrac{1}{2}\Tr[\Lambda (\rho - \sigma)] \pipe C(\Lambda) \leq G \big\}
	\end{align}
	and then taking limit over a polynomial number of gates to obtain the computational total variation distance (\cref{definition:computational_total_variation_distance})
	\begin{align}
		\underline{d}_{\mathrm{TV}}(\rho_n, \sigma_n) \colonsimeq \flim_{\ell \to \infty} d_{\mathrm{TV}}(\rho_n, \sigma_n; n^{\ell}).
	\end{align}
	We show that the computational total variation distance is related to the computational relative entropy by virtue of a computational Pinsker's inequality (\cref{theorem:computational_pinskers_inequality}):
	\begin{align}
		\underline{d}_{\mathrm{TV}}(\rho_n, \sigma_n) \lesssim \sqrt{\tfrac12 \underline{D}(\rho_n \fatpipe \sigma_n)}.
	\end{align} 
	We also derive a computational Bretagnolle-Huber inequality (\cref{theorem:computational_bretagnolle_huber}) which is often tighter. Interestingly enough, the proofs mirror their unbounded counterparts and mainly rest on data processing under the measurement optimal for the computational total variation distance which has polynomial complexity by definition. We think that the fact that we have a similar relation between these quantities as in the unbounded setting is a positive sign that our strategy of defining computational information-theoretic quantities via polynomial limits and regularization and, in particular, our definition of computational relative entropy captures the essence of quantum information processing under computational restrictions while recovering the look and feel of asymptotic information theory.
	
	\subsection{Computational smoothing and separations in hypothesis testing} 
	Computational indistinguishability is a key concept in computational information theory and cryptography, referring to the fact that two states $\rho_n$ and $\tilde\rho_n$ can be distinct, but any polynomially bounded observer can only ever realize a negligible difference in the output of any algorithm applied to the two states. If two states are computationally indistinguishable, we write $\rho_n \approx_c \tilde\rho_n$.
	Another concept central to information theory are smoothed divergence measures. We combine both concepts and show that the computational relative entropy is invariant under ``computational smoothing'', where we optimize over all states computationally indistinguishable from the first argument (\cref{theorem:undetectability_of_computational_smoothing}): 
	\begin{align}
		\underline{D}(\rho_n \fatpipe \sigma_n) \simeq \inf_{\tilde\rho_n \approx_c \rho_n}\underline{D}(\tilde\rho_n \fatpipe \sigma_n). 
	\end{align}
	
	Computational smoothing crucially allows us to establish separations between unbounded and computational relative entropy. The \enquote{workhorse} of this approach is the fact that for any two computationally indistinguishable states $\rho_n \approx_c \tilde\rho_n$ we have $\underline{D}(\rho_n \fatpipe \tilde\rho_n) \simeq 0$ (\cref{corollary:zero_rate_for_computationally_indistinguishable_states}).
	
	We exemplify this fact that distributions can be much harder -- or even impossible -- to distinguish for a computationally bounded observer relative to an unbounded one by giving explicit examples of separations. We further use tools from cryptography to establish a \emph{quantum-classical} separation where we show that the computational relative entropy restricted to classical processing can be smaller than the computational relative entropy when allowed quantum processing.
	\begin{mttheorem}[\cref{theorem:separations_computational_relative_entropy}]
		There exists states $\rho_n, \sigma_n \in (\bbC^2)^{\otimes n}$ such that 
		\vspace{-4pt}
		\begin{enumerate}[label=(\roman*)]
			\item $D(\rho_n \fatpipe \sigma_n) = \infty$, but $\underline{D}(\rho_n \fatpipe \sigma_n) \simeq 0$, unconditionally. The states have exponential complexity.
			\item $D(\rho_n \fatpipe \sigma_n) = \infty$, but $\underline{D}(\rho_n \fatpipe \sigma_n) \simeq 0$, assuming the existence of quantum-hard one-to-one one-way functions. The states have polynomial complexity.
			\item $\smash{\underline{D}^{\mathrm{qu}}(\rho_n \fatpipe \sigma_n) \gtrsim \omega(\log n)}$, but $\underline{D}^{\mathrm{cl}}(\rho_n \fatpipe \sigma_n) \simeq 0$, assuming the existence of classically-hard quantumly-easy one-to-one one-way functions. The superscripts $\mathrm{cl}$ and $\mathrm{qu}$ indicate classical and quantum tests, respectively. The states have polynomial complexity.
		\end{enumerate}
	\end{mttheorem}
	
	\subsection{Computationally measured divergences and computational Stein's lemma}
	Measured quantum divergences are a widely used way of extending the reach of classical divergence measures to the quantum realm. These divergences arise by taking a classical divergence measure $\mathbf{D}(p \fatpipe q)$ and optimizing the classical divergence over the post-measurement states of a certain class of measurements applied to quantum states. 
	
	An extension of this concept to the computationally bounded realm is quite natural. We obtain it by following the same blueprint as in our definition of the computational relative entropy -- we first define a complexity-restricted single-shot measure and then apply polynomial regularization to obtain a computational quantity. 
	Given any classical divergence measure $\mathbf{D}(p\fatpipe q)$, we define the associated $G$-$K$-complexity measured divergence of two quantum states $\rho$ and $\sigma$ as (\cref{def:g_k_complexity_measured_quantum_divergence})
	\begin{align}
		&\mathbf{D}^{\bbC}(\rho \fatpipe \sigma;G ;K) \coloneqq \\
		&~~\max \big\{ \mathbf{D}(\calM[\rho] \fatpipe \calM[\sigma]) \pipe C(\calM) \leq G, \calM \text{ is $K$-outcome}\big\}.\nonumber
	\end{align}
	Here, the superscript $\bbC$ indicates \enquote{complexity}. We emphasize that the restriction on the number of outcomes $K$ is equally important as the restriction on the gate complexity of the measurement $\calM$.
	
	Polynomial regularization then allows us to define a computational quantity. In our case, we apply it to hypothesis testing which is why we are focusing on two-outcome measurements. We obtain the \emph{computationally two-outcome measured quantum divergence} associated to $\mathbf{D}(p\fatpipe q)$ as follows. 
	
	\begin{mtdefinition}[\cref{def:computationally_measured_two_outcome_quantum_divergence}]
		Consider two quantum states $\rho_n, \sigma_n \in \calH_n$. Then, their \emph{computationally two-outcome measured quantum divergence} is defined as
		\begin{align}
			\underline{\mathbf{D}}^{\mathbbm{2}}(\rho_n \fatpipe \sigma_n) &\colonsimeq \flim_{\ell \to \infty}\fliminf_{k\to\infty} \frac{1}{n^k} \mathbf{D}^{\bbC}(\rho_n^{\otimes n^k}\fatpipe \sigma_n^{\otimes n^k}; n^{k\ell}; 2).
		\end{align}
	\end{mtdefinition}
	
	We instantiate the above construction with the relative entropy
	$D(p \fatpipe q) = \Tr[ p (\log p - \log q)]$
	and the Rényi relative entropies
	$D_{\alpha}(p\fatpipe q) = \frac{1}{\alpha -1} \log \Tr[p^{\alpha} q^{1-\alpha}]$ to obtain the quantities $\underline{D}^{\mathbbm{2}}$ and $\underline{D}_\alpha^{\mathbbm{2}}$, respectively.
	
	\paragraph*{Computational Stein's Lemma.}
	Stein's Lemma is the fundamental result in unbounded information theory that allows us to define the relative entropy by regularizing the hypothesis testing relative entropy. 
	With both the computational relative entropy and the computationally measured entropies introduced, we can now connect the two by a computational analog of Stein's Lemma that holds under a mild condition on the measured max-relative entropy.
	\begin{mttheorem}[\cref{theorem:computational_steins_lemma}]
		Consider two quantum states $\rho_n, \sigma_n \in \calH_n$ such that $\underline{D}_{\max}^{\mathbbm{2}}(\rho_n \fatpipe \sigma_n) \leq O(\poly(n))$. Then, 
		\begin{align}
			\underline{D}(\rho_n \fatpipe \sigma_n) \simeq \underline{D}^{\mathbbm{2}}(\rho_n \fatpipe \sigma_n).
		\end{align}
	\end{mttheorem}
	This result gives us a better handle on the computational relative entropy, as it equates it to an optimized information theoretic quantity. The analog of the above result in unbounded quantum information theory is the equality of the regularized hypothesis testing relative entropy to the regularized measured relative entropy, which yields the Umegaki relative entropy as the true quantum generalization of relative entropy.
	
	We furthermore give an analogue of the upper and lower bounds to relative entropy that are given by the Rényi relative entropies (\cref{theorem:renyi_upper_and_lower_bounds_computational_relative_entropy}). Formally, we have for all $0 < \alpha_- < 1 < \alpha_+$
	\begin{align}
		\underline{D}_{\alpha_-}^{\mathbbm{2}}(\rho_n\fatpipe \sigma_n) \lesssim \underline{D}(\rho_n \fatpipe \sigma_n) \lesssim \underline{D}_{\alpha_+}^{\mathbbm{2}}(\rho_n\fatpipe \sigma_n).
	\end{align}
	We stress that we derive this relation without imposing any regularity assumptions on the involved states.

	\section{Applications}
	The computational relative entropy, the other quantities defined above and the supporting result that link them form the basis of a computational information theory that achieves a good level of abstraction while assuring that computational restrictions are appropriately taken into account. In this section, we show that these definitions are actually useful as they can be utilized to quantify information processing tasks in the complexity-limited regime.
	
	\subsection{Computational entropy and compression}
	We {\em derive} a quantum version of a computational entropy from computational relative entropy. For any quantum state $\rho_n \in \calH_n$, its \emph{computational entropy} is defined as (\cref{definition:computational_entropy})
	\begin{align}
		\underline{S}(\rho_n) \colonsimeq -\underline{D}(\rho_n \fatpipe \bbI_n)
	\end{align}
	Computational entropy inherits many properties from the computational relative entropy (\cref{proposition:properties_of_computational_computational_entropy}), most prominently computational smoothing $\underline{S}(\rho_n) \simeq \sup_{\tilde\rho_n \approx_c \rho_n} \underline{S}(\tilde\rho_n)$ (\cref{item:computational_entropy_computational_smoothing}). This enables the proof of separations for compression of quantum states. 
	
	Importantly, we show in \cref{theorem:computational_entropy_as_compression_rate} that the computational entropy has a comparable interpretation to the unbounded entropy, as it quantifies the best achievable compression rate when restricted to polynomially many gates and copies of the state to be compressed. We take this as evidence that the computational relative entropy can -- similar to its unbounded counterpart -- to derive other computational information-theoretic quantities that have operational meaning.
	
	Computational analogs of entropy have been studied in the classical and quantum literature, with the most popular approaches being the HILL  entropies~\cite{hastad_pseudorandom_1999} and the Yao entropy~\cite{10.5555/1382436.1382790}. We note that these notions are usually single-shot. As the Yao entropy is based on a notion of incompressibility under computational restrictions, our result that the computational entropy quantifies the best achievable compression means that the two quantities are closely related. 
	
	
	We furthermore define entropies that derive from the computational two-outcome measured divergences as
	\begin{align}
		\underline{\mathbf{S}}^{\mathbbm{2}}(\rho_n) \colonsimeq -\underline{\mathbf{D}}^{\mathbbm{2}}(\rho_n \fatpipe \bbI).
	\end{align}
	Instantiating this definition with the relative entropy and the Rényi-relative entropies gives rise to notions of computational entropy from a measured perspective. Importantly, we can derive from the Quantum Stein's Lemma that they relate to the computational entropy in the expected way. I.e., for $0 < \alpha_- < 1 < \alpha_+$ we have that
	\begin{align}
		\underline{S}_{\alpha_-}^{\mathbbm{2}}(\rho_n) \gtrsim \underline{S}^{\mathbbm{2}}(\rho_n) \simeq \underline{S}(\rho_n) \gtrsim \underline{S}_{\alpha_+}^{\mathbbm{2}}(\rho_n),
	\end{align}
	where the middle equality holds under the technical requirement that $\log \dim(\rho_n) \leq O(\poly(n))$ (i.e., we have polynomially many qubits).
	
	\subsection{Computational entanglement theory}
	
	Entanglement, as a purely quantum phenomenon, serves as a great proving ground for computational quantum information theory. We build on a definition of computational distillable entanglement $\underline{E}_D(\rho_n^{A_n B_n})$ ({\cref{definition:computational_distillable_entanglement_rate}}) in line with previous works~\cite{arnon2023computationalentanglementtheory,leone2025entanglementtheorylimitedcomputational} which quantifies the rate at which shared entanglement can be distilled from a given bipartite state under computationally efficient LOCC operations. 
	
	We relate the tasks of entanglement manipulation back to our central definitions of computational relative entropy and computational entropy.
	We show that this quantity cannot exceed a computational version of the well-known Rains bound~\cite{rains1999bound} that derives from the computational relative entropy.
	\begin{mttheorem}[\cref{corollary:computational_rains_bound}]
		Consider a bipartite quantum state $\rho_n^{A_n B_n} \in \calH_n \otimes \calK_n$ such that $\log \dim(\calH_n), \log \dim(\calK_n) \leq O (\poly(n))$. Then, the \emph{Rains bound}
		\begin{align}
			&\underline{R}(\rho_n^{A_n B_n}) \colonsimeq \\&\inf_{} \big\{  \underline{D}(\rho_n^{A_n B_n} \fatpipe \sigma_n^{A_n B_n}) \pipe \sigma_n^{A_n B_n} \geq 0, \lVert (\sigma_n^{A_n B_n})^{T_{B_n}}\rVert_1\leq 1 \big\}.\nonumber
		\end{align}
		is a limit on the computational distillable entanglement
		\begin{align}
			\underline{E}_D(\rho_n^{A_n B_n}) \lesssim \underline{R}(\rho_n^{A_n B_n}).
		\end{align}
	\end{mttheorem}
	In fact, we derive a stronger result (\cref{theorem:converse_bound_on_distillable_entanglement}) that allows us to obtain bounds on computationally distillable entanglement involving a computationally smoothed log-negativity, the pseudo log-negativity ({\cref{corollary:pseudo_log_negativity_bound}}) and a computationally smoothed entanglement entropy, the pseudo entanglement-entropy (\cref{corollary:pseudo_entanglement_entropy_bound}). 
	
	The dual task to entanglement distillation is entanglement dilution, where we are tasked -- again under computationally efficient LOCC -- to create a bipartite state from shared entanglement. The computational entanglement cost $\underline{E}_C(\rho_n^{A_n B_n})$ (\cref{definition:computational_entanglement_cost_rate}) describes the optimal rate at which we consume entanglement in this process. We show that the computational entanglement cost is always larger or equal to the computational distillable entanglement (\cref{proposition:computational_entanglement_cost_larger_than_distillable_entanglement}). We furthermore establish an analog of the upper bound on entanglement cost given by the entanglement entropy (\cref{theorem:computational_entanglement_cost_via_compression})
	\begin{align}
		\underline{E}_C(\rho_n^{A_n B_n}) \lesssim \min \{ \underline{S}(\rho_n^{A_n}),\underline{S}(\rho_n^{B_n})\}.  
	\end{align}
	We emphasize that even for pure states, the minimum on the right hand side is necessary. This is because the computational entropy of subsystems of a bipartite pure state is not necessarily the same. 
	
	\section{Conclusions and outlook}
	
	In this work, we have developed a systematic framework for quantum information theory under computational constraints.
	Our central contribution is the computational relative entropy, defined through polynomial regularization of complexity-constrained hypothesis testing. This quantity serves as a mother quantity from which other computational information-theoretic measures naturally emerge. We have shown that it satisfies essential properties including boundedness, data processing, isometric invariance, and -- crucially -- computational smoothing. The latter property, which states that the computational relative entropy remains unchanged when replacing the first argument with a computationally indistinguishable state, connects our framework to the broader landscape of computational complexity and cryptography.
	
	The operational significance of our framework is demonstrated through several key results. The computational Stein's lemma establishes that the computational relative entropy equals the computational two-outcome measured relative entropy, providing a connection to the the information theoretic quantities arising in the unbounded setting. We have shown that the computational entropy precisely characterizes the optimal compression rate under complexity constraints, and derived computational versions of fundamental inequalities including Pinsker's inequality and the Bretagnolle-Huber inequality. In the realm of entanglement theory, we have established computational analogues of the Rains bound on distillable entanglement and related the entanglement cost to computational entropy.
	
	Perhaps most strikingly, our framework reveals fundamental separations between bounded and unbounded observers. Under standard cryptographic assumptions, we have shown that states can have infinite unbounded relative entropy while having zero computational relative entropy. Similarly, states can be incompressible to computationally bounded observers despite having negligible unbounded entropy. These separations underscore that computational constraints fundamentally alter the landscape of quantum information theory.
	
	This work establishes a novel approach to computational quantum information theory that recovers the look and feel of asymptotic information theory while still ensuring that computational restrictions are taken into account. It opens up numerous directions of research in information theory and physics alike, given that any theoretical treatment resting on information-theoretic quantities fundamentally disregards the computationally bounded nature of experimenters.
	
	\paragraph*{Computational thermodynamics.}
	Classical and quantum thermodynamics are probably the physical theories most intimately tied to information theory~\cite{landi2021irreversible,munson2025complexity-constrained}. We expect many insights to arise from the study of thermodynamics with quantities derived in this work and novel ones based on the same blueprint.
	
	\paragraph*{Channel coding and communication.} 
	Developing computational versions of channel capacities and communication protocols is a natural next step. How do computational constraints affect quantum communication limits and capacity analogues?
	
	\paragraph*{Other derived quantities.} Beyond computational entropy, many other important quantities like mutual information and conditional entropy await study, with similar operational interpretations expected.
	
	\paragraph*{Information-theoretic considerations.}
	Our computational Stein's Lemma works in the weak converse regime. Can we establish strong converses and computational Rényi divergence analogues from strong converse exponents in the computational realm?
	
	\paragraph*{Other notions of complexity.}
	It is an intriguing question to connect our definitions with recent work introducing the notion of unitary complexity~\cite{Metger}. Also considering computational models beyond polynomial time, like structured models (like $\mathsf{QAC^{0}}$) or depth-restricted models (like $\mathsf{BQNC}$) could yield interesting separations and applications across diverse complexity classes.
	
	\section{Acknowledgments}
	The authors would like to thank Milán Mosonyi for helpful input about the achievability of Rényi relative entropies with likelihood ratio tests. Further thanks goes to Ludovico Lami, Bartosz Regula, Mario Berta and Christoph Hirche for helpful discussions the Beyond IID conference 2025. We would also like to thank Greg White for insightful discussions about Choi states. 
	The authors wish to especially thank Rémi Ligez for pointing out the incorrectness of the mathematical justification of our construction in an earlier version of this manuscript.
	This project has been funded by the 
	BMFTR (QR.N, DAQC, QSol, MuniQC-Atoms, Hybrid++), 
	the BMWK (EniQmA), the DFG (SPP 2514),
	the Munich Quantum Valley (K-8),
	the Clusters of Excellence (ML4Q and MATH+),
	Berlin Quantum, the QuantERA (HQCC),
	the Quantum Flagship (Millenion and Pasquans2),
	and the European Research Council 
	(DebuQC).

	\bibliography{main}
	
	\let\addcontentsline=\oldaddcontentsline
	
	\clearpage 
	\appendix
	
	\title{Supplementary Material: Computational Relative Entropy}
	\maketitle

	\onecolumngrid
	\tableofcontents
	\clearpage 
	
	{   \renewcommand{\arraystretch}{1.3}
		\centering
		\begin{tabular}{lp{.7\linewidth}}
			Symbol & Description \\\midrule
			$f\leq g, f = g$ & Point-wise order of functions (\cref{sec:limits_of_polynomial_resources})\\
			$f \leq^* g, f =^* g$ & Eventual domination order of functions (\cref{sec:limits_of_polynomial_resources})\\
			$f \lesssim g, f \simeq g$ & Order of functions up to negligible terms (\cref{sec:limits_of_polynomial_resources})\\
			$\flimsup, \fliminf, \flim$ & Order-convergence limits of functions (\cref{sec:limits_of_polynomial_resources}) \\
			$C(\cdot)$ & Gate complexity of a unitary, quantum channel, state or POVM effect (\cref{section:complexity_limited_quantum_information_processing}) \\
			$\calQ(\calH; G)$ & Set of POVM effects over Hilbert space $\calH$ with gate complexity at most $G$ (\cref{definition:complexity_limited_measurements}) \\
			$\calM(\calH; G; K)$ & Set of POVMs with $K$ effects over Hilbert space $\calH$ with gate complexity at most $G$ (\cref{definition:complexity_limited_measurements})\\
			$D(\rho \fatpipe \sigma)$ & Umegaki relative entropy  \\
			$V(\rho \fatpipe \sigma)$ & Relative entropy variance  \\
			$D_{\alpha}(p \fatpipe q)$ & Classical $\alpha$-Rényi relative entropy \\
			$D_h^{\epsilon}(\rho \fatpipe \sigma)$ & Hypothesis testing relative entropy \\
			$D_h^{\epsilon}(\rho \fatpipe \sigma; G)$ & $G$-complexity hypothesis testing 
			relative entropy (\cref{def:g_complexity_hypothesis_testing_relative_entropy}) \\
			$\underline{D}(\rho_n \fatpipe \sigma_n)$ & Computational relative entropy (\cref{def:computational_relative_entropy}) \\
			$d_{\mathrm{TV}}(\rho, \sigma)$ & Quantum total variation distance \\
			$d_{\mathrm{TV}}(\rho, \sigma;G)$ & $G$-complexity quantum total variation distance (\cref{def:g_complexity_symmetric_hypothesis_testing_error}) \\
			$\underline{d}_{\mathrm{TV}}(\rho_n \fatpipe \sigma_n)$ & Computational total variation distance (\cref{definition:computational_total_variation_distance}) \\
			$\rho_n \approx_c \tilde\rho_n$ & The states $\rho_n$ and $\tilde\rho_n$ are computationally indistinguishable (\cref{def:computational_indistinguishability}) \\
			$\mathbf{D}^{\bbC}(\rho\fatpipe \sigma ; G; K)$ & $G$-$K$-complexity measured quantum divergence for a generic classical base divergence $\mathbf{D}(p \fatpipe q)$ (\cref{def:g_k_complexity_measured_quantum_divergence}) \\      
			$\underline{\mathbf{D}}^{\mathbbm{2}}(\rho_n \fatpipe \sigma_n)$ & Computational two-outcome measured quantum divergence for a generic classical base divergence $\mathbf{D}(p \fatpipe q)$ (\cref{def:computationally_measured_two_outcome_quantum_divergence}) \\      
			$\underline{{D}}^{\mathbbm{2}}(\rho_n \fatpipe \sigma_n)$ & Computational two-outcome measured relative entropy (Instantiation of \cref{def:computationally_measured_two_outcome_quantum_divergence} with $\mathbf{D} = D$) \\
			$\underline{{D}}_{\alpha}^{\mathbbm{2}}(\rho_n \fatpipe \sigma_n)$ & Computational two-outcome measured $\alpha$ Rényi-relative entropy (Instantiation of \cref{def:computationally_measured_two_outcome_quantum_divergence} with $\mathbf{D} = D_\alpha$) \\
			$\underline{S}(\rho_n)$ & Computational entropy (\cref{definition:computational_entropy}) \\
			$\mathbf{S}^{\bbC}(\rho;G;K)$ & $G$-$K$-complexity measured quantum entropy for classical base divergence $\mathbf{D}$ (\cref{def:g_k_complexity_measured_quantum_entropy}) \\
			$\underline{\mathbf{S}}^{\mathbbm{2}}(\rho_n)$ & Computational two-outcome measured entropy for classical base divergence $\mathbf{D}$ (\cref{def:computational_two_outcoe_measured_entropy}) \\
			$\underline{S}(\rho_n)$ & Computational two-outcome measured entropy (Instantiation of \cref{def:computational_two_outcoe_measured_entropy} for $\mathbf{D} = D$) \\
			$\underline{S}_{\alpha}(\rho_n)$ & Computational two-outcome measured $\alpha$-Rényi entropy (Instantiation of \cref{def:computational_two_outcoe_measured_entropy} for $\mathbf{D} = D_\alpha$) \\
			$E_{D,O}^{\epsilon}(\rho^{AB}; G)$ & $G$-complexity $\epsilon$-distillable entanglement relative to a class of operations $O$ like LOCC or PPT (\cref{def:g_complexity_distillable_entanglement}) \\
			$\underline{E}_{D,O}(\rho_n^{A_n B_n})$ & Computational distillable entanglement rate (\cref{definition:computational_distillable_entanglement_rate}) \\
			$E_{C,O}^{\epsilon}(\rho^{AB}; G)$ & $G$-complexity $\epsilon$-entanglement cost relative to a class of operations $O$ like LOCC or PPT (\cref{def:g_complexity_entanglement_cost}) \\
			$\underline{E}_{C,O}(\rho_n^{A_n B_n})$ & Computational entanglement cost rate (\cref{definition:computational_entanglement_cost_rate})
		\end{tabular}
	}
	
	\clearpage
	
	\section{Limits of polynomial resources}\label{sec:limits_of_polynomial_resources}
	
	As complexity theory typically makes asymptotic statements, we need a scaling variable. 
	As such, this work considers hypothesis testing problems between states $\rho_n$ and $\sigma_n$ living in a Hilbert space $\calH_n$ that are understood to be elements of sequences indexed by $n \in \bbN$. In the rest of this work, all objects indexed by $n$ are understood to be elements of such a sequence.

	It is our aim to capture the ultimate limits of what can be done with \emph{polynomial} resources, both considering the number of copies as well as the number of available gates. To do so, we let the resource in question scale as $n^k$ and then take the limit $k \to \infty$, echoing the definition of the complexity class $\mathsf{P}$ as the union $\mathsf{P} = \bigcup_{k\in\bbN} \mathsf{DTIME}(n^k)$.
	If we would understand these limits point-wise, i.e.\ for every fixed $n$ separately, we would not gain any novel insights as $\lim_{k\to\infty}n^k = \infty$ and we would hence go back to unbounded resources. In this section, we introduce the mathematics necessary to make sense of such limits while only caring about the asymptotic scaling. We will make use of the formalism of \emph{order convergence in complete lattices} to rigorously define a limes inferior and superior, which in turn allows us to define a limit when they both agree. Ref.~\cite{williams2025survey} has a good and compact exposition of the mathematics we are going to use, even if it treats an adjacent topic, a book reference is found in Ref.~\cite{roman2008lattices}.

	Let us introduce the necessary definitions.
	A \emph{partially ordered set}, or \emph{poset}, $(P,\leq)$ is a set $P$ with a partial order $\leq$. A partial order is reflexive, $p \leq p$ for all $p\in P$, antisymmetric, $p\leq q$ and $q \leq p$ imply $p = q$, and transitive, $p \leq q$ and $q \leq r$ imply $p \leq r$.
	To every subset $Q \subseteq P$ of a poset, we can assign a set of lower bounds $L(Q)$ and upper bounds $U(Q)$
	\begin{align}
		L(Q) &\coloneqq \{ p \in P \pipe p \leq q \text{ for all } q \in Q  \} ,\\
		U(Q) &\coloneqq \{ p \in P \pipe p \geq q \text{ for all } q \in Q  \}.
	\end{align}
	If the set of lower bounds $L(Q)$ contains a maximal element, i.e. an $l_{\max} \in L(Q)$ such that $l \leq l_{\max}$ for all $l\in L(Q)$, then this element is the \emph{infimum} or \emph{meet} of $Q$ denoted as $\inf Q$. The \emph{supremum} or \emph{join} is defined analogously.

	A \emph{lattice} is a poset $(P, \leq)$ such that any pair of elements $p, q \in P$ has an infimum and a supremum in $P$. A lattice $(P, \leq)$ is \emph{complete} if every subset $Q \subseteq L$ has an infimum and a supremum in $P$. The fact that every subset has a well-defined infimum and supremum means that every sequence $\{ p_k \}_{k\in \bbN} \subseteq P$ in a complete lattice is guaranteed to have a well-defined limes inferior and limes superior 
	\begin{align}
		\liminf_{k\to\infty} p_k &\coloneqq \sup_{k \in \bbN} \inf_{k' \geq k} p_{k'} ,\\
		\limsup_{k\to\infty} p_k &\coloneqq \inf_{k \in \bbN} \sup_{k' \geq k} p_{k'}
	\end{align}
	that exist in $P$. If these two limits agree, then we write
	\begin{align}
		\liminf_{k\to\infty} p_k = \limsup_{k\to\infty} p_k =\lim_{k\to\infty} p_k.
	\end{align}
	Note that the above definition of limit relates to the notion of \emph{order convergence} and not the usual definition in terms of metrics.
	Every complete lattice $P$ is bounded, as their are always a minimum element $\inf P$ and a maximum element $\sup P$. Hence, any sequence that is eventually monotonically increasing (decreasing) with respect to the partial order of the lattice converges to its supremum (infimum).

	Let us now turn to object of our study. We will always consider functions $f$ from $\bbN$ to $[0,\infty]$. On an intuitive level this could be the performance of a protocol when executed on the state $\rho_n$. The set of such functions is denoted as $[0,\infty]^\bbN$. This set of functions has a natural partial order given by point-wise comparison. Let $f,g\in [0,\infty]^\bbN$. Then,
	\begin{align}
		f \leq g \ \Leftrightarrow \ f(n) \leq g(n) \text{ for all } n \in \bbN.
	\end{align}
	The set $([0,\infty]^\bbN, \leq)$ is a complete lattice. It is, however, unsuitable for our applications because of the issue we touched upon earlier: the limit we obtain from the lattice definition is the function of point-wise limits which recover the limit of unbounded resources.

	To remedy this fact, we will make use of the \emph{eventual domination order}. Let again $f,g \in [0,\infty]^\bbN$. Then,
	\begin{align}
		f \leq^* g \ &\Leftrightarrow \ f(n) \leq g(n) \text{ for all but finitely many } n \in \bbN \\
		&\Leftrightarrow \ \text{there exists an } n_0 \in \bbN \text{ such that }f(n) \leq g(n) \text{ for all } n \geq n_0.
	\end{align}
	The authors of Ref.~\cite{arnon2023computationalentanglementtheory} used the symbol $\leq_\infty$ instead of $\leq^*$. We chose to go with the latter notation more commonly found in the literature. The eventual domination order nicely captures the asymptotic behaviour of functions, because it implicitly always sends the argument to infinity first before comparing. Consider for example the function $n^{-k}$. We have that for all $k$, $n^{-k} \geq^* e^{-n}$, but $n^{-k} \not\geq e^{-n}$ point-wise for some $k$.

	We wish to obtain a complete lattice that captures the eventual domination order. Our first obstacle on this way is the fact that the eventual domination order is not a proper partial order, as it lacks the antisymmetry property. This means that there exists distinct functions $f,g \in [0, \infty]^\bbN$ such that $f \leq^* g$ and $g \leq^* f$ but $f \neq g$. This is simply due to the fact that the eventual domination order does not \enquote{see} differences between functions on finitely many $n\in \bbN$. This problem is easily remedied by declaring such functions \emph{eventually equivalent}, $f =^* g$, and modding out this equivalence relation. We denote the thus-obtained equivalence classes as
	\begin{align}
		[f]^* \coloneqq \big\{ g \in [0,\infty]^\bbN \pipe g =^* f \big\}
	\end{align}
	and the obtained set as $[0,\infty]^{\bbN*} \coloneqq [0,\infty]^\bbN / =^*$. Because the eventual domination order is antisymmetric on the equivalence classes, we have obtained a partially ordered set $([0,\infty]^{\bbN*}, \leq^*)$.

	The next problem we have is that the partially ordered set $([0,\infty]^{\bbN*}, \leq^*)$ is not a complete lattice. An example of a sequence that has no limit inside of it is the sequence of inverse polynomials $\{ n^{-k} \}_{k\in \bbN}$. If there would exists, say the limes inferior of this sequence, there would be a function $f$ we could write down such that
	\begin{align}
		f \leq^* n^{-k} \text{ for all } k \in\bbN
	\end{align}
	but also $f \geq^* g$ for all functions that have this property. We call such functions that vanish faster than any polynomial \emph{negligible} and the existence of this limit would require the existence of a \emph{largest} negligible function. This is, however, impossible because if $f$ were this function, then $2 f$ would be negligible as well but $f \leq^* 2 f$, giving a clear contradiction.

	We overcome this issue by considering the \emph{Dedekind-MacNeille completion} of the partially ordered set $([0,\infty]^{\bbN*}, \leq^*)$. 
	It is a construction that yields the (in a meaningful sense) smallest complete lattice that into which a given partially ordered set can be \emph{order embedded}. An map $E$ between partially ordered sets $(P,\leq)$ and $(T, \preccurlyeq)$ is an order embedding if $p \leq q \ \Leftrightarrow E(p) \preccurlyeq E(q)$ for all $p, q\in P$.
	\begin{definition}[Dedekind-MacNeille completion]\label{def:dedekind_macneille}
		For a given partially ordered set $(P, \leq)$, the set
		\begin{align}
			\operatorname{DM}(P, \leq) \coloneqq \big\{ Q \subseteq P \pipe L(U(Q)) = Q\big\}
		\end{align}
		combined with set inclusion, $\subseteq$, forms a complete lattice $(\operatorname{DM}(P, \leq),\subseteq)$ referred to as the \emph{Dedekind-MacNeille completion} of $(P, \leq)$. The operations of infimum and supremum are given by set intersection $\cap$ and set union $\cup$. 
		The partially ordered set $(P, \leq)$ order embeds into its completion via $p \mapsto L(p)$ for all $p\in P$.
	\end{definition}
	The DM completion $(\operatorname{DM}(P, \leq),\subseteq)$ has a range of nice and important properties~\cite{roman2008lattices}: it is the smallest complete lattice that contains $(P, \leq)$, the embedding is order-preserving, i.e.\ $p\leq q$ implies $L(p) \subseteq L(q)$ and the embedding preserves existing infima and suprema, and hence also existing limits. To consider a canonical example, the extended reals $\bbR \cup \{-\infty,\infty\}$ are order-isomorphic to the DM completion of the rational numbers $\bbQ$.

	Our strategy is now to consider limits in the DM completion of $([0,\infty]^{\bbN*},\leq^*)$. To this end, we have to understand what the elements of the DM completion actually represent. We can see the elements of $([0,\infty]^{\bbN*},\leq^*)$, which are the equivalence classes $[f]^*$ under eventual equivalence, as elements of the DM completion by considering their embedding
	\begin{align}
		L([f]^*) = \big\{ [g]^* \in [0,\infty]^{\bbN*} \pipe [g]^* \leq^* [f]^*\big\},
	\end{align}
	i.e.\ an equivalence class $[f]^*$ is represented by the set of all functions eventually dominated by $f$. As such, the elements of the DM completion of $([0,\infty]^{\bbN*},\leq^*)$ represent types of \emph{asymptotic scalings}. In the case of functions, this is very close to seeing $o(f)$ as a set, with the difference that $f$ is not $o(f)$ but it is contained in $L([f]^*)$.

	Elements of the DM completion of $([0,\infty]^{\bbN*},\leq^*)$ that are not obtained by embedding an element of $[0,\infty]^{\bbN*}$ represent more general asymptotic scalings. A great example is the limit of inverse polynomials we already saw. Consider again the sequence $\{ n^{-k} \}_{k\in\bbN}$. Then, as it is monotonically decreasing with respect to $\leq^*$, it converges to its infimum
	\begin{align}
		\lim_{k\to\infty} L([n^{-k}]) = \inf_{k\in \bbN} L([n^{-k}]) = \bigcap_{k \in \bbN} L([n^{-k}]).
	\end{align}
	Hence, a function $f \in \flim_{k\to\infty} L([n^{-k}])$ if it vanishes faster than any polynomial, i.e.\
	\begin{align}
		\lim_{k\to\infty} L([n^{-k}]) = \negl(n),
	\end{align}
	where $\negl(n)$ denotes the set of negligible functions in $[0,\infty]^{\bbN}$. This is a great example of how the DM completion captures more general notions of asymptotic scaling -- indeed, we cannot define the concept of negligibility while only considering a single function as a reference. A similar argument shows that
	\begin{align}
		\lim_{k\to\infty} L([n^k]) = \sup_{k\in\bbN} L([n^k]) = \bigcup_{k \in \bbN} L([n^k]) = \poly(n). 
	\end{align}

	Having understood the elements of the DM completion, we now turn to definition of limits we use in this work.
	\begin{definition}[Limits]\label{def:limits}
		Consider a sequence of functions $\{ f_k \}_{k\in\bbN} \subseteq [0,\infty]^{\bbN}$. We define
		\begin{align}
			\flimsup_{k\to\infty} f_k &\coloneqq \limsup_{k\to\infty} L([f_k]^*) \\
			\fliminf_{k\to\infty} f_k &\coloneqq \liminf_{k\to\infty} L([f_k]^*),
		\end{align}
		i.e.\ we take limits in the DM completion of $([0,\infty]^{\bbN*}, \leq^*)$ where they are guaranteed to be well-defined. If those two limits agree, we simply write
		\begin{align}
			\flim_{k\to\infty} f_k.
		\end{align}
	\end{definition}
	A crucial property of the above construction of limits follows from the order-preserving nature of the DM completion. If we have a point-wise (or eventual) upper bound $g \geq^* f_k$ that holds as $k\to\infty$, then clearly
	\begin{align}
		\flimsup_{k\to\infty} f_k \subseteq L([g]^*).
	\end{align}
	In such a case, we will simply write
	\begin{align}
		\flimsup_{k\to\infty} f_k \leq^* g,
	\end{align}
	and analogously for lower bounds and the limes inferior. This notation exploits the order-preserving nature of the DM completion to compare elements of two different sets. Practically speaking, the DM completion fills the gaps in $([0,\infty]^{\bbN*}, \leq^*)$ to make sure that all limits exist, but it is nevertheless sensible to directly speak about bounds that hold inside of $([0,\infty]^{\bbN*}, \leq^*)$ already even though the limit itself is thought of as an element of the DM completion.
	This notation allows us to pin down computational quantities rigorously by upper- and lower-bounding them, even if they technically only exist as elements of the DM completion.

	Having established our notion of limits for computational quantities, we are left with a final bit of notation. In our work, we take the cryptographic attitude and only care about quantities up to \emph{negligible} terms. We formalize this by introducing yet another equivalence relation:
	\begin{align}
		f \simeq g \ \Leftrightarrow |f-g|\in \negl(n).
	\end{align}
	In words, if two functions $f$ and $g$ only differ by a negligible term, we write $f\simeq g$. We can associate to this an ordering relation
	\begin{align}
		f \lesssim g \ \Leftrightarrow \ \text{exists } g' \simeq g \text{ such that } f \leq g' \text{ point-wise.}
	\end{align}
	We note that point-wise ordering $\leq$, eventual domination ordering $\leq^*$ and ordering up to negligible functions $\lesssim$ are in descending order of strictness.
	
	\begin{remark}[Notation]\label{remark:notation}
		In this work, we use the notion of limits introduced above and will only use the ordering and equality up to negligible functions $(\simeq, \lesssim)$ to talk about computational quantities, even if the stricter eventual domination order $(=^*, \leq^*)$ or point-wise order $(=,\leq)$ could also be used. To underscore this fact, we will use the symbol $:\simeq$ for definitions.
	\end{remark}
	The choice of notation allows us to present our results in a streamlined way and to make full use of the regularization over polynomially many copies. Especially, the constructions presented in this section allow us to make sense of the following crucial lemma.
	\begin{lemma}[Limit of inverse polynomials is negligible]\label{lemma:limit_of_inverse_polynomials_is_negligible}
		Consider the sequence $\{ n^{-k} \}_{k\in \bbN} \subseteq [0,\infty]^{\bbN}$. Then,
		\begin{align}
			\flim_{k\to\infty} n^{-k} \simeq 0.
		\end{align}
	\end{lemma}
	\begin{proof}
		The limit is taken in the DM completion and we have already argued above that
		\begin{align}
			\flim_{k\to\infty} n^{-k} = \negl(n). 
		\end{align}
		As $\negl(n) = \{ f \in [0,\infty]^{\bbN} \pipe f \simeq 0\}$, the lemma follows.
	\end{proof}

	\section{Complexity-limited quantum information processing}\label{section:complexity_limited_quantum_information_processing}
	
	\subsection*{Gate complexity}
	The aim of this work is to understand how complexity limitations influence our ability to perform certain quantum information processing tasks. There are different ways of measuring complexity, the \emph{time complexity} is surely the most standard one. To measure \emph{how long} a computation takes we usually resort to the \emph{gate complexity} of the shortest circuit that implements the computation, with the assumption that executing gates takes unit time. Specifically, in the quantum context it captures how many quantum gates from a certain native gate set are necessary to construct a quantum operation. This is operationally meaningful and gives a proxy for the time complexity of executing the unitary in question and is also the canonical way of quantifying complexity in the quantum regime~\cite{bernstein1997quantum,brandao2021models}.
	In the following, we can think of a quantum system built from many qubits. It could also be built from other, higher-dimensional constituents, but we leave this arbitrary for now.
	\begin{definition}[Gate complexity of a unitary]\label{def:unitary_gate_complexity}
		Let $U$ be a unitary transformation. Then, its \emph{gate complexity} $C(U)$ is the minimal number of arbitrary 2-local unitaries necessary to realize $U$. 
	\end{definition}
	The above definition reflects the idea that any operation we want to perform on an actual quantum device has to be built up from elementary operations supported by the device. Usually these allow single-site operations and some native coupling between neighboring qubits. As we ultimately wish to go to polynomial complexity, we can well justify to have access to arbitrary 2-local unitaries by invoking the Solovay-Kitaev theorem~\cite{NielsenChuang}, which allows us to implement a constant locality unitary to exponential precision in polynomial depth from any universal gate set. We also did not include connectivity constraints, as the overhead for connecting two qubits in the system is at most on the order of the system size, which is necessarily polynomial, via a SWAP chain.
	
	When discussing the gate complexity of a unitary, it makes sense to also discuss the complexity of realizing it up to a small error. 
	We will not need this definition for our purposes, as we use a bottom-up approach where all our circuits are necessarily built from a finite number of gates and we care about all realizable unitaries in this setting.
	We additionally note that the choice of 2-local unitaries as the basic gate set is somewhat arbitrary, as the Solovay-Kitaev theorem~\cite{DawsonNielsen2006} guarantees that we can approximate any circuit to exponential precision with polynomial overhead over any universal gate set. 
	
	The notion of gate complexity also extends naturally to quantum channels by considering their corresponding Stinespring dilation.
	
	\begin{definition}[Gate complexity of a channel]\label{def:channel_gate_complexity}
		Let $\Phi\colon \calH \to \calK$ be a quantum channel. Then, its \emph{gate complexity} $C(\Phi)$ is the minimal number of arbitrary 2-local unitaries necessary to realize a unitary $U$ such that
		\begin{align}
			\Phi[X] = \Tr_{\mathrm{anc'}}[U ( |0\rangle\!\langle 0|_{\mathrm{anc}} \otimes X)U^{\dagger}]
		\end{align}
		for an auxiliary system of arbitrary size in the all-zero state. The partial trace is over as many auxiliary degrees of freedom as necessary to realize the quantum channel. 
	\end{definition}
	It is instructive to visualize the construction of the channel as a tensor diagram.
	\begin{center}
		\includegraphics[width=.45\columnwidth]{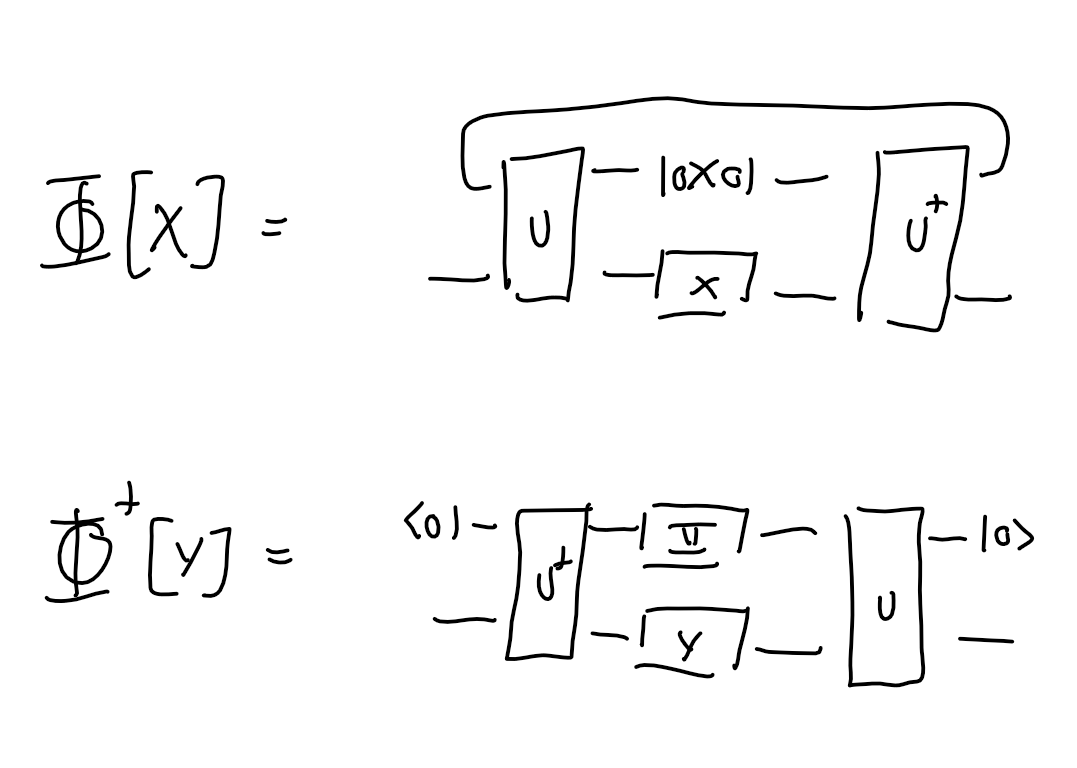}
	\end{center}
	We define the gate complexity of a quantum state along similar lines as the minimum complexity of a quantum channel that prepares the state. 
	\begin{definition}[Gate complexity of a quantum state]
		Let $\rho \in \calH$ be a quantum state. Then, its gate complexity is defined as the smallest complexity of a quantum channel that prepares $\rho$:
		\begin{align}
			C(\rho) \coloneqq \min \big\{ C(\Phi) \pipe \Phi[1] = \rho \big\}.
		\end{align}
		The input $1$ is 
		the trivial 1-dimensional state.
	\end{definition}

	\subsection*{Complexity-limited measurements}
	To formalize hypothesis testing under complexity limitations, we consider strategies whose implementation uses a bounded number of gates. We follow Ref.~\cite{munson2025complexity-constrained} and define the following set of \emph{positive operator valued measure} (POVM) 
	effects acting on a Hilbert space $\calH_n$ that exploit an a-priori arbitrarily large ancillary space.	
	\begin{definition}[Complexity-limited measurements]\label{definition:complexity_limited_measurements}
		Consider a Hilbert space $\calH$. We define the set of $G$-complexity POVM effects as
		\begin{align}
			\calQ(\calH; G) \coloneqq \left\{\left. \langle 0 |_{\mathrm{anc}} U \big[ \bigotimes_i A_i \big] U^{\dagger} |0\rangle_{\mathrm{anc}} \, \right| \, C(U) \leq G, A_i \in \big\{ |0\rangle\!\langle 0|, \bbI \big\} \right\}.
		\end{align}
		Similarly, we define the set of $G$-complexity $K$-outcome measurements as 
		\begin{align}
			\calM(\calH; G; K) \coloneqq \left\{\left. \Big\{ \Lambda_k = \langle 0 |_{\mathrm{anc}} U \big[ \bigotimes_i A^{(k)}_i \big] U^{\dagger} |0\rangle_{\mathrm{anc}} \Big\}_{k=1}^K \, \right| \, C(U) \leq G, A_i \in \big\{ |0\rangle\!\langle 0|, \bbI \big\}, \sum_{k=1}^K \big[ \bigotimes_i A^{(k)}_i \big] = \bbI \right\}
		\end{align}
		and assign to any $M \in \calM(\calH_n;G)$ a measurement channel
		\begin{align}
			X \mapsto M[X] \coloneqq \sum_k |k\rangle\!\langle k| \, \Tr[ \Lambda_k X].
		\end{align}
	\end{definition}
	In words, this amounts to all POVM effects that can be implemented by initializing an auxiliary system in the all-zero state, perform a unitary of gate complexity at most $G$ and then performing a suitable measurement in the computational basis. Note that the dimension of the auxiliary system is in principle unbounded, but we can only ever access $O(G)$ many of the auxiliary systems, as such the number of ancillas remains bounded at all times. It is understood in the above definition that $\calQ(\calH_n; \infty)$ is the set of all POVM effects and $\calM(\calH_n;\infty)$ the set of all measurements on $\calH_n$. We further stress that we could in principle allow arbitrary projectors for the operators $A_i$, the above form does, however, guarantee that our definition is symmetric in the sense that states should also be prepared from the all-zero state.
	
	We note that the above definition naturally encompasses the natural setup where first a unitary is applied, a measurement in the computational basis is performed and subsequently a classical function is evaluated that maps the outcome to either $0$ or $1$ representing the outcome of the measurement. This is the consequence of the following proposition which means that we can efficiently realize classical boolean functions in our setup:
	\begin{proposition}\label{prop:implementation_of_classical_boolean_function}
		Any classical boolean function $f\colon \{0, 1\}^k\rightarrow \{0, 1\}$ with classical gate complexity $C(f)\leq G$ can be implemented by a unitary quantum circuit with gate complexity $C(U) \leq G \cdot M \cdot 2^{k+2}\cdot k'/k$ complemented with at most $C(U)$ many auxiliary qubits, where $k$ and $k'$ is the uniform upper bound on the FANIN and FANOUT (respectively) for each of the $G$ many gates and $M$ is some universal constant.
	\end{proposition}
	\begin{proof}
		Let us assume that each of the $G$ many gates has a uniformly bounded FANIN of $k$ and a uniformly bounded FANOUT $k'$, where $k$ is not necessarily the same as $k$. This means that each gate can be thought of as computing a function $f: \{0, 1\}^k\rightarrow \{0, 1\}^{k'}$. We can split this into $k'$ many Boolean functions $f: \{0, 1\}^k\rightarrow \{0, 1\}$. Due to Lupanov \cite{lupanov1958method}, we know that any Boolean function from $k$ bits to 1 bit can be evaluated using a de Morgan circuit of size at most $M \cdot 2^k/k $, where $M$ is a universal constant and de Morgan circuit is composed only of $\mathsf{AND}$, $\mathsf{NOT}$, and $\mathsf{OR}$ gates with FANIN being 2 for $\mathsf{AND}$ and $\mathsf{OR}$ gates. Thus, we can evaluate any such classical circuit with $G$ gates by using a de Morgan circuit of size at most $G \cdot M \cdot 2^k\cdot k'/k$. Furthermore, we can write an $\mathsf{OR}$ gate using three $\mathsf{NOT}$ gates and one $\mathsf{AND}$ gate since $x \vee y = \neg (\neg x \wedge \neg y)$. Now we are allowed to only use $\mathsf{AND}$ and $\mathsf{NOT}$ gates at the expense of the circuit size now being $G \cdot M \cdot 2^{k+2}\cdot k'/k$. We can realize a $\mathsf{NOT}$ by a Pauli-X gate. To realize the two-bit $\mathsf{AND}$ gate, we will use a 3-qubit Toffoli. Since $\mathsf{Toffoli}(x, y, z) \rightarrow (x, y, z\oplus (x \wedge y))$, we can write $\mathsf{AND}(x, y) = \mathsf{Toffoli}(x, y, 0)$. Note that this procedure takes one auxiliary qubit per Toffoli gate. Thus, auxiliary qubits as much as the total gates suffice for implementing any classical computation reversibly.
	\end{proof}
	
	We will make recurrent use of the following straightforward lemma.
	\begin{lemma}[Absorbing processing into measurement sets]\label{lemma:absorbing_channels_into_measurement_sets}
		Consider a quantum channel $\Phi\colon \calH \to \calK$ with gate complexity $C(\Phi)$. Then, we have the inclusions
		\begin{align}
			\Phi^{\dagger}[\calQ(\calK; G)] &\subseteq \calQ(\calH; G + C(\Phi)) ,\\
			\calM(\calK; G; K) \circ \Phi &\subseteq \calM(\calH; G + C(\Phi);K).
		\end{align}
		If $\Phi = \calV$ is an isometry, then we additionally have the reverse inclusions
		\begin{align}
			\calV^{\dagger}[\calQ(\calK; G)] &\supseteq \calQ(\calH; G - C(\calV)) , \\
			\calM(\calK; G; K) \circ \calV &\supseteq \calM(\calH; G - C(\calV);K).
		\end{align}
	\end{lemma}
	\begin{proof}
		Let us draw the tensor diagrams for $\Lambda \in \calQ(\calK; G)$ and $\Phi^{\dagger}$:
		\begin{center}
			\includegraphics[width=0.7\columnwidth]{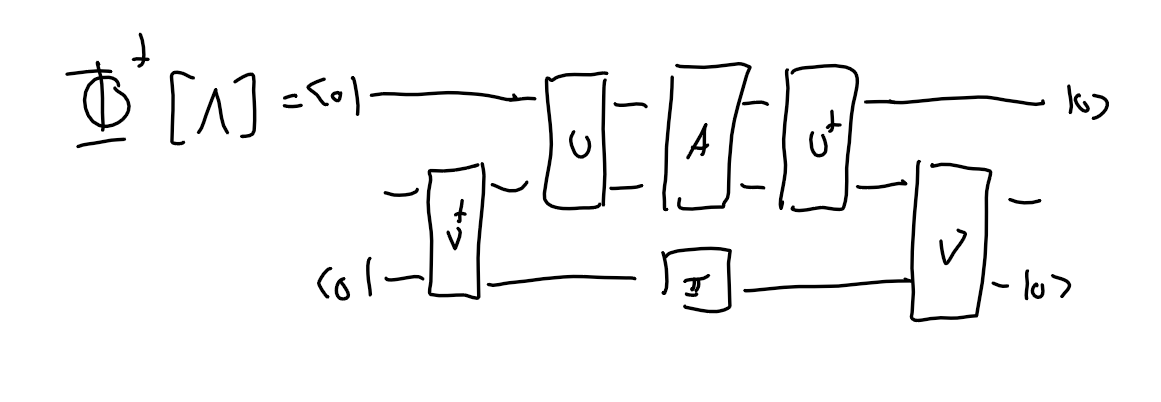}
		\end{center}
		Here we have used  $U$ to denote the unitary associated to $\Lambda$ via \cref{definition:complexity_limited_measurements} and $V$ the unitary associated to $\Phi$ via \cref{def:channel_gate_complexity}. A look at the tensor diagram makes it immediately clear that $\Phi^{\dagger}[\Lambda]$ can be understood as a POVM with central operator $A \otimes \bbI$ and associated unitary $(U \otimes \bbI)(\bbI \otimes V^{\dagger})$. This means the gate complexity of $\Lambda$ is at most $C((U \otimes \bbI)(\bbI \otimes V^{\dagger})) \leq C(U) + C(V)$, implying the desired statement
		\begin{align}
			\Phi^{\dagger}[\calQ(\calK; G)] \subseteq \calQ(\calH; G + C(\Phi))
		\end{align}
		as $\Lambda$ was arbitrary. The second statement follows directly from the above as all the POVM effects $\Lambda_k$ making up the measurements $M \in \calM(\calK; G; K)$ are inside $\calQ(\calK; G)$ and that prepending the measurement channel with $\Phi$ amounts to measuring the POVM effects $\Phi^{\dagger}[\Lambda_k]$.
		
		The reverse inclusions for isometries $\Phi = \calV$ follow from the simple fact that we can always undo the isometry at gate cost $C(\calV)$, which means we can certainly realize all operations that require $G - C(\calV)$ gates.
	\end{proof}

	\section{Warmup: The likelihood ratio test}
	The most important primitive in classical hypothesis testing is the likelihood ratio test. The Neyman–Pearson lemma establishes that this test is optimal with respect to the trade-off between type I and type II error probabilities. We will use a two-sided version of the likelihood ratio test that essentially captures the relative entropy typical set, but because it essentially comes down to the same test we will still refer to it as the likelihood ratio test. We will use the fact that this test can be efficiently implemented when the classical distributions have polynomial support multiple times in this work, as such we present a detailed analysis (see also Ref.\ \cite{cover1991elements}). We note, however, that the efficient implementation refers to the the \emph{existence} of a polynomial complexity test, which does not necessarily imply that the sequence of tests can be generated in polynomial complexity across different values of the scaling parameter $n$. We will revisit this distinction later.	
	\begin{proposition}[Efficient likelihood ratio testing]\label{proposition:efficient_likelihood_ratio_testing}
		Consider two classical probability distributions $p_n, q_n \in \calH_n$ with polynomial support. Their relative entropy and relative entropy variance are given by 
		\begin{align}
			D(p_n \fatpipe q_n) &= \Tr\Big[p_n \log\frac{p_n}{q_n}\Big], \\
			V(p_n \fatpipe q_n) &= \Tr\Big[p_n \Big(\log \frac{p_n}{q_n}\Big)^2\Big] - D(p_n \fatpipe q_n)^2.
		\end{align}
		There exists a statistical test $S_n^m(\delta)$ on $m = \poly(n)$ copies implementable in polynomial complexity that fulfills
		\begin{align}
			-\frac{1}{m}\log \Tr[ q_n^{\otimes m} S_n^m(\delta)] &\geq  D(p_n \fatpipe q_n) - \delta - \negl(n),\\
			\Tr[ p_n^{\otimes m} S_n^m(\delta)] &\geq 1- \frac{1}{m} \frac{V(p_n \fatpipe q_n)}{(\delta - \negl(n))^2}.
		\end{align}
	\end{proposition}
	\begin{proof}
		Let $\calX$ denote the alphabet of the distributions $p_n$ and $q_n$.
		We wish to implement a two-sided version of a likelihood ratio test, which one could also see as a typicality test for the relative entropy. When considering $m$ copies, it takes the form 
		\begin{align}
			S_n^m(\delta) \coloneqq \Bigg\{ \Bigg|\frac{1}{m} \sum_{i=1}^{m} \log \frac{p^{(i)}_n}{q^{(i)}_n} - D(p_n\fatpipe q_n) \Bigg|\leq \delta \Bigg\},
			\label{eq:test}
		\end{align}
		where the index $i$ enumerates the independent copies. Above, the notation $\{ \dots \}$ indicates the projector onto all values that fulfill the condition inside of the parentheses.
		We do so by measuring the $m$ copies independently to obtain $m$ samples $x_1, x_2, \dots x_m \in \calX$. In general, to perform a likelihood ratio test, we can compute the average log-likelihoods
		\begin{align}
			\frac{1}{m}\log p_n^{\otimes m}(x_1, x_2, \dots, x_m) &= \frac{1}{m}\left(\log p_n(x_1) + \log p_n(x_2) + \dots +\log p_n(x_m)\right),\\
			\frac{1}{m}\log q_n^{\otimes m}(x_1, x_2, \dots, x_m) &= \frac{1}{m}\left(\log q_n(x_1) + \log q_n(x_2) + \dots +\log q_n(x_m)\right)
		\end{align}
		and compare their difference to the reference value $D(p_n \fatpipe q_n)$. Our idea to do so is to store the individual log-likelihoods $\log p_n(x)$ and $\log q_n(x)$ in a lookup table of size $|\calX| \leq \poly(n)$. The question we face is to what additive precision $\kappa$ we have to store the individual log-likelihoods. Because the average log-likelihood is a sample mean, the additive precision on the computed values will also be $\kappa$. This means that the effective test effect under precision $\kappa$, $S_n^m(\delta; \kappa)$ fulfills
		\begin{align}
			S_n^m(\delta - 2\kappa) \leq S_n^m(\delta; \kappa) \leq S_n^m(\delta+ 2 \kappa),
		\end{align}
		where the factor $2$ arises because the errors on both means could conspire together and $\leq$ is the positive semidefinite ordering.
		We can thus perform a likelihood ratio test with threshold $\delta + 2\kappa$ using a lookup table of size $O(|\calX| \log \frac{1}{\kappa})$. As looking into the lookup table and basic arithmetic can be performed in polynomial time, this means that we can perform the total test in polynomial time for all $|\calX| \leq \poly(n)$ and $\kappa \geq \exp(-\poly(n))$.
		
		What follows is a standard treatment of the likelihood ratio test. 
		We can see the average log-likelihood ratio as the sample mean of independent copies of the random variables $Z(x) \coloneqq \log \frac{p_n(x)}{q_n(x)}$. When sampling from $p_n$, we expect it to concentrate around its mean
		\begin{align}
			\Tr\left[ p_n\log \frac{p_n}{q_n}\right] = D(p_n \fatpipe q_n).
		\end{align}
		For a given value of $\delta$ and under precision $\kappa$, we have a type II error bounded as
		\begin{align}
			\Tr[q_n^{\otimes m} S_n^m(\delta; \kappa)] 
			&\leq \Tr[q_n^{\otimes m} S_n^m(\delta+2 \kappa)] \\
			\nonumber
			&\leq e^{-m [D(p_n \fatpipe q_n) - \delta -  2\kappa]} \Tr[p_n^{\otimes m} S^m_n(\delta + 2\kappa)] \\
			\nonumber
			&\leq e^{-m [D(p_n \fatpipe q_n) - \delta - 2\kappa]}.
			\nonumber
		\end{align}
		Here we have exploited the fact that, by definition of typical test (\cref{eq:test}), on the support of $S_n^m(\delta)$ we have $q_n^{\otimes m} \leq e^{-m (D(p_n\fatpipe q_n) - \delta)} p_n^{\otimes m}$ and that the remaining projector term is bounded by 1. 
		The type I error is given by
		\begin{align}
			\Tr[ p_n^{\otimes m} (\bbI - S_n^m(\delta;\kappa))] 
			&\leq \Tr[ p_n^{\otimes m} (\bbI - S_n^m(\delta -2 \kappa))] 
			\\
			&= \Tr\Bigg[ p_n^{\otimes m}\Bigg\{ \Bigg|\frac{1}{m} \sum_{i=1}^{m} Z^{(i)} -D(p_n \fatpipe q_n) \Bigg| >   \delta - 2 \kappa \Bigg\}\Bigg],
			\nonumber
		\end{align}
		where we have used the index $i$ to index the indpendent copies of $Z$.
		In other words, this is the probability that the sample mean of $m$ copies of the variable $Z$ deviates from its mean by at least $\delta - 2\kappa$.            
		Chebyshev's inequality tells us that this probability is bounded as
		\begin{align}
			\Tr[ p_n^{\otimes m} (\bbI - S_n^m(\delta-2\kappa))] \leq
			\frac{V(p_n\fatpipe q_n)}{m (\delta - 2\kappa)^2},
		\end{align}
		where we recall that
		\begin{align}
			V(p_n \fatpipe q_n) = \Tr\left[ p_n \log^2 \frac{p_n}{q_n}\right] - D(p_n \fatpipe q_n)^2
		\end{align}
		is the variance of the random variable $Z$. The proposition follows by choosing $\kappa = e^{-n} = \negl(n)$. 
	\end{proof}
	
	A further primitive we will use is the following more restrictive test, which uses a projector onto the strongly typical set of the first argument. It incurs more stringent requirements on the distributions that are tested, but at the same time also gives us more structure to work with.
	\begin{proposition}[Efficient typicality testing]\label{proposition:efficient_typicality_testing}
		Consider two classical probability distributions $p_n, q_n \in \calH_n$ with dimension $d_n \leq O(\poly(n))$. Then, the strongly $\delta$-typical projector associated to $p_n$ over $m$ copies
		\begin{align}
			T_n^m(\delta) &= \Bigg\{ 1 +\delta \geq \frac{1}{p_n(x)}\frac{m_x}{m} \geq 1-\delta \text{ for all }x \Bigg\},
		\end{align}
		where $m_x$ is the empirical frequency of the symbol $x$, can be implemented in polynomial complexity and guarantees, 
		\begin{align}
			-\frac{1}{m} \log \Tr[ q_n^{\otimes m} T_n^m(\delta)] &\geq D(p_n \fatpipe q_n)  -\tilde\delta (D(p_n \fatpipe q_n) + 2 \log d_n + 4) -\frac{ d_n \log (1+m)}{m} \\
			\Tr[p_n^{\otimes m} T_n^m(\delta)]&\geq 1 - \exp\left( -m2 \tilde\delta^2 p_{n, \min}^2 +\log (2 d_n)\right),
		\end{align}
		where $p_{n,\min} \coloneqq \min_x p_n(x)$ and $\tilde\delta \simeq \delta$.
	\end{proposition}
	\begin{proof}
		The bound on the type II error is a direct consequence of Sanov's Theorem~\cite{cover1991elements}, which posits that
		\begin{align}
			-\log \Tr[ q_n^{\otimes m} T_n^m(\delta)] \geq d_n \log (m+1) +m \min_{p^* \in T_n^m(\delta)} D(p^{*} \fatpipe q_n).
		\end{align}
		As all distributions $p^* \in T_n^m(\delta)$ fulfill $1-\delta \leq \frac{p^*}{p_n} \leq 1+ \delta$ by construction, we have the bound
		\begin{align}
			D(p^* \fatpipe q_n) &= \Tr[p^* \log p^*] - \Tr[p^* \log q] \\
			&\geq (1+\delta) \Tr[ p_n\log ((1-\delta) p_n)] - (1-\delta)\Tr[ p_n \log q]\nonumber  \\
			&= (1+\delta) \Tr[ p_n\log p_n] + (1+\delta)\log (1-\delta) - (1-\delta)\Tr[ p_n \log q] \nonumber \\
			&= D(p_n\fatpipe q_n) + \delta \Tr[p_n (\log p_n + \log q)] + (1+\delta) \log (1-\delta)\nonumber  \\
			&\geq  D(p_n\fatpipe q_n) - \delta(D(p_n \fatpipe q_n) + 2 S(p_n) + 4)\nonumber \\
			&\geq  D(p_n\fatpipe q_n) - \delta(D(p_n \fatpipe q_n) + 2\log d_n + 4).\nonumber 
		\end{align}
		The signs in $(1\pm\delta)$ are governed by the fact that $\log p \leq 0$ for all probability distributions. We additionally used that $(1+\delta) \log (1-\delta) \geq -4\delta$ as long as $\delta \leq 4/5$ and the formula of the relative entropy to rewrite $\Tr[p_n \log q_n]$ in terms of the relative entropy and the entropy $S(p_n) = -\Tr[p_n \log p_n]$, which in turn we bounded by its maximal value $\log d_n$.
		
		The bound on the type I error follows from standard concentration arguments. The empirical frequency $0 \leq m_x/m \leq 1$ can be seen as the sample mean of $m$ independent random variables that indicate the presence of the symbol $x$. As such, we can apply Hoeffding's inequality which states that
		\begin{align}
			\bbP\big[ |m_x/m - p_n(x)| \leq \delta  p_n(x)] \geq 1- 2 \exp\left( -m2 \delta^2 p_n^2(x)\right). 
		\end{align}
		Taking the union bound over all $x$ leads to
		\begin{align}
			\bbP\big[  |m_x/m - p_n(x)| \leq \delta  p_n(x) \text{ for all } x\big]&\geq \prod_{x=1}^{d_n} \left( 1- 2 \exp\left( -m2 \delta^2 p_n^2(x)\right)\right) \\
			&\geq 1- 2 \sum_{x=1}^{d_n}\exp\left( -m2 \delta^2 p_n^2(x)\right) \nonumber \\
			&\geq 1 - 2 d_n \exp\left( -m2 \delta^2 p_{n, \min}^2 \right)\nonumber \\
			&= 1 - \exp\left( -m2 \delta^2 p_{n, \min}^2 +\log (2 d_n)\right).\nonumber 
		\end{align}
		We can hence achieve a type I error of $\epsilon$ if we choose
		\begin{align}
			m \geq \frac{1}{2 \delta^2 p_{n,\min}^2} \log \frac{2d_n}{\epsilon}.
		\end{align}
		
		Having figured out the error probabilities, we are left to determine the complexity of implementing the typicality test. The only critical calculation is the comparison of the empirical frequency with $p_n(x)$. We store all $d_n$ values of $\log p_n(x)$ in a lookup table, which with polynomial complexity allows us for exponential precision in $\log p_n(x)$. We then compare it to the hardcoded values of $\log (1\pm\delta)$. Any additive error $\kappa$ on $\log p_n(x)$ translates to a fudging of $\delta + O(\kappa)$ for sufficiently small $\delta$ and $\kappa$. As we are free to choose $\kappa = e^{-n}$ which is negligible, we obtain the desired result.
	\end{proof}

	\section{Asymmetric hypothesis testing under complexity limitations: Computational relative entropy}
	The relative entropy is arguably one of the most important quantities in standard asymptotic information theory. It does not only quantify many interesting information theoretic tasks directly, but it can also be used to derive a host of other information-theoretic quantities such as entropy or mutual information.
	
	It is our aim to establish a similarly foundational quantity in computational quantum information theory. To do so, we follow the operational route stemming from asymmetric hypothesis testing. Indeed, the relative entropy can be defined by regularizing the \emph{hypothesis testing relative entropy}~\cite{khatri2024principles}, which quantifies the best achievable error exponent in an asymmetric hypothesis test in a single-shot way, over infinitely many i.i.d.\ copies of the involved states.
	Many of the simplification that asymptotic information theory has compared to single-shot information theory can be seen to stem from the asymptotic limit \enquote{regularizing away} correction terms stemming from finite-size effects. 
	
	Our approach to defining a computational relative entropy follows this idea: we need a version of the hypothesis testing relative entropy that takes complexity limitations into account, and afterwards regularize it over \emph{polynomially many} copies. We have already built up the machinery to make sense of polynomial regularization in \cref{sec:rigorous_negligible_limits}, so let us see how we come to the definition.
	
	We base our definition on the $G$-complexity hypothesis testing relative entropy introduced in Ref.~\cite{munson2025complexity-constrained}. We use a slightly different computational model in our definition, but the spirit remains the same. Intuitively, this quantity captures the best achievable error exponent in an asymmetric hypothesis test when the test has to be implemented with $G$ many gates.
	\begin{definition}[$G$-complexity hypothesis testing relative entropy~\cite{munson2025complexity-constrained}]\label{def:g_complexity_hypothesis_testing_relative_entropy}
		Consider two quantum states $\rho, \sigma \in \calH$. Their $G$-complexity hypothesis testing relative entropy is defined as 
		\begin{align}
			{D}^{\epsilon}_h(\rho \fatpipe \sigma ; G) \coloneqq - \log \min_{\Lambda \in \calQ(\calH; G)} \big\{ \Tr[ \Lambda \sigma ]
			\bigpipe \Tr[\Lambda \rho] \geq 1- \epsilon \big\}.
		\end{align}
		The associated type II error is denoted as $\beta_h^{\epsilon}(\rho \fatpipe \sigma; G) \coloneqq \exp( - D_h^{\epsilon}(\rho_n \fatpipe \sigma; G))$.
	\end{definition}
	We note that the minimum in the above optimization always exists. As the set of two-local unitaries is a set of matrices, we obtain their compactness from the fact that the set is closed and bounded. Their combination into larger unitaries yields a compact set of unitary circuits of bounded complexity. The set of POVM effects is obtained by combining this with a discrete set of elementary POVM effects. Hence, the set of POVM effects of bounded complexity is compact. The mapping from this set to the objective above is continuous, hence the claim that the minimum always exists.
	
	We are now ready to give the definition of the computational relative entropy by regularizing the $G$-complexity hypothesis testing relative entropy. We again emphasize that care has to be taken that the number of copies used for the regularization is also polynomial, as a super-polynomial number of copies cannot even be touched with polynomially many gates.
	\begin{definition}[Computational relative entropy]\label{def:computational_relative_entropy}
		Consider two states $\rho_n, \sigma_n \in \calH_n$. The \emph{computational relative entropy} is defined as the worst-case exponent of asymmetrically testing $\rho_n$ versus $\sigma_n$ with vanishing type I error:
		\begin{align}\label{eq:asympt-poly-time-upper-lower} 
			{\underline{D}}(\rho_n \fatpipe \sigma_n) &\colonsimeq \flim_{\epsilon \to 0}\flim_{\ell \to \infty} \fliminf_{k \to \infty} \left\{ \frac{1}{n^k} D^{\epsilon}_h(\rho_n^{\otimes n^k} \fatpipe \sigma_n^{\otimes n^k}; n^{k\ell}) \right\}.
		\end{align}
	\end{definition}
	In the above definition, we achieve a polynomial regularization by taking $n^k$ copies and taking the exponent to infinity, with the aim of later using \cref{lemma:limit_of_inverse_polynomials_is_negligible} to regularize away correction terms as they would become negligible. At the same time, we are allowed $n^{k\ell} = (n^{k})^{\ell}$ many quantum gates and take the exponent $\ell$ to infinity as well, which gives us access to any polynomial in both the scaling variable and the number of copies. As increasing the number of gates can only increase the complexity-limited hypothesis testing relative entropy, the limit in $\ell$ always exists. The same reasoning applies to the limit $\epsilon\to 0$, it exists because the complexity-limited hypothesis testing relative entropy is monotonically decreasing as $\epsilon$ decreases. For more information, the reader may consult \cref{sec:limits_of_polynomial_resources}.
	
	We further note that the above is a \emph{non-uniform} definition of computational relative entropy. This means that we are content with a circuit of polynomial complexity existing for each $n$. It would also be conceivable to define a \emph{uniform} version of the computational relative entropy where there has to exist a program of polynomial complexity that upon input of $n$ outputs a polynomial complexity circuit to perform the hypothesis test. We settle for the non-uniform approach because it better captures the ultimate limits of computationally bounded hypothesis testing and because we can always perform hypothesis test of distributions of constant dimension in the non-uniform model, but not in the uniform model. This can be seen by encoding the unary halting problem into the sequence of states $\rho_n$ and $\sigma_n$. We believe a statistical quantity should allow us to perform such tests, underpinning our choice of the non-uniform approach.
	
	Another way of defining a restricted version of relative entropy would be to only constrain the number of gates relative to the number of copies and not relative to the dimension of the states in question. In such a model, however, the complexity of the states in question does not factor in and Haar random states could be perfectly distinguished. We find that this approach does not capture the essence of computational restrictions. Considering only polynomially many copies and unrestricted tests -- the setting of polynomial block length -- is similarly problematic in our opinion.
	
	The computational relative entropy has a number of interesting and natural properties.
	\begin{proposition}[Properties of computational relative entropy]\label{proposition:properties_of_computational_relative_entropy}
		The computational relative entropy has the following properties:
		\begin{propenum}
			\item \label{item:computational_relative_entropy_boundedness} \emph{Boundedness.} We have
			\begin{align}
				0 \lesssim \underline{D}(\rho_n \fatpipe \sigma_n) \lesssim D(\rho_n \fatpipe \sigma_n).
			\end{align}
			
			\item \label{item:computational_entropy_multiplication_of_second_argument}
			\emph{Prefactor in second argument.}
			We have
			\begin{align}
				\underline{D}(\rho_n \fatpipe \alpha \sigma_n) \simeq \underline{D}(\rho_n \fatpipe \sigma_n) - \log \alpha.
			\end{align}
			
			\item \label{item:computational_relative_entropy_tensor_additivity} \emph{Additivity under tensor powers.} Let $m\leq \poly(n)$. Then,
			\begin{align}
				\underline{D}(\rho_n^{\otimes m} \fatpipe \sigma_n^{\otimes m}) \simeq m \underline{D}(\rho_n \fatpipe \sigma_n).
			\end{align}
			\item \label{item:computational_relative_entropy_superadditivity} \emph{Superadditivity under tensor products.} We have 
			\begin{align}
				\underline{D}(\rho_n\otimes \tau_n \fatpipe \sigma_n \otimes \upsilon_n) \gtrsim \underline{D}(\rho_n\fatpipe \sigma_n ) + \underline{D}(\tau_n \fatpipe \upsilon_n).  
			\end{align}
			
			\item \label{item:computational_relative_entropy_polynomial_catalyst}
			\emph{No polynomial catalysts.} Let $\tau_n$ be a state of polynomial complexity, $C(\tau_n) \leq O(\poly(n))$. Then,
			\begin{align}
				\underline{D}(\rho_n \otimes \tau_n \fatpipe \sigma_n \otimes \tau_n) &\simeq \underline{D}(\rho_n \fatpipe \sigma_n).
			\end{align}
			
			\item \label{item:computational_relative_entropy_data_processing} \emph{Data processing.} Let $\Phi_n\colon \calH_n \to \calK_n$ be a quantum channel of polynomial gate complexity. Then,
			\begin{align}
				\underline{D}(\Phi_n[\rho_n] \fatpipe \Phi_n[\sigma_n]) &\lesssim \underline{D}(\rho_n \fatpipe \sigma_n).
			\end{align}
			\item \label{item:computational_relative_entropy_isometric_invariance} \emph{Isometric invariance.} Let $V_n\colon\calH_n \to \calK_n$ be an isometry of polynomial gate complexity. Then,
			\begin{align}
				\underline{D}(V_n \rho_n V_n^{\dagger}\fatpipe  V_n\sigma_n V_n^{\dagger}) \simeq \underline{D}(\rho_n \fatpipe \sigma_n).
			\end{align}
			
			\item \emph{Classical-quantum states.} \label{item:computational_relative_entropy_classical_quantum_states}
			Consider two classical-quantum states $\sum_{x=1}^{M_n} p_n(x) |x\rangle\!\langle x| \otimes \rho_n^x$ and $\sum_{x=1}^{M_n} q_n(x) |x\rangle\!\langle x| \otimes \sigma_n^x$ with polynomially sized classical registers such that $D(p_n \fatpipe q_n), \underline{D}^{\mathbbm{2}}(\rho_n^x \fatpipe \sigma_n^x), \max_x(p_n^{-1}(x)) \leq O(\poly(n))$. Then,
			\begin{align}
				\underline{D}\left(\left.\sum_{x=1}^{M_n} p_n(x) |x\rangle\!\langle x| \otimes \rho_n^x \, \right\| \, \sum_{x=1}^{M_n} q(x) |x\rangle\!\langle x| \otimes \sigma_n^x\right)
				&\gtrsim D(p_n \fatpipe q_n) + \sum_{x=1}^{M_n} p_n(x) \underline{D}(\rho_n^x \fatpipe \sigma_n^x).
			\end{align}
		\end{propenum}		
	\end{proposition}
	We note that data processing and isometric invariance do in general not hold under arbitrary maps.
	\begin{proof}
		We provide separate proofs for each fact.
		\begin{enumerate}[label=(\roman*)]
			\item The lower bound is trivially following from the non-negativity of the hypothesis testing relative entropy. For the upper bound, we use the standard converse for the hypothesis testing relative entropy \cite[Proposition 7.70]{khatri2024principles}
			\begin{align}
				\underline{D}(\rho_n \fatpipe \sigma_n) &\simeq \flim_{\epsilon \to 0} \flim_{\ell \to \infty}\fliminf_{k\to\infty} 
				\frac{1}{n^k} D^{\epsilon}_h(\rho_n^{\otimes n^k} \fatpipe \sigma_n^{\otimes n^k}; n^{k \ell}) \\
				\nonumber
				&\lesssim \flim_{\epsilon \to 0} \flim_{\ell \to \infty}\fliminf_{k\to\infty} 
				\frac{1}{n^k} D^{\epsilon}_h(\rho_n^{\otimes n^k} \fatpipe \sigma_n^{\otimes n^k})\\
				\nonumber
				&\lesssim \flim_{\epsilon \to 0} \flim_{\ell \to \infty}\fliminf_{k\to\infty} 
				\frac{1}{1-\epsilon}\left( D(\rho_n \fatpipe \sigma_n) + \frac{1}{n^k}h_2(\epsilon) \right) \\
				\nonumber
				&\simeq D(\rho_n \fatpipe \sigma_n),
			\end{align}
			where we have concluded the last line as $\flim_{\epsilon \to 0} h_2(\epsilon) = 0$.
			
			\item This follows directly from the definition of the $G$-complexity hypothesis testing relative entropy.
			
			\item As $m \leq \poly(n)$, we have that $m \leq n^q$ for sufficiently large $n$ and hence $m^\ell \leq n^{q\ell}$. We expand the definition to observe
			\begin{align}
				\underline{D}(\rho_n^{\otimes m} \fatpipe \sigma_n^{\otimes m}) &\simeq m \flim_{\epsilon \to 0} \flim_{\ell \to \infty}\fliminf_{k\to\infty} 
				\frac{1}{m n^k} D^{\epsilon}_h(\rho_n^{\otimes m n^k} \fatpipe \sigma_n^{\otimes m n^k}; n^{k \ell}).
			\end{align}
			We see that $m n^k$ is a polynomial number of copies, and as such the limit $k \to \infty$ will result in the limit over all polynomial numbers of copies. We are left to see that we have enough gates, as we would need $m^\ell n^{k\ell}$. As $m^\ell \leq n^{q\ell}$ we have that $n^{k\ell'} \geq m^\ell n^{k\ell}$ if $k \ell' \geq (k + q)\ell$ which is guaranteed for $\ell' \geq \frac{k+q}{k} \to 1$ as $k\to \infty$. Hence, we can assume $k$ sufficiently large and choose $\ell' = 2\ell$. As taking $\ell \to \infty$ is the same as taking $\ell' \to \infty$, we observe the desired equality.
			
			\item 
			We first study the additivity of the $G$-complexity hypothesis testing relative entropy for states $\rho, \sigma \in \calH$ and $\tau, \upsilon \in \calK$
			\begin{align}
				D_h^{\epsilon}(\rho \otimes \tau \fatpipe \sigma \otimes \upsilon; G) &= -\log \min_{\Lambda \in \calQ(\calH \otimes \calK; G)}\big\{ \Tr[\Lambda(\sigma \otimes \upsilon)] \pipe \Tr[\Lambda (\rho \otimes \tau)]\geq 1-\epsilon\big\}.
			\end{align}
			Let now $\Lambda_1$ and $\Lambda_2$ denote the optimal POVM effects for the calculation of $D_h^{\epsilon/2}(\rho  \fatpipe \sigma ; G/2)$ and $D_h^{\epsilon/2}(\tau  \fatpipe \upsilon ; G/2)$. Then, the combined POVM effect $\Lambda_1 \otimes \Lambda_2$ has complexity at most $G$ and achieves
			\begin{align}
				\Tr[(\Lambda_1 \otimes \Lambda_2)(\rho \otimes \tau)] = \Tr[\Lambda_1 \rho]\Tr[\Lambda_2 \tau] \geq (1-\epsilon/2)^2 = 1 - \epsilon + \epsilon^2/4 \geq 1 - \epsilon
			\end{align}
			and
			\begin{align}
				-\log \Tr[(\Lambda_1 \otimes \Lambda_2)(\sigma \otimes \upsilon)] = -\log \Tr[\Lambda_1 \sigma] - \log \Tr[\Lambda_2 \upsilon] = D_h^{\epsilon/2}(\rho \fatpipe \sigma; G/2) + D_h^{\epsilon/2}(\tau\fatpipe \upsilon; G/2),
			\end{align}    
			showing that
			\begin{align}
				D_h^{\epsilon}(\rho \otimes \tau \fatpipe \sigma \otimes \upsilon; G) \geq D_h^{\epsilon/2}(\rho \fatpipe \sigma; G/2) + D_h^{\epsilon/2}(\tau\fatpipe \upsilon; G/2).
			\end{align}
			To apply this result, we expand the definition of the computational relative entropy
			\begin{align}
				\underline{D}(\rho_n \otimes \tau_n \fatpipe \sigma_n \otimes \upsilon_n) &\simeq \flim_{\epsilon \to 0} \flim_{\ell \to \infty}\fliminf_{k\to\infty} 
				\frac{1}{n^k} D^{\epsilon}_h(\rho_n^{\otimes n^k} \otimes \tau_n^{\otimes n^k} \fatpipe \sigma_n^{\otimes n^k} \otimes \upsilon_n^{\otimes n^k}; n^{k \ell}) \\
				\nonumber
				&\gtrsim \flim_{\epsilon \to 0} \flim_{\ell \to \infty}\fliminf_{k\to\infty} 
				\frac{1}{n^k} D^{\epsilon/2}_h(\rho_n^{\otimes n^k}  \fatpipe \sigma_n^{\otimes n^k}; n^{k \ell} / 2) + \frac{1}{n^k} D^{\epsilon/2}_h(\tau_n^{\otimes n^k}  \fatpipe \upsilon_n^{\otimes n^k}; n^{k \ell} / 2) \\
				\nonumber
				&\simeq  \underline{D}(\rho_n\fatpipe \sigma_n ) + \underline{D}(\tau_n \fatpipe \upsilon_n).
				\nonumber
			\end{align}
			
			\item This property follows from the fact that we can absorb the preparation of $\tau_n$ into any measurement. We define $\Lambda_\tau \coloneqq \Tr_{\calK}[\Lambda (\bbI \otimes \tau)]$
			\begin{align}
				D_h^{\epsilon}(\rho \otimes \tau \fatpipe \sigma \otimes \tau; G) &= -\log \min_{\Lambda \in \calQ(\calH \otimes \calK; G)}\big\{ \Tr[\Lambda(\sigma \otimes \tau)] \pipe \Tr[\Lambda (\rho \otimes \tau)]\geq 1-\epsilon\big\} \\
				\nonumber
				&= -\log \min_{\Lambda \in \calQ(\calH \otimes \calK; G)}\big\{ \Tr[\Lambda_\tau \sigma ] \pipe \Tr[\Lambda_\tau \rho ]\geq 1-\epsilon\big\} \\
				\nonumber
				&\leq - \log \min_{\Lambda \in \calQ(\calH \otimes \calK; G + C(\tau))}\big\{ \Tr[\Lambda \sigma ] \pipe \Tr[\Lambda \rho ]\geq 1-\epsilon\big\} \\
				\nonumber
				&= D_h^{\epsilon}(\rho \fatpipe \sigma; G + C(\tau)).
			\end{align}
			Here, the inequality follows from the fact that $C(\Lambda_\tau) \leq G + C(\tau)$. As we assumed $\tau_n$ is a state of polynomial complexity, we have $C(\tau_n) \leq n^q$ for some $q \in \bbN$ for sufficiently large $n$ and can conclude that
			\begin{align}
				\underline{D}(\rho_n \otimes \tau_n \fatpipe \sigma_n \otimes \tau_n)  
				&\simeq \flim_{\epsilon \to 0} \flim_{\ell \to \infty}\fliminf_{k\to\infty} 
				\frac{1}{n^k} D^{\epsilon}_h(\rho_n^{\otimes n^k} \otimes \tau_n^{\otimes n^k} \fatpipe \sigma_n^{\otimes n^k} \otimes \tau_n^{\otimes n^k}; n^{k \ell}) \\
				\nonumber
				&\lesssim \flim_{\epsilon \to 0} \flim_{\ell \to \infty}\fliminf_{k\to\infty} 
				\frac{1}{n^k} D^{\epsilon}_h(\rho_n^{\otimes n^k}\fatpipe \sigma_n^{\otimes n^k}; n^{k \ell} + n^k C(\tau_n))
				\\
				\nonumber
				&\lesssim \flim_{\epsilon \to 0} \flim_{\ell \to \infty}\fliminf_{k\to\infty} 
				\frac{1}{n^k} D^{\epsilon}_h(\rho_n^{\otimes n^k}\fatpipe \sigma_n^{\otimes n^k}; n^{k (\ell+q)}) \\
				&\simeq \underline{D}(\rho_n \fatpipe \sigma_n).
				\nonumber
			\end{align}
			Here, we exploited that the limited $\ell \to \infty$ is the same as $\ell + q\to \infty$.
			The stated result then follows by combining with the superadditivity shown above (\cref{item:computational_relative_entropy_superadditivity}).
			
			\item 
			As $\Phi_n$ has polynomial gate complexity, we know there exists a $q \in \bbN$ such that for sufficiently large $n$ we have $C(\Phi_n)\leq n^q$. We expand the expression of the computational relative entropy as 
			\begin{align}
				\underline{D}(\Phi_n[\rho_n] \fatpipe \Phi_n[\sigma_n]) &\simeq \flim_{\epsilon \to 0} \flim_{\ell \to \infty}\fliminf_{k\to\infty} 
				\frac{1}{n^k} D^{\epsilon}_h(\Phi_n[\rho_n]^{\otimes n^k} \fatpipe \Phi_n[\sigma_n]^{\otimes n^k}; n^{k \ell}) 
				\\
				\nonumber
				&\simeq\flim_{\epsilon \to 0} \flim_{\ell \to \infty}\fliminf_{k\to\infty} 
				- \frac{1}{n^k} \log \min_{\Lambda \in \calQ(\calK_n^{\otimes n^k}; n^{k \ell})} \left\{\left. \Tr[ \Lambda \Phi_n[\sigma_n]^{\otimes n^k}] \, \right| \, \Tr[ \Lambda \Phi_n[\rho_n]^{\otimes n^k}] \geq 1 - \epsilon\right\} \\
				\nonumber
				&\simeq\flim_{\epsilon \to 0} \flim_{\ell \to \infty}\fliminf_{k\to\infty} 
				- \frac{1}{n^k} \log \min_{\Lambda \in \calQ(\calK_n^{\otimes n^k}; n^{k \ell})} \left\{\left. \Tr[ \Phi_n^{\otimes n^k \dagger}[ \Lambda ] \sigma_n^{\otimes n^k}] \, \right| \, \Tr[\Phi_n^{\otimes n^k \dagger}[ \Lambda ] \rho_n^{\otimes n^k}] \geq 1 - \epsilon\right\} \\
				\nonumber
				&\simeq\flim_{\epsilon \to 0} \flim_{\ell \to \infty}\fliminf_{k\to\infty} 
				- \frac{1}{n^k} \log \min_{\Lambda \in \Phi_n^{\otimes n^k \dagger}[\calQ(\calK_n^{\otimes n^k}; n^{k \ell})]} \left\{\left. \Tr[  \Lambda  \sigma_n^{\otimes n^k}] \, \right| \, \Tr[\Lambda  \rho_n^{\otimes n^k}] \geq 1 - \epsilon\right\} \\
				\nonumber
				&\lesssim\flim_{\epsilon \to 0} \flim_{\ell \to \infty}\fliminf_{k\to\infty} 
				- \frac{1}{n^k} \log \min_{\Lambda \in \calQ(\calH_n^{\otimes n^k}; n^{k \ell} + n^{k}C(\Phi_n))} \left\{\left. \Tr[  \Lambda  \sigma_n^{\otimes n^k}] \, \right| \, \Tr[\Lambda  \rho_n^{\otimes n^k}] \geq 1 - \epsilon\right\} \\
				\nonumber
				&\lesssim\flim_{\epsilon \to 0} \flim_{\ell \to \infty}\fliminf_{k\to\infty} 
				- \frac{1}{n^k} \log \min_{\Lambda \in \calQ(\calH_n^{\otimes n^k}; n^{k (\ell+q)})} \left\{\left. \Tr[  \Lambda  \sigma_n^{\otimes n^k}] \, \right| \, \Tr[\Lambda  \rho_n^{\otimes n^k}] \geq 1 - \epsilon\right\} \\
				\nonumber
				&\simeq\underline{D}(\rho_n \fatpipe \sigma_n). 
			\end{align}
			The crucial step above is the application of \cref{lemma:absorbing_channels_into_measurement_sets} to conclude
			\begin{align}
				\Phi_n^{\otimes n^k \dagger}[\calQ(\calK_n^{\otimes n^k}; n^{k \ell})] \subseteq \calQ(\calH_n^{\otimes n^k}; n^{k \ell} + n^k C(\Phi_n)).
			\end{align}
			Maximizing over a larger set always yields a larger value, implying the desired result together with the observation that taking $\ell \to \infty$ is the same as taking $\ell + q \to \infty$.

			\item 
			As $V_n$ has polynomial gate complexity, we know there exists a $q \in \bbN$ such that for sufficiently large $n$ we have $C(V_n)\leq n^q$. We furthermore know that $V_n$ is a polynomial complexity channel, implying that $\underline{D}(V_n \rho_n V_n^{\dagger}\fatpipe  V_n\sigma_n V_n^{\dagger}) \leq \underline{D}( \rho_n \fatpipe  \sigma_n)$. To prove the reverse direction, we let $\calV_n$ denote the channel associated to $V_n$ and we expand the expression of the computational relative entropy  
			\begin{align}
				\underline{D}(\calV_n[\rho_n] \fatpipe \calV_n[\sigma_n]) &\simeq \flim_{\epsilon \to 0} \flim_{\ell \to \infty}\fliminf_{k\to\infty} 
				\frac{1}{n^k} D^{\epsilon}_h(\calV_n[\rho_n]^{\otimes n^k} \fatpipe \calV_n[\sigma_n]^{\otimes n^k}; n^{k \ell}) 
				\\
				\nonumber
				&\simeq\flim_{\epsilon \to 0} \flim_{\ell \to \infty}\fliminf_{k\to\infty} 
				- \frac{1}{n^k} \log \min_{\Lambda \in \calQ(\calK_n^{\otimes n^k}; n^{k \ell})} \left\{\left. \Tr[ \Lambda \calV_n[\sigma_n]^{\otimes n^k}] \, \right| \, \Tr[ \Lambda \calV_n[\rho_n]^{\otimes n^k}] \geq 1 - \epsilon\right\} \\ \nonumber
				&\simeq\flim_{\epsilon \to 0} \flim_{\ell \to \infty}\fliminf_{k\to\infty} 
				- \frac{1}{n^k} \log \min_{\Lambda \in \calQ(\calK_n^{\otimes n^k}; n^{k \ell})} \left\{\left. \Tr[ \calV_n^{\otimes n^k \dagger}[ \Lambda ] \sigma_n^{\otimes n^k}] \, \right| \, \Tr[\calV_n^{\otimes n^k \dagger}[ \Lambda ] \rho_n^{\otimes n^k}] \geq 1 - \epsilon\right\} \\
				\nonumber
				&\simeq\flim_{\epsilon \to 0} \flim_{\ell \to \infty}\fliminf_{k\to\infty} 
				- \frac{1}{n^k} \log \min_{\Lambda \in \calV_n^{\otimes n^k \dagger}[\calQ(\calK_n^{\otimes n^k}; n^{k \ell})]} \left\{\left. \Tr[  \Lambda  \sigma_n^{\otimes n^k}] \, \right| \, \Tr[\Lambda  \rho_n^{\otimes n^k}] \geq 1 - \epsilon\right\} \\
				\nonumber
				&\gtrsim\flim_{\epsilon \to 0} \flim_{\ell \to \infty}\fliminf_{k\to\infty} 
				- \frac{1}{n^k} \log \min_{\Lambda \in \calQ(\calH_n^{\otimes n^k}; n^{k \ell} - n^{k}C(\calV_n))} \left\{\left. \Tr[  \Lambda  \sigma_n^{\otimes n^k}] \, \right| \, \Tr[\Lambda  \rho_n^{\otimes n^k}] \geq 1 - \epsilon\right\} \\
				\nonumber
				&\gtrsim\flim_{\epsilon \to 0} \flim_{\ell \to \infty}\fliminf_{k\to\infty} 
				- \frac{1}{n^k} \log \min_{\Lambda \in \calQ(\calH_n^{\otimes n^k}; n^{k (\ell-q)})} \left\{\left. \Tr[  \Lambda  \sigma_n^{\otimes n^k}] \, \right| \, \Tr[\Lambda  \rho_n^{\otimes n^k}] \geq 1 - \epsilon\right\} \\
				\nonumber
				&\simeq\underline{D}(\rho_n \fatpipe \sigma_n). 
			\end{align}
			The crucial step above is the application of \cref{lemma:absorbing_channels_into_measurement_sets} to conclude
			\begin{align}
				\calV_n^{\otimes n^k \dagger}[\calQ(\calK_n^{\otimes n^k}; n^{k \ell})] \supseteq \calQ(\calH_n^{\otimes n^k}; n^{k \ell} - n^k C(\calV_n)).
			\end{align}
			Maximizing over a smaller set always yields a smaller value, implying the desired result together with the observation that taking $\ell \to \infty$ is the same as taking $\ell - q \to \infty$.  
			
			\item On a high level, we will establish this result via an achievability statement. We obtain it by using a strongly typical projector on the classical register, which allows us to control the empirical frequencies of seeing the states $\rho_x$ and $\sigma_x$, respectively. As we know what states we have on the different remaining quantum registers after the test on the classical register, we perform the optimal complexity limited tests on those registers.
			
			Let us first expand the definition of the computational relative entropy for the classical-quantum states $\rho_n^{XA}$ and $\sigma_n^{XA}$
			\begin{align}
				\begin{split}        &\underline{D}\left(\left.\sum_{x=1}^{M_n} p_n(x) |x\rangle\!\langle x| \otimes \rho_n^x \, \right\| \, \sum_{x=1}^{M_n} q(x) |x\rangle\!\langle x| \otimes \sigma_n^x\right) \\
					&= \flim_{\epsilon \to 0} \flim_{\ell \to \infty} \fliminf_{k\to\infty} \frac{1}{n^k} {D}_h^{\epsilon}\left(\left.\left[\sum_{x=1}^{M_n} p_n(x) |x\rangle\!\langle x| \otimes \rho_n^x\right]^{\otimes n^k} \, \right\| \, \left[\sum_{x=1}^{M_n} q_n(x) |x\rangle\!\langle x| \otimes \sigma_n^x\right]^{\otimes n^k}; n^{k\ell} \right).
				\end{split}
			\end{align}
			The test we perform will have the structure of 
			first measuring the classical register using the strongly typical projector $T_n^{n^k}(\delta_0)$ of \cref{proposition:efficient_typicality_testing} which has polynomial complexity $C(T_n^{n^k}(\delta)) \leq n^{kq_0}$ for some $q_0\in\bbN$ and achieves a rate close to $D(p_n \fatpipe q_n)$ for some $\delta_0$ to be determined later. As side information, we obtain a string $x_1 x_2 \dots x_{n^k}$ which indicates the particular state $\rho_x$ or $\sigma_x$ left on the individual copies. Let $m_x$ denote the number of times we observed the symbol $x$ such that $\sum_x m_x = n^k$. In the case that the strongly typical projector $T_n^{n^k}(\delta_0)$ accepts, we know that $(1+\delta_0) p(x) \geq m_x/n^k \geq (1-\delta_0) p(x)$.
			We use that and perform the optimal tests $\Lambda_x$ for $D_h^{\epsilon_x}((\rho^x)^{\otimes m_x} \fatpipe (\sigma_x)^{\otimes m_x}; G_x)$ for some $\epsilon_x, G_x$ to be determined later. We only predict $\rho^{XA}$ if the first test 
			outputs $p$ and all subsequent tests output $\rho_x$. Denote the associated POVM effect as $\Lambda$. Because the prediction $\rho^{XA}$ only happens when all of the (independent) events happen, we can multiply their probabilities together and get the type II error 
			\begin{align}
				\Tr[ \Lambda \left[\sum_{x=1}^{M_n} q_n(x) |x\rangle\!\langle x| \otimes \sigma_n^x\right]^{\otimes n^k}] &\leq \Tr[q_n^{\otimes n^k} T_n^{n^k}(\delta_0)] \prod_{x=1}^{M_n}\beta_h^{\epsilon_x}((\rho_n^x)^{\otimes m_x} \fatpipe (\sigma_n^x)^{\otimes m_x} ; G_x),
			\end{align}
			which corresponds to an error exponent of
			\begin{align}
				\begin{split}            
					-&\frac{1}{n^k}\log \Tr\left[ \Lambda \left[\sum_{x=1}^{M_n} q_n(x) |x\rangle\!\langle x| \otimes \sigma_n^x\right]^{\otimes n^k}\right] 
					\\ &= -\frac{1}{n^k}\log \Tr[q_n^{\otimes n^k} T_n^{n^k}(\delta_0)] + \sum_{x=1}^{M_n} \frac{m_x}{n^k} \frac
					1{m_x} D_h^{\epsilon_x}((\rho_n^x)^{\otimes m_x} \fatpipe (\sigma_n^x)^{\otimes m_x} ; G_x).
				\end{split}
			\end{align}
			As we know the empirical frequencies have to be close to the true probability as the strongly typical projector accepted, we write $m_x/n^k = p(x) + \Delta(x)$ where $|\Delta(x)| \leq \delta_0 p(x)$ so that
			\begin{align}
				-&\frac{1}{n^k}\log \Tr\left[ \Lambda \left[\sum_{x=1}^{M_n} q_n(x) |x\rangle\!\langle x| \otimes \sigma_n^x\right]^{\otimes n^k}\right] 
				\nonumber \\ &= -\frac{1}{n^k}\log \Tr[q_n^{\otimes n^k} T_n^{n^k}(\delta_0)] + \sum_{x=1}^{M_n} p(x) \frac
				1{m_x} D_h^{\epsilon_x}((\rho_n^x)^{\otimes m_x} \fatpipe (\sigma_n^x)^{\otimes m_x} ; G_x) \\&\qquad+ \sum_{x=1}^{M_n} \Delta(x) \frac
				1{m_x} D_h^{\epsilon_x}((\rho_n^x)^{\otimes m_x} \fatpipe (\sigma_n^x)^{\otimes m_x} ; G_x).\nonumber 
			\end{align}
			Now, as we know that $m_x = p(x) n^k + \Delta(x) n^k$ we know it is $O(n^k)$, which means taking $k\to\infty$ will induce a polynomial regularization as we need to make the computational relative entropy appear. Let us now take the total number of gates $n^{k\ell}$ and distribute them among the different tests we are performing, recall that the strongly typical projector is implemented with at most $n^{kq_0}$ gates for some fixed $q_0\in \bbN$. The logical that takes the observations made in the first test (the string $x_1 x_2 \dots x_{n^k}$) and assigns the different remaining copies to the appropriate tests $(\rho_n^x)^{\otimes m_x}$ versus $(\sigma_n^x)^{\otimes m_x}$ is simple and obviously also polynomial in $M_n n^{k}$, which we quantify by another $q_L \in \bbN$ such that the complexity is at most $n^{k q_L}$. We then associate to every test on $m_x$ copies, $m_x^{\ell_x} \approx p(x) n^{k\ell_x}$ many gates. In total we have
			\begin{align}
				n^{kq_0} + n^{kq_L} + \sum_x m_x^{\ell_x}
			\end{align}
			which is clearly polynomial. Choosing $\ell_x = \sqrt{\ell}$ allows us to take the exponent to infinity while guaranteeing that at some point $n^{k\ell}$ is sufficiently large to allow the implementation of our protocol. Let us next discuss the type I error. By the same argument as above, we can multiply the type I success probabilities together to obtain
			\begin{align}
				\Tr\left[ \Lambda \left[\sum_{x=1}^{M_n} p_n(x) |x\rangle\!\langle x| \otimes \rho_n^x\right]^{\otimes n^k}\right]  \geq (1-\epsilon_0)\prod_{x=1}^{M_n}(1-\epsilon_x)
				\geq 1 - \epsilon_0 - \sum_{x=1}^{M_n}\epsilon_x
			\end{align}
			which we can make $\geq 1 - \epsilon$ by choosing $\epsilon_0 = \epsilon/2$ and $\epsilon_x = p(x) \epsilon/2$ such that $\epsilon\to0$ implies $\epsilon_0 \to 0$ and $\epsilon_x \to 0$.
			
			We are left to discuss the choice of threshold $\delta_0$. Checking the guarantees from \cref{proposition:efficient_typicality_testing}, we obtain a type I error of $\epsilon/2$ if 
			\begin{align}
				\delta_0 \gtrsim \sqrt{\frac{1}{2 n^k p_{n,\min}^2} \log \frac{4d_n}{\epsilon}}.
			\end{align}
			Let us choose this value for $\delta_0$, in this case, we obtain the following achievability statement: 
			\begin{align}
				&\underline{D}\left(\left.\sum_{x=1}^{M_n} p_n(x) |x\rangle\!\langle x| \otimes \rho_n^x \, \right\| \, \sum_{x=1}^{M_n} q(x) |x\rangle\!\langle x| \otimes \sigma_n^x\right)\nonumber  \\
				&\gtrsim 
				D(p_n \fatpipe q_n) + \sum_{x=1}^{M_n} p_n(x) \underline{D}(\rho_n^x \fatpipe \sigma_n^x)
				\\ &- \flim_{\epsilon \to 0} \flim_{\ell \to \infty} \fliminf_{k\to\infty} \sqrt{\frac{1}{2 n^k p_{n,\min}^2} \log \frac{4d_n}{\epsilon}} \Bigg\{ D(p_n\fatpipe q_n) + 2 \log d_n + 4\nonumber  \\
				&\qquad + \max_{x}\frac{1}{1-p_n(x) \epsilon/2}\left( D^{\bbC}(\rho_n^x \fatpipe \sigma_n^x; m_x^{\sqrt{\ell}}; 2) + h_2(p_n(x)\epsilon/2)\right)\Bigg\} - \frac{d_n \log n^k}{n^k} \nonumber \\
				&\simeq D(p_n \fatpipe q_n) + \sum_{x=1}^{M_n} p_n(x) \underline{D}(\rho_n^x \fatpipe \sigma_n^x)\nonumber 
			\end{align}
			by our assumptions that $D(p_n \fatpipe q_n), p_{n,\min}^{-1}, d_n, \underline{D}^{\mathbbm{2}}(\rho_n^x \fatpipe \sigma_n^x) \leq O(\poly(n))$ for all $x$, which means we can use \cref{lemma:limit_of_inverse_polynomials_is_negligible} to regularize away all the remaining terms.
		\end{enumerate}   
	\end{proof}
	
	The following proposition allows us to recover the unbounded relative entropy of classical states when certain conditions -- especially a polynomial dimension -- are met. We will use it later as a proof tool.
	\begin{proposition}[Computational relative entropy for polynomial size states]\label{prop:computational_relative_entropy_poly_size_and_variance}
		Consider two classical states $p_n, q_n \in \calH_n$ with dimension $d_n \leq \poly(n)$ such that their relative entropy variance is polynomially bounded $V(p_n \fatpipe q_n) \leq O(\poly(n))$. Then,
		\begin{align}
			\underline{D}(p_n \fatpipe q_n) \simeq D(p_n \fatpipe q_n).
		\end{align}
	\end{proposition}
	\begin{proof}
		The converse bound $\underline{D}(p_n \fatpipe q_n) \lesssim D(p_n \fatpipe q_n)$ follows from the standard relative entropy converse for the hypothesis testing relative entropy~\cite[Proposition 7.70]{khatri2024principles}
		\begin{align}
			\underline{D}(p_n \fatpipe q_n) &\simeq\flim_{\epsilon \to 0} \flim_{\ell \to \infty} \fliminf_{k\to\infty} \frac{1}{n^k} D_h^{\epsilon}(p_n^{\otimes n^k}\fatpipe q_n^{\otimes n^k}; n^{k\ell}) \\
			\nonumber
			&\lesssim \flim_{\epsilon \to 0} \flim_{\ell \to \infty} \fliminf_{k\to\infty} \frac{1}{n^k} D_h^{\epsilon}(p_n^{\otimes n^k}\fatpipe q_n^{\otimes n^k})\\
			\nonumber
			&\lesssim \flim_{\epsilon \to 0} \flim_{\ell \to \infty} \fliminf_{k\to\infty} \frac{1}{n^k} \frac{1}{1-\epsilon}\left[ D(p_n^{\otimes n^k}\fatpipe q_n^{\otimes n^k}) + h_2(\epsilon)\right] \\
			\nonumber
			&\lesssim \flim_{\epsilon \to 0} \flim_{\ell \to \infty} \fliminf_{k\to\infty} \frac{1}{1-\epsilon}\left[ D(p_n\fatpipe q_n) + \frac{h_2(\tfrac{1}{2})}{n^k} \right] \\
			\nonumber
			&\lesssim \flim_{\epsilon \to 0} \frac{1}{1-\epsilon} D(p_n\fatpipe q_n) \\
			\nonumber
			&\simeq D(p_n\fatpipe q_n),
		\end{align}       
		where the first inequality follows by removing the gate restrictions and we have used  \cref{lemma:limit_of_inverse_polynomials_is_negligible} to treat the term proportional to $h_2(\tfrac{1}{2})$.
		For the achievability part, we expand the computational relative entropy again
		\begin{align}
			\underline{D}(p_n \fatpipe q_n) &=\flim_{\epsilon \to 0} \flim_{\ell \to \infty} \fliminf_{k\to\infty} \frac{1}{n^k} D_h^{\epsilon}(p_n^{\otimes n^k}\fatpipe q_n^{\otimes n^k}; n^{k\ell}). 
		\end{align}
		As the states $p_n$ and $q_n$ have polynomial support we can implement a likelihood ratio test using a polynomial size lookup table as detailed in \cref{proposition:efficient_likelihood_ratio_testing}. In the following we assume that $\ell$ is sufficiently large to be able to implement the likelihood ratio test. We want to use $S(\delta)$ from \cref{proposition:efficient_likelihood_ratio_testing} as a candidate in the optimization of $D_h^{\epsilon}(p_n^{\otimes n^k}\fatpipe q_n^{\otimes n^k}; n^{k\ell})$. To do so, we have to choose $\delta$ such that
		\begin{align}
			\frac{1}{n^k}\frac{V(p_n \fatpipe q_n)}{(\delta - \negl(n))^2} = \epsilon \ \Leftrightarrow \ 
			\delta = \sqrt{\frac{V(p_n \fatpipe q_n)}{n^k \epsilon}} + \negl(n).
		\end{align}
		This means we obtain a lower bound
		\begin{align}
			\frac{1}{n^k}D_h^{\epsilon}(p_n^{\otimes n^k}\fatpipe q_n^{\otimes n^k}; n^{k\ell}) &\geq D(p_n \fatpipe q_n)
			- \sqrt{\frac{V(p_n \fatpipe q_n)}{n^k \epsilon}} - \negl(n).
		\end{align}
		We now use the assumption that the relative entropy variance is polynomially bounded, hence there exists a $q \in \bbN$ such that $V(p_n \fatpipe q_n) \leq n^q$ for sufficiently large $n$. We insert this and obtain the bound
		\begin{align}
			\underline{D}(p_n \fatpipe q_n) &\gtrsim \flim_{\epsilon \to 0} \flim_{\ell \to \infty} \fliminf_{k\to\infty} D(p_n\fatpipe q_n) - n^{q - k/2}\sqrt{\frac{1}{\epsilon}} - \negl(n) \\
			&\gtrsim D(p_n\fatpipe q_n),
		\end{align}
		where we have concluded by \cref{lemma:limit_of_inverse_polynomials_is_negligible}.
	\end{proof}
	
	\section{Symmetric hypothesis testing under complexity limitations: Computational total variation distance}
	Symmetric hypothesis testing is a primitive that is of similar importance as asymmetric hypothesis testing. The success probability of a symmetric hypothesis test is closely related to the total variation distance, which in turn is related to the computational relative entropy via Pinsker's inequality.

	The aim of this section is to define a computational analog of the total variation distance and show that it relates to the computational relative entropy via a computational Pinsker's inequality. This showcases that we can build a self-consistent framework of computational quantities that relate to each other in ways that we expect from asymptotic information theory.
	Let us start with a complexity-limited version of total variation distance.
	
	\begin{definition}[$G$-complexity total variation distance]\label{def:g_complexity_symmetric_hypothesis_testing_error}
		Consider two quantum states $\rho, \sigma \in \calH$. Their $G$-complexity total variation distance is defined as
		\begin{align}
			{d}_{\mathrm{TV}}(\rho, \sigma; G) &\coloneqq \max_{\Lambda \in \calQ(\calH; G)} \Tr[ \Lambda (\rho-\sigma)].
		\end{align}
	\end{definition}
	Similarly to the situation for the $G$-complexity hypothesis testing relative entropy as defined in \cref{def:g_complexity_hypothesis_testing_relative_entropy}, the maximum above is guaranteed to always exist.
	Now optimizing over all polynomial numbers of gates yields the computational total variation distance.
	\begin{definition}[Computational total variation distance]
		\label{definition:computational_total_variation_distance}
		Consider two quantum states $\rho_n, \sigma_n \in \calH_n$. Their \emph{computational total variation distance} is defined as
		\begin{align}
			\underline{d}_{\mathrm{TV}}(\rho_n, \sigma_n) &\colonsimeq \flim_{\ell \to \infty} \underline{d}_{\mathrm{TV}}(\rho_n, \sigma_n; n^{\ell}).
		\end{align}
	\end{definition}
	The limit in $\ell$ is guaranteed to exist as the $G$-complexity total variation distance can only increase with the number of available gates. This monotonicity is then enough to guarantee the existence of the limit as argued in \cref{sec:limits_of_polynomial_resources}.
	We note that the above definition is a single-copy definition, contrary to how we have defined the computational relative entropy. While we could employ regularization to define a computational analog of the Chernoff divergence, we are more interested in the single-shot quantity at this point because we can relate it to the computational relative entropy by giving computational analogues of Pinsker's inequality and the Bretagnolle-Huber inequality. 
	\begin{theorem}[Computational Pinsker's inequality]\label{theorem:computational_pinskers_inequality}
		Consider two states $\rho_n, \sigma_n \in \calH_n$. Then,
		\begin{align}
			\underline{d}_{\mathrm{TV}}(\rho_n, \sigma_n) \lesssim \sqrt{\tfrac{1}{2}\underline{D}(\rho_n \fatpipe \sigma_n)}.
		\end{align}
	\end{theorem}
	\begin{theorem}[Computational Bretagnolle-Huber inequality]\label{theorem:computational_bretagnolle_huber}
		Consider two states $\rho_n, \sigma_n \in \calH_n$. Then,
		\begin{align}
			\underline{d}_{\mathrm{TV}}(\rho_n, \sigma_n) \lesssim \sqrt{1 - \exp(-\underline{D}(\rho_n \fatpipe \sigma_n))}.
		\end{align}
	\end{theorem}
	\begin{proof}
		Pinsker's inequality is proven by using data processing under the optimal measurement that achieves the trace distance between the two states. 
		Let $\Lambda_n$ denote the POVM effect that optimizes the complexity limited total variation distance
		\begin{align}
			d_{\mathrm{TV}}(\rho_n, \sigma_n; n^\ell) = \Tr[\Lambda_n \sigma_n] - \Tr[ \Lambda_n \rho_n]
		\end{align}
		for some fixed exponent $\ell$. 
		The outcome distributions for $\rho_n$ and $\sigma_n$ under this measurement are 
		\begin{align}
			\calM_{n}[\rho_n] &= \begin{pmatrix}
				p_n & 0 \\
				0 & 1-p_n
			\end{pmatrix}, \quad
			\calM_{n}[\sigma_n] = \begin{pmatrix}
				q_n & 0 \\
				0 & 1-q_n
			\end{pmatrix},
		\end{align}
		such that $d_{\mathrm{TV}}(\rho_n, \sigma_n; n^{\ell})= q_n - p_n$.
		The measurement by definition has complexity $C(\calM_n) \leq n^{\ell}$. This means we can use the data processing under measurement channels of the computational relative entropy. We have
		\begin{align}
			\underline{D}(\rho_n \fatpipe \sigma_n) &\gtrsim
			\underline{D}(\calM[\rho_n] \fatpipe \calM[\sigma_n]) \simeq \underline{D}\left(\left. \begin{pmatrix}
				p_n & 0 \\
				0 & 1-p_n
			\end{pmatrix} \, \right\| \, \begin{pmatrix}
				q_n & 0 \\
				0 & 1-q_n
			\end{pmatrix} \right). 
		\end{align}
		To make sure these states are well-behaved, we further apply a depolarizing noise channel $\calD_{\lambda_n}(p_n) = (1-\lambda_n) p_n + \lambda_n \omega_n$. As the depolarizing noise channel can be efficiently implemented using a randomized bit flip, we can again use data processing under polynomially bounded channels. We have
		\begin{align}
			\underline{D}(\rho_n \fatpipe \sigma_n) &\gtrsim\underline{D}((\calD_{\lambda_n} \circ \calM)[\rho_n] \fatpipe (\calD_{\lambda_n} \circ \calM)[\sigma_n]) \\
			\nonumber
			&\simeq  \underline{D}\left(\left. \begin{pmatrix}
				(1-\lambda_n)p_n + \lambda_n/2 & 0 \\
				0 & (1-\lambda_n)(1-p_n) + \lambda_n/2
			\end{pmatrix} \, \right\| \, \begin{pmatrix}
				(1-\lambda_n) q_n + \lambda_n/2 & 0 \\
				0 & (1-\lambda_n)(1-q_n) + \lambda_n/2
			\end{pmatrix} \right)\\
			\nonumber
			&\simeq  \underline{D}\left(\left. \begin{pmatrix}
				p_n(\lambda_n) & 0 \\
				0 & 1-p_n(\lambda_n)
			\end{pmatrix} \, \right\| \, \begin{pmatrix}
				q_n(\lambda_n) & 0 \\
				0 & 1-q_n(\lambda_n)
			\end{pmatrix} \right),
			\nonumber
		\end{align}        
		where we have defined $p_n(\lambda_n) \coloneqq (1-\lambda_n)p_n + \lambda_n/2$ and $q_n(\lambda) \coloneqq (1-\lambda_n)q_n + \lambda_n/2$.
		We are now looking at the computational relative entropy for two classical states of constant support relative to $n$. 
		We would like to apply \cref{prop:computational_relative_entropy_poly_size_and_variance}, which requires that the relative entropy variance is polynomially bounded. In our case, we have that
		\begin{align}
			V&\left(\left. \begin{pmatrix}
				p_n(\lambda_n) & 0 \\
				0 & 1-p_n(\lambda_n)
			\end{pmatrix} \, \right\| \, \begin{pmatrix}
				q_n(\lambda_n) & 0 \\
				0 & 1-q_n(\lambda_n)
			\end{pmatrix} \right) \\
			\nonumber
			&= p_n(\lambda_n)(1-p_n(\lambda_n))\left[\log \frac{p_n(\lambda_n)}{q_n(\lambda_n)} - \log \frac{1-p_n(\lambda_n)}{1-q_n(\lambda_n)} \right]^2 \\
			\nonumber
			&\leq 16 \max \left\{ \log \frac{1}{q_n(\lambda_n)}, \log \frac{1}{1-q_n(\lambda_n)},\log \frac{1}{p_n(\lambda_n)},\log \frac{1}{1-p_n(\lambda_n)}\right\}^2 \\
			\nonumber
			&\leq 16\log^2 \frac{2}{\lambda_n}.
			\nonumber
		\end{align}
		For the first inequality we simply bounded $p_n(\lambda_n)(1-p_n(\lambda)) \leq 1$ and expanded the logarithms in the brackets, yielding four terms which we bounded by the maximal one.
		If we choose $\lambda_n = e^{-n}$, we are polynomially bounded and can apply \cref{prop:computational_relative_entropy_poly_size_and_variance} and the classical Pinsker's inequality to obtain
		\begin{align}
			\underline{D}(\rho_n \fatpipe \sigma_n) &\gtrsim D\left(\left. \begin{pmatrix}
				p_n(\lambda_n) & 0 \\
				0 & 1-p_n(\lambda_n)
			\end{pmatrix} \, \right\| \, \begin{pmatrix}
				q_n(\lambda_n) & 0 \\
				0 & 1-q_n(\lambda_n)
			\end{pmatrix} \right) \\
			\nonumber
			&\geq 2 d_{\mathrm{TV}}\left(\left. \begin{pmatrix}
				p_n(\lambda_n) & 0 \\
				0 & 1-p_n(\lambda_n)
			\end{pmatrix} \, \right\| \, \begin{pmatrix}
				q_n(\lambda_n) & 0 \\
				0 & 1-q_n(\lambda_n)
			\end{pmatrix} \right)^2 \\
			\nonumber
			&= 2 (p_n(\lambda_n) - q_n(\lambda_n))^2 \\
			\nonumber
			&= 2 (1-\lambda_n)^2 (p_n - q_n) \\
			\nonumber
			&= 2 (1-\lambda_n)^2 d_{\mathrm{TV}}(\rho_n \fatpipe \sigma_n; n^\ell)^2 \\
			\nonumber
			&\geq 2d_{\mathrm{TV}}(\rho_n \fatpipe \sigma_n; n^\ell)^2 - 4 \lambda_n d_{\mathrm{TV}}(\rho_n \fatpipe \sigma_n; n^\ell)\\
			\nonumber
			&\geq 2d_{\mathrm{TV}}(\rho_n \fatpipe \sigma_n; n^\ell)^2 - 4 \lambda_n\\
			\nonumber
			&\geq 2d_{\mathrm{TV}}(\rho_n \fatpipe \sigma_n; n^\ell)^2 - 4 e^{-n} \\
			\nonumber
			&\simeq 2d_{\mathrm{TV}}(\rho_n \fatpipe \sigma_n; n^\ell)^2.
			\nonumber
		\end{align}        
		Here, we have used  that, by construction, $q_n - p_n = d_{\mathrm{TV}}(\rho_n \fatpipe \sigma_n; n^\ell) \leq 1$. The last step uses that $e^{-n}$ is a negligible function.
		As the above holds for any $\ell$, we can take the limit $\ell\to\infty$ and obtain the statement of the theorem by rearranging. The proof for the computational Bretagnolle-Huber inequality proceeds analogously, we only use the classical Bretagnolle-Huber inequality in the last step instead of Pinsker's inequality.
	\end{proof}
	
	\section{Computational smoothing}\label{section:computational_smoothing}
	Computational indistinguishability is one of the core tenets of computational information theory. Two distributions or quantum states are computationally indistinguishable if any complexity limited algorithm that tries to distinguish them has negligible success probability. We formally define it in the language of our work as follows.
	\begin{definition}[Computational indistinguishability]\label{def:computational_indistinguishability}
		We say two quantum states $\rho_n, \sigma_n \in \calH_n$ are \emph{computationally indistinguishable}, or $\rho_n \approx_c \sigma_n$, if for all polynomial numbers of copies $m \leq O(\poly(n))$ and polynomial complexity POVM effects $C(\Lambda) \leq O(\poly(n))$ we have a negligible difference of the associated probabilities:
		\begin{align}
			|\Tr[\Lambda \rho_n^{\otimes m}] - \Tr[ \Lambda \sigma_n^{\otimes m}]| \simeq 0.
		\end{align}
	\end{definition}	
	A central property of the computational relative entropy we develop is that it remains the same if the first argument is replaced with a state that is computationally indistinguishable. To establish this, we need the following lemma.
	\begin{lemma}[Single-shot computational smoothing]\label{lemma:single_shot_computational_smoothing}
		Consider two states $\rho, \sigma \in \calH$. Then, for any $\tilde{\rho}$ we set
		\begin{align}
			\delta \coloneqq d_{\mathrm{TV}}(\rho, \tilde{\rho}; G).
		\end{align}
		We have that
		\begin{align}
			D_h^{\epsilon}(\rho \fatpipe \sigma; G) \leq D_h^{\epsilon + \delta}(\tilde{\rho} \fatpipe \sigma; G).
		\end{align}
	\end{lemma}
	\begin{proof}
		Let $\Lambda$ be the $G$-complexity POVM effect achieving the optimization
		\begin{align}
			D_h^{\epsilon}(\rho \fatpipe \sigma; G) = -\log \min_{\Lambda \in \calQ(\calH; G)} \big\{ \Tr[ \Lambda \sigma] \pipe \Tr[\Lambda \rho] \geq 1 - \epsilon \big\}.
		\end{align}
		Then, we have that
		\begin{align}
			\Tr[ \Lambda \tilde\rho] &= \Tr[\Lambda \rho] + \Tr[\Lambda (\tilde\rho - \rho)] \\
			\nonumber
			&\geq 1-\epsilon - \max_{\Lambda \in \calQ(\calH;G)} | \Tr[\Lambda (\tilde\rho - \rho)]| \\
			\nonumber
			&= 1 - \epsilon - d_{\mathrm{TV}}(\rho, \tilde\rho; G) \\
			\nonumber
			&= 1-\epsilon - \delta.
		\end{align}
		Hence, $\Lambda$ is an admissible POVM effect in the optimization of $D_h^{\epsilon + \delta}(\tilde\rho \fatpipe \sigma; G)$, implying the claimed bound.
	\end{proof}
	The above lemma then directly implies the desired result.
	\begin{theorem}[Computational smoothing of computational relative entropy]\label{theorem:undetectability_of_computational_smoothing}
		Consider two states $\rho_n, \sigma_n \in \calH_n$. Then, 
		\begin{align}
			\underline{D}(\rho_n \fatpipe \sigma_n) \simeq \inf_{\Tilde{\rho}_n \approx_c \rho_n} \underline{D}(\tilde{\rho}_n \fatpipe \sigma_n).
		\end{align}
	\end{theorem}
	\begin{proof}
		Let $\tilde\rho_n$ denote any state that is computationally indistinguishable from $\rho_n$, $\tilde\rho_n \approx_c \rho_n$. 
		The lower bound
		\begin{align}
			\inf_{\tilde\rho_n \approx_c \rho_n}\underline{D}(\tilde{\rho}_n \fatpipe \sigma_n) \lesssim \underline{D}(\rho_n \fatpipe \sigma_n)
		\end{align}
		is trivial as $\rho_n$ is computationally indistinguishable from itself and hence a candidate in the minimization. 
		For the upper bound, we expand the definition of the computational relative entropy and apply \cref{lemma:single_shot_computational_smoothing}:
		\begin{align}
			\underline{D}({\rho}_n \fatpipe \sigma_n) &\simeq \flim_{\epsilon \to 0} \flim_{\ell \to \infty} \fliminf_{k\to \infty} D_h^{\epsilon}(\rho_n^{\otimes n^k} \fatpipe \sigma_n^{\otimes n^k}; n^{k\ell}) \\
			\nonumber
			&\lesssim  \flim_{\epsilon \to 0} \flim_{\ell \to \infty} \fliminf_{k\to \infty} D_h^{\epsilon + \delta_n}(\tilde\rho_n^{\otimes n^k} \fatpipe \sigma_n^{\otimes n^k}; n^{k\ell}),
		\end{align}
		where $\delta_n$ is a negligible function of $n$ by definition of computational indistinguishability.
		For sufficiently large $n$, we have $\delta_n \leq \epsilon$. 
		As we only care about asymptotics, we can always choose $n$ sufficiently large such that this is the case, which means we have
		\begin{align}
			\underline{D}({\rho}_n \fatpipe \sigma_n) &\lesssim \flim_{\epsilon \to 0} \flim_{\ell \to \infty} \fliminf_{k\to \infty} D_h^{2 \epsilon}(\tilde\rho_n^{\otimes n^k} \fatpipe \sigma_n^{\otimes n^k}; n^{k\ell}) \simeq \underline{D}(\tilde\rho_n\fatpipe \sigma_n)
		\end{align}
		as desired. As the above holds for any computationally indistinguishable state, it holds for the infimum which exists in the sense explained in \cref{sec:limits_of_polynomial_resources}.
	\end{proof}	
	The importance of computational smoothing cannot be overstated. It allows us to connect to established constructions in classical and quantum cryptography and is the workhorse for us proving separations between computational and unbounded quantities. We also note that any computational information theoretic quantity derived from the computational relative entropy will inherit computational smoothing, which makes it a core tenet of our approach to computational quantum information theory.
	
	\subsection*{Separations between computational and unbounded relative entropy}
	Computational smoothing implies that computationally indistinguishable states are also indistinguishable in computational relative entropy.
	\begin{corollary}[Zero rate for computationally indistinguishable states]\label{corollary:zero_rate_for_computationally_indistinguishable_states}
		Consider two computationally indistinguishable states $\rho_n \approx_c \tilde\rho_n$. Then, 
		\begin{align}
			\underline{D}(\tilde\rho_n \fatpipe \rho_n) \simeq 0.
		\end{align}
	\end{corollary}
	\begin{proof}
		To prove this statement, we apply \cref{theorem:undetectability_of_computational_smoothing} and \cref{item:computational_relative_entropy_boundedness} to obtain
		\begin{align}
			0 \lesssim \underline{D}(\tilde\rho_n \fatpipe \rho_n) \simeq\underline{D}(\rho_n \fatpipe \rho_n) \lesssim D(\rho_n \fatpipe \rho_n) = 0.
		\end{align}
	\end{proof}
	We will make use of the above corollary to establish separations between the computational and the unbounded relative entropy. To do so, we construct different families of quantum states that are computationally indistinguishable under different assumptions and have different properties.
	
	We start with an unconditional result that requires exponential complexity of the involved states.
	\begin{lemma}[Unconditionally indistinguishable states]\label{lemma:quantum-hard-exponential-complexity}
		With no computational assumptions, there exists two quantum $n$-qubit states $\rho_n, \sigma_n \in (\bbC^2)^{\otimes n}$ which are computationally indistinguishable (see \cref{def:computational_indistinguishability}) but whose supports are not mutually contained inside each other.        
	\end{lemma}
	\begin{proof}
		Let us pick the states to be pure states $\rho_n = |\psi_n\rangle\!\langle \psi_n|, \sigma_n = |\phi_n \rangle\!\langle \phi_n|$ sampled according to the Haar measure $\ket{\psi_n}, \ket{\phi_n} \sim \haar(n)$. 
		By Levy's lemma we have that for an arbitrary POVM effect $\Lambda$ and number of copies $k$~\cite{mele2024introduction}
		\begin{equation}
			\operatornamewithlimits{\bbP}_{\psi \sim \haar(n)}\big[| \Tr[\psi^{\otimes k}\Lambda] - \mathbb{E}\{\Tr\left[ \psi^{\otimes k}\Lambda \big]\} | \geq \delta  \right] \leq 2 \exp\left(- \frac{d \delta^2}{9\pi^3 k^2}\right),
		\end{equation}
		where $d = 2^n$.
		This follows from the Lipschitz constant bound
		\begin{equation}
			| \Tr[\psi^{\otimes k}\Lambda] - \Tr[\phi^{\otimes k}\Lambda]| \leq k \|\ket{\psi}-\ket{\phi}\|_2.
		\end{equation}
		We now consider an $\epsilon$-net $\calN$ over computationally efficient POVM effects in operator norm. We seek to establish that such a net has a size of at most $\exp(\poly(n))$.    
		Consider as a start a POVM with gate complexity $G$ according to our circuit model. To cover the set of unitaries circuits of gate complexity $G$ in diamond norm, we need at most~\cite[Eq.~(B4)]{zhao2024learning}
		\begin{align}
			\exp\left( 32 G \log \frac{8 G}{\epsilon} + 2G \log 2G\right)
		\end{align}
		many elements in the net, where the factor $\log 2G$ in the last term comes from the circuit being able to at most touch $2G$ many qubits. Finally, to cover the whole set of POVM effects, we need to count the different measurements we can perform after the unitary was applied. In our circuit model, this is at most $2^{2G} = \exp(2G \log 2)$, as for any of the at most $2G$ qubits we have two possibilities to choose from, either projecting on the zero state or ignoring. Hence, we obtain an $\epsilon$-net in operator norm for the set of efficient POVM effects using at most
		\begin{align}
			\exp\left( 32 G \log \frac{8 G}{\epsilon} + 2G \log 2G + 2G \log 2\right) = \exp\left(O\left(G \log \frac{G}{\epsilon}\right)\right)
		\end{align}
		many elements. As $G \leq O(\poly(n))$ always, we find that the overall number of elements is indeed at most $\exp\left(O\left(\poly n \log \frac{1}{\epsilon}\right)\right)$.
		
		Using the union bound, we thus have
		\begin{align}
			\operatornamewithlimits{\bbP}_{\psi \sim \haar(n)}\big[| \Tr[\psi^{\otimes k}\Lambda] - \mathbb{E}\{\Tr\left[ \psi^{\otimes k}\Lambda \big]\} | \geq \delta \text{ for all } \Lambda \in \calN  \right] \leq 2 \exp\left( -\frac{d \delta^2}{9\pi^3 k^2} + O\left(\poly(n) \log \frac{1}{\epsilon}\right) \right).
		\end{align}
		As $\calN$ is an $\epsilon$-net, the above holding implies that 
		\begin{align}
			\operatornamewithlimits{\bbP}_{\psi \sim \haar(n)}\big[| \Tr[\psi^{\otimes k}\Lambda] - \mathbb{E}\{\Tr\left[ \psi^{\otimes k}\Lambda \big]\} | \geq \delta - \epsilon \text{ for all } C(\Lambda)\leq O(\poly(n))  \right] \leq 2 \exp\left( -\frac{d \delta^2}{9\pi^3 k^2} + O\left(\poly(n) \log \frac{1}{\epsilon}\right) \right).
		\end{align}
		Let us now choose $\delta = 2 d^{-1/4}$ and $\epsilon = d^{-1/4}$ to obtain
		\begin{align}
			\operatornamewithlimits{\bbP}_{\psi \sim \haar(n)}\big[| \Tr[\psi^{\otimes k}\Lambda] - \mathbb{E}\{\Tr\left[ \psi^{\otimes k}\Lambda \big]\} | \geq d^{-1/4} \text{ for all } C(\Lambda)\leq O(\poly(n))  \right] \leq 2 \exp\left( -\frac{d^{1/2}}{9\pi^3 k^2} + O\left(\poly(n) \log d\right) \right).
		\end{align}
		As $d = 2^n$ and $k \leq O(\poly(n))$ we find that a deviation that is more than negligible is doubly-exponentially unlikely for Haar random states.
		
		This especially implies that with high probability, we can find two different Haar random states such that for any efficient POVM the measurement error probability is close to the corresponding Haar average and by the triangle inequality this also means that the corresponding expectation values need to be close.
		This concludes the proof.
	\end{proof}
	
	We continue with a construction that achieves computational indistinguishability for polynomial-complexity states, but which requires a cryptographic assumption.
	\begin{lemma}[Efficiently preparable indistinguishable states from one-way functions]\label{lem:quantum-hard-one-way}
		Assuming the existence of quantum-hard one-to-one one-way functions, there exist two classical $n$-qubit states $\rho_n, \sigma_n \in (\bbC^{2})^{\otimes n}$ which are computationally indistinguishable (see \cref{def:computational_indistinguishability}) and at the same time efficiently preparable, i.e.,
		\begin{align}
			\rho_n \approx_c \sigma_n \text{, but }C(\rho_n), C(\sigma_n) \leq O(\poly(n)).
		\end{align}
		Furthermore, the states $\rho_n$ and $\sigma_n$ have disjoint support and are thus perfectly distinguishable by unbounded observers.
	\end{lemma}
	We remark that a version of the above lemma would directly follow from the quantum hardness of the decision version of the Learning With Errors (LWE) problem, i.e., distinguishing polynomially many samples of an LWE generator from uniform samples. The statement we give is more general, as it would also be valid if one would come up with a quantum algorithm with LWE and if a different quantum-hard construction for a one-to-one one-way function would exist.
	\begin{proof}
		We consider bitstrings $x \in \{0,1\}^n$ and for given functions $f_n\colon \{0,1\}^n\to \{0,1\}^n$, $b_n\colon \{0,1\}^n\to \{0,1\}$ define the set
		\begin{align}
			F_n \coloneqq \big\{x \in \{0,1\}^n \bigpipe b_n(x) = 1 \big\}
		\end{align}
		and the associated classical states
		\begin{align}
			\rho_n &\coloneqq \frac{1}{|F_n|} \sum_{x \in F_n} \ket{f_n(x)}\!\!\bra{f_n(x)} \\
			\sigma_n &\coloneqq \frac{1}{2^n-|F_n|} \sum_{x \not\in F_n} \ket{f_n(x)}\!\!\bra{f_n(x)}
		\end{align}
		We note that these states are easily distinguishable by the POVM effect $M_n = |F_n| \rho_n$ which achieves $\Tr[ M_n (\rho_n-\sigma_n)] = 1$ whenever $f_n$ is one-to-one.
		
		We now proceed to construct $f_n$ and $b_n$ such that $\rho_n$ and $\sigma_n$ are hard to distinguish. Given a quantum-hard one-to-one one-way function $g_m$ defined for all input sizes $m$, we use the Goldreich-Levin construction~\cite{Rubinfeld2012}. We consider without loss of generality inputs $x$ of even size, such that we can split them in half $x = (x_0, x_1)$ and define
		\begin{align}
			f_n(x) &\coloneqq (g_{n/2}(x_0),x_1) \\
			b_n(x) &\coloneqq \langle x_0, x_1 \rangle \text{ mod } 2,
		\end{align}
		where $\langle x_0, x_1\rangle$ is the standard inner product.
		The Goldreich-Levin Theorem~\cite{Rubinfeld2012} states that the function $b_n$ forms a \emph{hardcore predicate} or \emph{hardcore bit} for the one-way function $f_n$, meaning that for any quantum polynomial-time algorithm $\mathcal{A}$, we have that
		\begin{equation}\label{eq:prob-success}
			\operatornamewithlimits{\bbP}_{x \sim \mathrm{Uniform}} \big[\mathcal{A}(f_n(x)) = b_n(x)\big] \leq \frac{1}{2} + \negl(n) \ \Leftrightarrow \ \operatornamewithlimits{\bbP}_{x \sim \mathrm{Uniform}} \big[\mathcal{A}(f_n(x)) = b_n(x)\big]\lesssim \frac{1}{2}.
		\end{equation}
		Here, we recall the notation \smash{$\lesssim$} from \cref{sec:limits_of_polynomial_resources}.
		
		Let us first show that for any efficient single-copy POVM $\Lambda_n$ we have $\Tr[ \Lambda_n (\rho_n-\sigma_n)] \simeq 0$. We take the algorithm $\mathcal{A}$ to be the outcome of the POVM $\Lambda_n$ applied on any given input $\ket{f_n(x)}\!\!\bra{f_n(x)}$. We note that due to the definition of $b_n$ as a parity function, whenever $x_1\neq 0^n$, $b_n$ is perfectly balanced, meaning that $|F_n|/2^n \approx 1/2$ up to an exponentially small correction. This allows us to rewrite: 
		\begin{align}
			\rho_n - \sigma_n &= \frac{1}{2^n} \left( \frac{2^n}{|F_n|}  \sum_{x \in F_n} \ket{f_n(x)}\!\!\bra{f_n(x)} - \frac{2^n}{2^n-|F_n|}  \sum_{x \not\in F_n} \ket{f_n(x)}\!\!\bra{f_n(x)} \right)\label{eq:rhon-sigman}\\
			&\simeq \frac{2}{2^n}\left(\sum_{x \in F_n} \ket{f_n(x)}\!\!\bra{f_n(x)} - \sum_{x \not\in F_n} \ket{f_n(x)}\!\!\bra{f_n(x)}\right).\nonumber
		\end{align}
		In the last line, we have made use of the notation $\simeq$ to indicate equality up to negligible operators as per \cref{definition:negligible_operators}.
		On the other hand, we can also rewrite the probability of success of $\mathcal{A}$ as:
		\begin{align}
			\operatornamewithlimits{\bbP}_{x \sim \mathrm{Uniform}} \big[\mathcal{A}(f_n(x)) = b_n(x)\big] &= \frac{1}{2^n} \sum_x \bbP \big[\mathcal{A}(f_n(x)) = b_n(x)\big]\\
			&= \frac{1}{2^n}\left(\sum_{x\in F_n} \Tr[\Lambda_n \ket{f_n(x)}\!\!\bra{f_n(x)}] + \sum_{x\not\in F_n} \Tr[(\bbI-\Lambda_n) \ket{f_n(x)}\!\!\bra{f_n(x)}] \right) \nonumber\\
			&\simeq \frac{1}{2} + \frac{1}{2^n}\left(\sum_{x\in F_n} \Tr[\Lambda_n \ket{f_n(x)}\!\!\bra{f_n(x)}] - \sum_{x\not\in F_n} \Tr[\Lambda_n \ket{f_n(x)}\!\!\bra{f_n(x)}] \right)\nonumber\\
			&\simeq \frac{1}{2} + \frac{1}{2} \Tr[\Lambda_n (\rho_n - \sigma_n)],\nonumber
		\end{align}
		where we plugged in \cref{eq:rhon-sigman} in the last equation. Combining with \cref{eq:prob-success} then lets us conclude that
		\begin{align}
			\operatornamewithlimits{\bbP}_{x \sim \mathrm{Uniform}} \big[\mathcal{A}(f_n(x)) = b_n(x)\big] \simeq \frac12 + \frac12\Tr[ \Lambda_n (\rho_n-\sigma_n)] \lesssim \frac12
		\end{align}
		which means that
		\begin{align}
			\Tr[ \Lambda_n (\rho_n-\sigma_n)] \simeq 0
		\end{align}
		as desired.
		
		From here, we can move to the general $k$-copy statement. Consider a POVM $\Lambda_n$ that acts on $k \leq O(\poly(n))$ copies of either $\rho_n$ or $\sigma_n$. By recursion, and linearity of the trace, we can rewrite: 
		\begin{align}
			\Tr[ \Lambda_n (\rho_n^{\otimes k}-\sigma_n^{\otimes k})] &= \Tr[ \Lambda_n (\rho_n^{\otimes k}-\rho_n\otimes\sigma_n^{\otimes k-1})] + \Tr[ \Lambda_n (\rho_n\otimes\sigma_n^{\otimes k-1}-\sigma_n^{\otimes k})]\label{eq:recursive-POVM}\\
			&= \Tr[ \Lambda_n (\rho_n^{\otimes k}-\rho_n^{\otimes 2}\otimes\sigma_n^{\otimes k-2})] + \Tr[ \Lambda_n (\rho_n^{\otimes 2}\otimes\sigma_n^{\otimes k-2}-\rho_n\otimes\sigma_n^{\otimes k-1})] + \Tr[ \Lambda_n (\rho_n\otimes\sigma_n^{\otimes k-1}-\sigma_n^{\otimes k})] \nonumber \\
			&= \cdots \nonumber\\
			&= \sum_{i=0}^{k-1} \Tr [\Lambda_n (\rho_n^{\otimes k-i}\otimes\sigma_n^{\otimes i}-\rho_n^{\otimes k-i-1}\otimes\sigma_n^{\otimes i+1})] \nonumber
		\end{align}
		where now, for each term in the sum, $\rho_n^{\otimes k-i}\otimes\sigma_n^{\otimes i}$ only differs from $\rho_n^{\otimes k-i-1}\otimes\sigma_n^{\otimes i+1}$ on the $(k-i)$-th register.
		For every $i \in \{0, \ldots, k\}$, consider now the POVM $\Lambda_{n,i}$ that takes in a single copy of a state $\phi_n$ (that is either $\rho_n$ or $\sigma_n$), first creates the state $\rho_n^{\otimes k-i-1}\otimes \phi_n \otimes \sigma_n^{\otimes i}$ by appending to $\phi_n$ copies of $\rho_n$ and $\sigma_n$, then applies $\Lambda_n$ on this state. $\Lambda_{n,i}$ is efficiently computable since $\Lambda_n$ is assumed to be so, and the states $\rho_n$ and $\sigma_n$ are also efficiently preparable (one-way functions are efficient to compute in the forward direction, by definition, and it is easy to sample a random $x$ that either satisfies $b_n(x) = 0$ or $b_n(x) = 1$).  Therefore, from \cref{eq:recursive-POVM}, we get:
		\begin{equation}
			\Tr[ \Lambda_n (\rho_n^{\otimes k}-\sigma_n^{\otimes k})] = \sum_{i=0}^{k-1} \Tr[\Lambda_{n,i}(\rho_n-\sigma_n)] \leq k \cdot\negl(n) \simeq 0
		\end{equation}
		since we showed that any single-copy efficient POVM must have $\Tr[\Lambda_{n,i}(\rho_n-\sigma_n)] = \negl(n)$ in the first part of this proof and assumed $k\leq O(\poly(n))$ as required by computational indistinguishability.
		
		In a case where we care about uniformly selecting the POVMs for all $n$, one has to be more careful here: in a reasoning by contradiction, that is, assuming that we have a multi-copy POVM $\Lambda_n$ that achieves $\Tr[ \Lambda_n (\rho_n^{\otimes k}-\sigma_n^{\otimes k})] = 1/\poly(n)$, the last argument above does not tell us how to uniformly select a $\Lambda_{n,i}$ with similar performance. To deal with this scenario, we can take a final single-copy POVM $\tilde{\Lambda}_n$ that applies a POVM from $\{\Lambda_{n,i}\}_i$ uniformly at random. The performance of $\tilde{\Lambda}_n$ is then simply the average performance of the POVMs $\{\Lambda_{n,i}\}_i$:
		\begin{equation}
			\Tr[ \tilde{\Lambda}_n (\rho_n-\sigma_n)] = \frac{1}{k}\sum_{i=0}^{k-1} \Tr[\Lambda_{n,i}(\rho_n-\sigma_n)] = \frac{1}{k} \Tr[ \Lambda_n (\rho_n^{\otimes k}-\sigma_n^{\otimes k})].
		\end{equation}
		Assuming $\Tr[ \Lambda_n (\rho_n^{\otimes k}-\sigma_n^{\otimes k})] = 1/\poly(n)$ then leads to $\Tr[ \tilde{\Lambda}_n (\rho_n-\sigma_n)]=1/\poly(n)$ since $k=\poly(n)$, giving us the desired result.
	\end{proof}
	
	Finally, we present a family of classical states computationally indistinguishable using purely classical statistical tests, but distinguishable using a quantum computer. On a technical level, we use the explicit notation $\approx_{c, \mathrm{cl}}$ and $\approx_{c,\mathrm{qu}}$ to distinguish between computational indistinguishable for classical and quantum observers. 
	\begin{lemma}[Efficiently preparable states for quantum advantage from one-way functions]\label{lem:classically-hard-one-way}
		Assuming the existence of classically-hard, quantumly-easy, one-to-one one-way functions, there exist two classical $n$-qubit states $\rho_n, \sigma_n \in (\bbC^2)^{\otimes n}$ which are computationally indistinguishable (see \cref{def:computational_indistinguishability}) to classical observers but at the same time easily distinguishable for quantum observers and efficiently preparable, i.e.,
		\begin{align}
			\rho_n \approx_{c,\mathrm{cl}} \sigma_n \text{, but } \rho_n \not\approx_{c,\mathrm{qu}} \sigma_n \text{ and } C(\rho_n), C(\sigma_n) \leq O(\poly(n)).
		\end{align}
		More so, the states $\rho_n$ and $\sigma_n$ have disjoint support and there exists an efficient quantum POVM to perfectly distinguish $\rho_n$ from $\sigma_n$, i.e.,
		\begin{align}
			\underline{d}_{\mathrm{TV}}(\rho_n, \sigma_n) \simeq 1.
		\end{align}
	\end{lemma}
	A candidate for the one-way function in the above lemma would be modular exponentiation, as we know an efficient quantum algorithm to invert it (Shor's algorithm for the discrete logarithm), but no algorithm in $\textsf{BPP}$ or even in $\textsf{P/poly}$ is known.
	\begin{proof}
		We follow the same proof structure as in \cref{lem:quantum-hard-one-way} to show classical hardness, but consider instead a classically-hard, quantumly-easy, one-to-one one-way function for $g_m$. As for the quantum easiness, we construct a POVM effect $M_n$ as follows: measure the input state in the computational basis, giving a random sample $f_n(x) = (g_{n/2}(x_0),x_1)$, use the quantum algorithm that inverts $g_{n/2}$ (assumed to exist from the quantum easiness of the one-way function) to recover $x_0$ and finally classify the state by computing $b_n(x) = \langle x_0,x_1 \rangle \text{ mod } 2$. In polynomial time, the error probability of this algorithm can be boosted from any constant $> 1/2$ to $1-O(1/2^n)$ by repeated probabilistic computation of $x_0$ (then $b_n(x)$) until we find the right $x_0$ that gives $f_n(x)$ (or outputting a random bit otherwise).
	\end{proof}
	
	While we did not specifically mention this in our statements of \cref{lem:quantum-hard-one-way} and \cref{lem:classically-hard-one-way}, these results can each take two equally valid forms, whether we look into uniform or non-uniform constructions for the POVMs $\{\Lambda_n\}_n$. Each of these forms boils down to considering one-way functions that are hard to invert with either uniform complexity classes (such as $\textsf{HeurBPP},\ \textsf{HeurBQP}$) or non-uniform complexity classes (such as $\textsf{HeurP/poly},\ \textsf{HeurBQP/qpoly}$). Indeed, our reductions from the hardness of inverting one-way functions only make use of machinery in $\textsf{BPP}$, which preserves all of the computational classes mentioned above.
	
	Let us now synthesize the implications of the above lemmas for the computational relative entropy.
	\begin{theorem}[Computational separations in asymmetric hypothesis testing]\label{theorem:separations_computational_relative_entropy}
		We can exhibit the following separations between the relative entropy $D$ and the computational relative entropy $\underline{D}$:
		\begin{propenum}
			\item Without computational assumptions, there exists states $\rho_n, \sigma_n \in (\bbC^2)^{\otimes n}$ of exponential complexity such that
			\begin{align}
				D(\rho_n \fatpipe \sigma_n) = \infty
				\text{, but }
				\underline{D}(\rho_n \fatpipe \sigma_n) \simeq 0.
			\end{align}
			
			\item Assuming the existence of quantum-hard one-to-one one-way functions, there exist states $\rho_n, \sigma_n \in (\bbC^2)^{\otimes n}$ of polynomial complexity such that
			\begin{align}
				D(\rho_n \fatpipe \sigma_n) = \infty 
				\text{, but }
				\underline{D}(\rho_n \fatpipe \sigma_n) \simeq 0.
			\end{align}
			
			\item Assuming the existence of classically-hard quantumly-easy one-way functions, there exist states $\rho_n, \sigma_n \in (\bbC^2)^{\otimes n}$ of polynomial complexity such that 
			\begin{align}
				\underline{D}^{\mathrm{qu}}(\rho_n \fatpipe \sigma_n) \gtrsim \omega(\log n) 
				\text{, but }
				\underline{D}^{\mathrm{cl}}(\rho_n \fatpipe \sigma_n) \simeq 0.
			\end{align}
			Here, we make the quantum and classical nature of the employed tests explicit with the notation $\underline{D}^{\mathrm{cl}}$ and $\underline{D}^{\mathrm{qu}}$.
		\end{propenum}
	\end{theorem}
	\begin{proof}
		We provide separate proofs of each fact.
		\begin{enumerate}[label=(\roman*)]
			\item We use the states from \cref{lemma:quantum-hard-exponential-complexity}, which are pure states such that their support are not mutually contained one inside the other. This directly implies that $D(\rho_n\fatpipe \sigma_n) = \infty$. The fact $\underline{D}(\rho_n \fatpipe \sigma) \simeq 0$ follows from computational smoothing via \cref{corollary:zero_rate_for_computationally_indistinguishable_states}.
			
			\item We use the states from \cref{lem:quantum-hard-one-way}, which are classical mixed states with disjoint support. This directly implies that $D(\rho_n\fatpipe \sigma_n) = \infty$. The fact $\underline{D}(\rho_n \fatpipe \sigma) \simeq 0$ follows from computational smoothing via \cref{corollary:zero_rate_for_computationally_indistinguishable_states}.
			
			\item We use the states from \cref{lem:classically-hard-one-way}, which are classical mixed states with disjoint support. The classical indistinguishability implies \smash{$\underline{D}^{\mathrm{cl}}(\rho_n \fatpipe \sigma_n) \simeq 0$} from computational smoothing via \cref{corollary:zero_rate_for_computationally_indistinguishable_states}. The lower bound on the quantum computational relative entropy follows from the computational Bretagnolle-Huber inequality \cref{theorem:computational_bretagnolle_huber} and \smash{$\underline{d}_{\mathrm{TV}}(\rho_n, \sigma_n) \simeq 1$}. 
		\end{enumerate}
	\end{proof}
	
	\subsection*{Explicit example of catalysis and superadditivity}
	While the relative entropy is additive with respect to tensor products, we have only established superadditivity of the computational relative entropy in \cref{item:computational_relative_entropy_superadditivity}. In the following theorem, which also makes use of computational indistinguishability, we give an example of states for which the superadditivity is strict. As the tensored state is the same for both arguments, this also gives an example of \emph{catalysis} by an exponential complexity state, where adding an additional state that is useless in terms of relative entropy actually allows one to increase the computational relative entropy. 
	\begin{theorem}[Explicit example of catalysis and superadditivity]\label{theorem:explicit_examples_of_catalysis_and_superadditivity}
		There exist states $\rho_n, \sigma_n, \tau_n \in \calH_n$ such that 
		\begin{align}
			\underline{D}(\rho_n \fatpipe \sigma_n) \simeq \underline{D}(\tau_n \fatpipe \tau_n) \simeq 0
		\end{align}
		but
		\begin{align}
			\underline{D}(\rho_n \otimes \tau_n \fatpipe \sigma_n \otimes \tau_n) \gtrsim C > 0.
		\end{align}
		for some constant $C$. As such, the computational relative entropy can be strictly superadditive with respect to tensor products. Furthermore, it shows that catalysis is possible with catalysts that are of super-polynomial complexity.
	\end{theorem}
	\begin{proof}
		Our approach is to choose $\rho_n = \psi_n$ and $\sigma_n = \phi_n$, where $\psi_n$ and $\phi_n$ are the states used in \cref{lemma:quantum-hard-exponential-complexity} that were sampled from the Haar measure. As these states were shown to be computationally indistinguishable, we know that $\underline{D}(\rho_n \fatpipe \sigma_n) \simeq 0$ as claimed. \cref{lemma:quantum-hard-exponential-complexity} also shows that the overlap of $\psi_n$ and $\phi_n$ is negligible. 
		If we now choose $\tau_n = \rho_n = \psi_n$ as an advice state, we can use the fact that we know that the second state is $\psi_n$ to actually distinguish them. This works by simply running a SWAP test between the individual copies of the unknown and the advice state and accepting if all the SWAP tests return the outcome zero. In the type I case, this happens with unit probability, as the two states are actually the same pure state. In the type II case, we perform a SWAP test between $\psi_n$ and $\phi_n$, two states with negligible overlap. The SWAP test returns the outcome zero with probability $\frac{1}{2} + \negl(n)$. The probability of all SWAP tests returning zero therefore decays exponentially in the number of copies and we get the bound
		\begin{align}
			\underline{D}(\rho_n \otimes \tau_n \fatpipe \sigma_n \otimes \tau_n) &\simeq \underline{D}(\psi_n \otimes \psi_n \fatpipe \phi_n \otimes \psi_n) \gtrsim -\log \left(\frac{1}{2} + \negl(n)\right) \simeq \log 2,
		\end{align}
		proving the claimed result.
	\end{proof}
	
	\section{Computationally measured quantum divergences} \label{section:computationally_measured_divergences}
	A natural way of lifting classical divergence measures to the quantum realm is to optimize them over measurements. This approach has a long history, and versions of this technique have been considered for arbitrary measurements, projective measurements or LOCC measurements. 
	
	In this section, we introduce a version of measured divergences that take complexity restrictions into account. To do so, we first define a single-shot version that captures measured divergences with a fixed number of gates and measurement outcomes.
	\begin{definition}[$G$-$K$-complexity measured quantum divergence]\label{def:g_k_complexity_measured_quantum_divergence}
		Consider a classical divergence measure between probability distributions $\mathbf{D}(p\fatpipe q) \geq 0$. The associated \emph{$G$-$K$-complexity measured quantum divergence} for two states $\rho, \sigma \in \calH$ is defined as
		\begin{align}
			\mathbf{D}^{\bbC}(\rho \fatpipe \sigma; G; K) \coloneqq \max_{\calM \in \calM({\calH}; G; K)}
			\mathbf{D}(\calM[\rho] \fatpipe \calM[\sigma]).
		\end{align}    
	\end{definition}
	In the above definition, we follow the common notation of indicating measured entropies using superscripts in blackboard font. The particular superscript $\bbC$ is to indicate \enquote{complexity}. The maximum exists for the same reason as in \cref{def:g_complexity_hypothesis_testing_relative_entropy}.
	
	We now proceed analogously to our definition of the computational relative entropy and regularize the above quantity over polynomially many copies and polynomially complex measurements, but with a restriction to two measurement outcomes. As it is our aim to connect the thus-defined divergence measure to hypothesis testing, the restriction to two measurement outcomes arises naturally.
	The advantage of having a fixed number of outcomes is that it implicitly takes care of addressing any possible computational bottleneck. Mosonyi and Hiai refer to the similar concept without computational restrictions as \enquote{test-measured}~\cite{mosonyi2023test-measured}.
	
	\begin{definition}[Computational two-outcome measured quantum divergence]\label{def:computationally_measured_two_outcome_quantum_divergence}
		Consider a classical divergence measure between probability distributions $\mathbf{D}(p\fatpipe q) \geq 0$. The \emph{computational two-outcome measured quantum divergence} associated to $\mathbf{D}$ for two states $\rho_n, \sigma_n \in \calH_n$ is defined as
		\begin{align}
			\underline{\mathbf{D}}^{\mathbbm{2}}(\rho_n \fatpipe \sigma_n) &\colonsimeq \flim_{\ell \to \infty} \fliminf_{k \to \infty}  \frac{1}{n^k} \mathbf{D}^{\bbC}(\rho_n^{\otimes n^k} \fatpipe \sigma_n^{\otimes n^k}; n^{k\ell}; 2).
		\end{align}
	\end{definition}
	The limit in $\ell$ is guaranteed to exist as the $G$-$K$-complexity measured quantum divergence can only increase with the number of available gates. This monotonicity is then enough to guarantee the existence of the limit as argued in \cref{sec:limits_of_polynomial_resources}.
	
	Like the computational relative entropy, the computational two-outcome measured quantum divergence has some natural properties.
	\begin{proposition}[Properties of computational two-outcome measured quantum divergence]\label{proposition:properties_of_computational_two_outcome_divergence}
		The computational two-outcome measured quantum divergence has the following properties:
		\begin{propenum}
			\item \emph{Additivity under tensor products.} Let $m\leq \poly(n)$. Then,
			\begin{align}
				\underline{\mathbf{D}}^{\mathbbm{2}}(\rho_n^{\otimes m} \fatpipe \sigma_n^{\otimes m}) \simeq m \underline{\mathbf{D}}^{\mathbbm{2}}(\rho_n \fatpipe \sigma_n).
			\end{align}
			
			\item 
			\emph{No polynomial catalysts.} Let $\tau_n$ be a state of polynomial complexity, $C(\tau_n) \leq O(\poly(n))$. Then,
			\begin{align}
				\underline{\mathbf{D}}^{\mathbbm{2}}(\rho_n \otimes \tau_n \fatpipe \sigma_n \otimes \tau_n) &\simeq \underline{\mathbf{D}}^{\mathbbm{2}}(\rho_n \fatpipe \sigma_n).
			\end{align}
			
			\item \emph{Data processing.} Let $\Phi_n\colon \calH_n \to \calK_n$ be a quantum channel of polynomial gate complexity. Then,
			\begin{align}
				\underline{\mathbf{D}}^{\mathbbm{2}}(\Phi_n[\rho_n] \fatpipe \Phi_n[\sigma_n]) &\lesssim \underline{\mathbf{D}}^{\mathbbm{2}}(\rho_n \fatpipe \sigma_n).
			\end{align}
			\item \emph{Isometric invariance.} Let $V_n\colon\calH_n \to \calK_n$ be an isometry of polynomial gate complexity. Then,
			\begin{align}
				\underline{\mathbf{D}}^{\mathbbm{2}}(V_n \rho_n V_n^{\dagger}\fatpipe  V_n\sigma_n V_n^{\dagger}) \simeq \underline{\mathbf{D}}^{\mathbbm{2}}(\rho_n \fatpipe \sigma_n).
			\end{align}
		\end{propenum}		
	\end{proposition}
	\begin{proof}
		We provide separate proofs for each fact.
		\begin{enumerate}[label=(\roman*)]			
			\item As $m \leq \poly(n)$, we have that $m \leq n^q$ for sufficiently large $n$ and hence $m^\ell \leq n^{q\ell}$. We expand the definition to observe
			\begin{align}
				\underline{\mathbf{D}}^{\mathbbm{2}}(\rho_n^{\otimes m} \fatpipe \sigma_n^{\otimes m}) &\simeq m\flim_{\ell \to \infty}\fliminf_{k\to\infty} 
				\frac{1}{m n^k} \mathbf{D}^{\bbC}(\rho_n^{\otimes m n^k} \fatpipe \sigma_n^{\otimes m n^k}; n^{k \ell}; 2).
			\end{align}
			We see that $m n^k$ is a polynomial number of copies, and as such the limit $k \to \infty$ will result in the limit over all polynomial numbers of copies. We are left to see that we have enough gates, as we would need $m^\ell n^{k\ell}$. As $m^\ell \leq n^{q\ell}$ we have that $n^{k\ell'} \geq m^\ell n^{k\ell}$ if $k \ell' \geq (k + q)\ell$ which is guaranteed for $\ell' \geq \frac{k+q}{k} \to 1$ as $k\to \infty$. Hence, we can assume $k$ sufficiently large and choose $\ell' = 2\ell$. As taking $\ell \to \infty$ is the same as taking $\ell' \to \infty$, we observe the desired equality.
			
			\item This property follows from the fact that we can absorb the preparation of $\tau_n$ into any measurement. Let us look at the single-shot case first. 
			Let $\calM$ denote a two-outcome measurement channel with gate complexity $G$ acting on a system $\calH \otimes \calK$ as 
			\begin{align}
				\calM[ X \otimes \tau] = \Tr[ \Lambda (X \otimes \tau)] \oplus (1-\Tr[\Lambda (X \otimes \tau)]).
			\end{align}            
			We define $\Lambda_\tau \coloneqq \Tr_{\calK}[\Lambda (\bbI \otimes \tau)]$. Then, the above can be written as
			\begin{align}
				\calM[ X \otimes \tau] = \Tr[ \Lambda_{\tau} X] \oplus (1-\Tr[\Lambda_\tau X]),
			\end{align}
			which induces a channel $\calM_\tau[X] = \Tr[ \Lambda_{\tau} X] \oplus (1-\Tr[\Lambda_\tau X])$. The maximal complexity of the map $\calM_\tau$ is $C(\calM_\tau) \leq G + C(\tau)$. We thus have that            
			\begin{align}
				\mathbf{D}^{\bbC}(\rho \otimes \tau \fatpipe \sigma \otimes \tau; G; 2) &= \max_{\calM \in \calM(\calH \otimes \calK; G;2)}\mathbf{D}(\calM[\rho \otimes \tau] \fatpipe \calM[\sigma \otimes \tau]) \\
				\nonumber
				&= \max_{\calM \in \calM(\calH \otimes \calK; G;2)}\mathbf{D}(\calM_\tau[\rho] \fatpipe \calM_\tau[\sigma])  \\
				\nonumber
				&\leq \max_{\calM' \in \calM(\calH \otimes \calK; G + C(\tau);2)}\mathbf{D}(\calM'[\rho] \fatpipe \calM'[\sigma])  \nonumber \\
				&= \mathbf{D}^{\bbC}(\rho \fatpipe \sigma; G + C(\tau); 2) \nonumber.
			\end{align}
			We can thus conclude that, for $\tau_n$ such that $C(\tau_n) \leq O(\poly(n)) \leq n^q$ for sufficiently large $n$, we have
			\begin{align}
				\underline{\mathbf{D}}^{\mathbbm{2}}(\rho_n \otimes \tau_n \fatpipe \sigma_n \otimes \tau_n)
				&\simeq \flim_{\ell \to \infty} \fliminf_{k \to \infty} \frac{1}{n^k} \mathbf{D}^{\bbC}(\rho_n^{\otimes n^k} \otimes \tau_n^{\otimes n^k} \fatpipe \sigma_n^{\otimes n^k} \otimes \tau_n^{\otimes n^k}; n^{k\ell}; 2) \\
				&\lesssim \flim_{\ell \to \infty} \fliminf_{k \to \infty} \frac{1}{n^k} \mathbf{D}^{\bbC}(\rho_n^{\otimes n^k} \fatpipe \sigma_n^{\otimes n^k}; n^{k\ell} + n^k C(\tau_n); 2)\nonumber \\
				&\lesssim \flim_{\ell \to \infty} \fliminf_{k \to \infty} \frac{1}{n^k} \mathbf{D}^{\bbC}(\rho_n^{\otimes n^k} \fatpipe \sigma_n^{\otimes n^k}; n^{k(\ell+q)}; 2) \nonumber \\
				&\simeq \underline{\mathbf{D}}^{\mathbbm{2}}(\rho_n \fatpipe \sigma_n),\nonumber 
			\end{align}
			where the last line follows because taking $\ell + q$ to infinity is the same as taking $\ell$ to infinity. The reverse direction
			\begin{align}
				\underline{\mathbf{D}}^{\mathbbm{2}}(\rho_n \otimes \tau_n \fatpipe \sigma_n \otimes \tau_n) &\gtrsim \underline{\mathbf{D}}^{\mathbbm{2}}(\rho_n \fatpipe \sigma_n)
			\end{align}
			is a simple consequence of data processing (proven in the next item) under tracing out the second subsystem, which is an efficient operation. This concludes the proof of the fact.

			\item 
			As $\Phi_n$ has polynomial gate complexity, we know that there exists a $q \in \bbN$ such that for sufficiently large $n$, we have $C(\Phi_n) \leq n^q$.
			\begin{align}
				\underline{\mathbf{D}}^{\mathbbm{2}}(\Phi_n[\rho_n] \fatpipe \Phi_n[\sigma_n])  &\simeq 
				\flim_{\ell \to \infty} \fliminf_{k \to \infty}  \frac{1}{n^k} \mathbf{D}^{\bbC}(\Phi_n[\rho_n]^{\otimes n^k} \fatpipe \Phi_n[\sigma_n]^{\otimes n^k}; (n^{k})^\ell) \\
				\nonumber
				&\simeq \flim_{\ell \to \infty} \fliminf_{k \to \infty}  \frac{1}{n^k}\max_{\calM \in \calM(\calK_n^{\otimes n^k}; n^{k\ell}, 2)}
				\mathbf{D}(\calM[\Phi_n[\rho_n]^{\otimes n^k}] \fatpipe \calM[\Phi_n[\sigma_n]^{\otimes n^k}])\\
				\nonumber
				&\simeq \flim_{\ell \to \infty} \fliminf_{k \to \infty}  \frac{1}{n^k}\max_{\calM' \in \calM(\calK_n^{\otimes n^k}; n^{k\ell}, 2) \circ \Phi_n^{\otimes n^k}}
				\mathbf{D}(\calM'[\rho_n^{\otimes n^k}] \fatpipe \calM'[\sigma_n^{\otimes n^k}])
				\\
				\nonumber
				&\lesssim \flim_{\ell \to \infty} \fliminf_{k \to \infty}  \frac{1}{n^k}\max_{\calM' \in \calM(\calH_n^{\otimes n^k}; n^{k\ell} + n^{k} C(\Phi_n), 2)}
				\mathbf{D}(\calM'[\rho_n^{\otimes n^k}] \fatpipe \calM'[\sigma_n^{\otimes n^k}])
				\\
				\nonumber
				&\lesssim \flim_{\ell \to \infty} \fliminf_{k \to \infty}  \frac{1}{n^k}\max_{\calM' \in \calM(\calH_n^{\otimes n^k}; n^{k(\ell+q)}, 2)}
				\mathbf{D}(\calM'[\rho_n^{\otimes n^k}] \fatpipe \calM'[\sigma_n^{\otimes n^k}])
				\\
				\nonumber
				&\simeq \underline{\mathbf{D}}^{\mathbbm{2}}(\rho_n \fatpipe \sigma_n). 
				\nonumber
			\end{align}
			The crucial step above is the application of \cref{lemma:absorbing_channels_into_measurement_sets} to conclude
			\begin{align}
				\calM(\calK_n^{\otimes n^k}, n^{k\ell},2) \circ \Phi^{\otimes n^k} \subseteq \calM(\calH_n^{\otimes n^k}, n^{k\ell} + n^k C(\Phi),2).
			\end{align}
			Maximizing over a larger set always yields a larger value, implying the desired result together with the observation that taking $\ell \to \infty$ is the same as taking $\ell + q \to \infty$.

			\item 
			As $V_n$ has polynomial gate complexity, we know there exists a $q \in \bbN$ such that for sufficiently large $n$ we have $C(V_n)\leq n^q$. We furthermore know that $V_n$ is a polynomial complexity channel, implying that $\underline{\mathbf{D}}^{\mathbbm{2}}(V_n \rho_n V_n^{\dagger}\fatpipe  V_n\sigma_n V_n^{\dagger}) \lesssim \underline{\mathbf{D}}^{\mathbbm{2}}( \rho_n \fatpipe  \sigma_n)$. To prove the reverse direction, we let $\calV_n$ denote the channel associated to $V_n$ and we expand the expression of the computational relative entropy  
			\begin{align}
				\underline{\mathbf{D}}^{\mathbbm{2}}(\calV_n[\rho_n] \fatpipe \calV_n[\sigma_n])  &\simeq 
				\flim_{\ell \to \infty} \fliminf_{k \to \infty}  \frac{1}{n^k} \mathbf{D}^{\bbC}(\calV_n[\rho_n]^{\otimes n^k} \fatpipe \calV_n[\sigma_n]^{\otimes n^k}; (n^{k})^\ell) \\
				\nonumber
				&\simeq \flim_{\ell \to \infty} \fliminf_{k \to \infty}  \frac{1}{n^k}\max_{\calM \in \calM(\calK_n^{\otimes n^k}; n^{k\ell}, 2)}
				\mathbf{D}(\calM[\calV_n[\rho_n]^{\otimes n^k}] \fatpipe \calM[\calV_n[\sigma_n]^{\otimes n^k}])\\
				\nonumber
				&\simeq \flim_{\ell \to \infty} \fliminf_{k \to \infty}  \frac{1}{n^k}\max_{\calM' \in \calM(\calK_n^{\otimes n^k}; n^{k\ell}, 2) \circ \calV_n^{\otimes n^k}}
				\mathbf{D}(\calM'[\rho_n^{\otimes n^k}] \fatpipe \calM'[\sigma_n^{\otimes n^k}])
				\\
				\nonumber
				&\gtrsim \flim_{\ell \to \infty} \fliminf_{k \to \infty}  \frac{1}{n^k}\max_{\calM' \in \calM(\calH_n^{\otimes n^k}; n^{k\ell} - n^{k} C(\calV_n), 2)}
				\mathbf{D}(\calM'[\rho_n^{\otimes n^k}] \fatpipe \calM'[\sigma_n^{\otimes n^k}])
				\\
				\nonumber
				&\gtrsim \flim_{\ell \to \infty} \fliminf_{k \to \infty}  \frac{1}{n^k}\max_{\calM' \in \calM(\calH_n^{\otimes n^k}; n^{k(\ell-q)}, 2)}
				\mathbf{D}(\calM'[\rho_n^{\otimes n^k}] \fatpipe \calM'[\sigma_n^{\otimes n^k}])
				\\
				\nonumber
				&\simeq \underline{\mathbf{D}}^{\mathbbm{2}}(\rho_n \fatpipe \sigma_n). 
				\nonumber
			\end{align}
			The crucial step above is the application of \cref{lemma:absorbing_channels_into_measurement_sets} to conclude
			\begin{align}
				\calM(\calK_n^{\otimes n^k}, n^{k\ell},2) \circ \calV_n^{\otimes n^k} \supseteq \calM(\calH_n^{\otimes n^k}, n^{k\ell} - n^k C(\calV_n),2).
			\end{align}
			Maximizing over a 
			smaller set always yields a smaller value, implying the desired result together with the observation that taking $\ell \to \infty$ is the same as taking $\ell - q \to \infty$.  
		\end{enumerate}   
	\end{proof}
	
	The next result concerns the computationally measured two-outcome divergence associated to the relative entropy
	\begin{align}
		D(p \fatpipe q) = \Tr[ p (\log p - \log q)]
	\end{align}
	and shows that it is equal to the unbounded relative entropy for certain distributions.
	\begin{proposition}[Computationally measured two-outcome relative entropy for polynomial support]
		Consider two probability distributions $p_n, q_n$ with polynomial support such that $\log \lVert q_n^{-1}\rVert_{\infty},\log \lVert p_n^{-1}\rVert_{\infty} \leq \poly(n)$. Then, the computationally measured two-outcome relative entropy agrees to the regular relative entropy up to negligible terms:
		\begin{align}
			\underline{D}^{\mathbbm{2}}(p_n \fatpipe q_n) \simeq D(p_n \fatpipe q_n).
		\end{align}
	\end{proposition}
	\begin{proof}
		The upper bound $\underline{D}^{\mathbbm{2}}(p_n \fatpipe q_n) \leq D(p_n \fatpipe q_n)$ is a simple consequence of data processing under the optimal two-outcome measurement.
		For the achievability part, we will analyze an explicit measurement $\calM$ on $n^k$ copies.
		As the states $p_n$ and $q_n$ have polynomial support we can implement a likelihood ratio test $S(\delta)$ as shown in \cref{proposition:efficient_likelihood_ratio_testing} such that
		\begin{align}
			-\log\Tr[ q_n^{\otimes n^k}S(\delta)] &\geq n^k\left[ D(p_n\fatpipe q_n) - \delta - \negl(n) \right],\\
			\Tr[ p_n^{\otimes n^k} S(\delta)] &\geq 1- \frac{1}{n^k}\frac{V(p_n \fatpipe q_n)}{(\delta -\negl(n))^2},
		\end{align}
		where $D(p_n\fatpipe q_n)$ is the relative entropy and $V(p_n\fatpipe q_n)$ the relative entropy variance.
		
		We now look at the relative entropy of the post-measurement state
		\begin{align}
			D(\calM[p_n^{\otimes n^k}] \fatpipe \calM[q_n^{\otimes n^k}])
			&= \Tr[ p_n^{\otimes n^k} S(\delta)] \log \frac{1 }{\Tr[ q_n^{\otimes n^k} S(\delta)] }
			\\
			\nonumber
			&\quad + (1-\Tr[ p_n^{\otimes n^k} S(\delta)]) \log \frac{1}{1-\Tr[ q_n^{\otimes n^k} S(\delta)]} - h_2(\Tr[ p_n^{\otimes n^k} S(\delta)]) \\
			\nonumber
			&\geq \left( 1 - \frac{V(p_n\fatpipe q_n)}{n^k (\delta-\negl(n))^2}\right)n^k[D(p_n \fatpipe q_n) - \delta-\negl(n)] - h_2\left(\frac{1}{2}\right) \\
			\nonumber
			&\geq n^k[D(p_n \fatpipe q_n) - \delta-\negl(n)] - \frac{V(p_n\fatpipe q_n)}{(\delta-\negl(n))^2}[D(p_n \fatpipe q_n) - \delta-\negl(n)] - h_2\left(\frac{1}{2}\right).
			\nonumber
		\end{align}
		The assumptions on $p_n$ and $q_n$ imply that
		\begin{align}
			D(p_n \fatpipe q_n) &\leq\max\left\{ \log \lVert p_n^{-1}\rVert_{\infty},\log \lVert q_n^{-1}\rVert_{\infty} \right\} \leq \poly(n) \leq n^q,\\
			V(p_n \fatpipe q_n) &\leq \max\left\{ \log \lVert p_n^{-1}\rVert_{\infty},\log \lVert q_n^{-1}\rVert_{\infty} \right\}^2 \leq\poly(n) \leq n^q
		\end{align}
		for some $q \in \bbN$, which means we have
		\begin{align}
			\frac{1}{n^k}D(\calM[p_n^{\otimes n^k}] \fatpipe \calM[q_n^{\otimes n^k}])
			\geq D(p_n \fatpipe q_n) - \delta - \frac{n^{2 q}}{n^k (\delta-\negl(n))^2} - \frac{1}{n^k} h_2\left(\frac{1}{2}\right) -\negl(n).
		\end{align}
		We now choose $\delta = n^{-k/3}$ and send $k\to\infty$ (the limit $\ell \to \infty$ is not necessary anymore because we implemented the likelihood ratio test for some large but fixed $\ell$) to obtain
		\begin{align}
			\underline{D}^{\mathbbm{2}}(p_n \fatpipe q_n)
			&\gtrsim D(p_n \fatpipe q_n),
		\end{align}
		concluding the proof.
	\end{proof}
	
	\section{Computational quantum Stein's lemma}
	Stein's Lemma is a foundational result in asymptotic information theory, identifying the relative entropy as the regularized hypothesis testing relative entropy and thus endowing it with its operational interpretation. In the quantum setting, the quantum Stein's Lemma furthermore singles out the Umegaki relative entropy
	\begin{align}
		D(\rho \fatpipe \sigma) = \Tr[ \rho (\log \rho - \log \sigma)]
	\end{align}
	as the \enquote{correct} generalization of the relative entropy.
	
	We desire to prove a computational version of Stein's Lemma that relates the computational relative entropy to the computational two-outcome measured relative entropy. 
	The first step is to establish a single-shot converse bound that relates the $G$-complexity hypothesis testing relative entropy to $G$-$K$-measured quantum divergences associated to the relative entropy.
	\begin{lemma}[Single-shot converse (Relative entropy)]\label{lemma:single_shot_converse_relative_entropy}
		Consider two states $\rho, \sigma \in \calH$ and $0<\epsilon\leq1/2$. Then, for any $K\geq 2$ we have 
		\begin{align}
			D_h^{\epsilon}(\rho \fatpipe \sigma; G) \leq \frac{1}{1-\epsilon}\left( D^{\bbC}(\rho \fatpipe \sigma; G; K) + h_2(\epsilon)\right).
		\end{align}
		For $1/2 < \epsilon \leq 1$ the same inequality holds with $h_2(\epsilon)$ replaced by $h_2(1/2)$.
	\end{lemma}
	\begin{proof}
		We closely follow the proof of Proposition 7.70 of Ref.~\cite{khatri2024principles}. 
		We consider the $G$-$K$-complexity measured Rényi relative entropy for $G$ gates and $K$ measurement outcomes. 
		Let now 
		\begin{align}
			\calM_{\Lambda}[\tau] = \begin{pmatrix}
				\Tr[ \Lambda \tau] & 0 \\
				0 & \Tr[ (\bbI - \Lambda) \tau]
			\end{pmatrix}
		\end{align}
		be the measurement channel associated to the POVM that achieves optimality in the calculation of $D_h^{\epsilon}(\rho \fatpipe \sigma; G)$. For brevity, we denote $\beta \coloneqq \beta_h^{\epsilon}(\rho \fatpipe \sigma; G) =\exp(-D_h^\epsilon(\rho\fatpipe \sigma; G)$. With this notation, we have
		\begin{align}
			\calM_{\Lambda}[\rho] &= \begin{pmatrix}
				1-\epsilon' & 0 \\
				0 & \epsilon'
			\end{pmatrix}, \\
			\calM_{\Lambda}[\sigma] &= \begin{pmatrix}
				\beta & 0 \\
				0 & 1- \beta
			\end{pmatrix},
		\end{align}
		where $\epsilon' \leq \epsilon$. We furthermore know that the complexity of the measurement is bounded as $C(\calM_{\Lambda}) \leq G$ with 2 outcomes, by construction. Hence, $\calM_{\Lambda}$ is a candidate in the optimization in ${D}^{\bbC}(\rho\fatpipe \sigma; G; K)$. We, therefore, have that
		\begin{align}
			{D}^{\bbC}(\rho \fatpipe \sigma; G; K) 
			&\geq D(\calM_{\Lambda}[\rho] \fatpipe \calM_{\Lambda}[\sigma]  ) \\
			\nonumber
			&= D\left(\left. \begin{pmatrix}
				1-\epsilon' & 0 \\
				0 & \epsilon'
			\end{pmatrix} \, \right\| \,\begin{pmatrix}
				\beta & 0 \\
				0 & 1- \beta
			\end{pmatrix} \right) \\
			\nonumber
			&= (1-\epsilon')\log\frac{1-\epsilon'}{\beta} + \epsilon' \log \frac{\epsilon'}{1-\beta} 
			\\
			\nonumber
			&= (1-\epsilon') D_h^{\epsilon}(\rho\fatpipe \sigma; G) -h_2(\epsilon') + \epsilon' \log\frac{1}{1-\beta} \\
			\nonumber
			&\geq (1-\epsilon') D_h^{\epsilon}(\rho\fatpipe \sigma; G) -h_2(\epsilon')\\
			\nonumber
			&\geq (1-\epsilon) D_h^{\epsilon}(\rho\fatpipe \sigma; G) -h_2(\epsilon).
			\nonumber
		\end{align}
		Here, we have used  the fact that $\log \frac{1}{1-\beta} \geq 0$ and $\epsilon' \leq \epsilon \leq 1/2$. For $\epsilon > 1/2$ it can happen that $h_2(\epsilon')$ is actually slightly larger than $\epsilon$, which is why we upper bound it by the worst case $h_2(1/2)$ in that case.
		The statement of the Lemma follows from rearranging.
	\end{proof}
	
	This converse, together with the results on efficient likelihood ratio testing of \cref{proposition:efficient_likelihood_ratio_testing} is already sufficient to prove the computational version of Stein's Lemma. The statement contains a technical condition on the computationally measured two-outcome max-relative entropy    \begin{align}\label{eqn:def_computationally_measured_two_outcome_root_variance}
		\underline{D}_{\max}^{\mathbbm{2}}(\rho_n \fatpipe \sigma_n) \colonsimeq \flim_{\ell \to \infty} \fliminf_{k\to\infty} \frac{1}{n^k} \max_{\calM \in \calM(\calH_n^{\otimes n^k}; n^{k\ell}; 2)} D_{\max}(\calM[\rho_n^{\otimes n^k}] \fatpipe \calM[\sigma_n^{\otimes n^k}]),
	\end{align}
	where classically
	\begin{align}
		D_{\max}(p \fatpipe q) = \log \left\lVert \frac{p}{q} \right\rVert_{\infty}.
	\end{align}
	In essence, this condition ensures that the the computational relative entropy and the correction terms do not grow faster than any polynomial, in which case we cannot expect to get a well-behaved quantity through polynomial regularization.    
	\begin{theorem}[Computational Stein's lemma]\label{theorem:computational_steins_lemma}
		Consider two states $\rho_n, \sigma_n \in \calH_n$ such that $\underline{D}_{\max}^{\mathbbm{2}}(\rho_n \fatpipe \sigma_n) \lesssim O(\poly(n))$. Then, 
		\begin{align}
			\underline{D}(\rho_n \fatpipe \sigma_n) \simeq \underline{D}^{\mathbbm{2}}(\rho_n \fatpipe \sigma_n).
		\end{align}
	\end{theorem}
	Before we come to the proof, we wish to emphasize that the above statement indeed represents a gain in information about the computational relative entropy, as it equates a purely operational quantity to one where we optimize a well-known information theoretic quantity (the relative entropy) over computationally bounded measurements.
	\begin{proof}
		We start with the converse bound $\underline{D}(\rho_n \fatpipe \sigma_n) \lesssim \underline{D}^{\mathbbm{2}}(\rho_n \fatpipe \sigma_n)$. We expand the definition of the computational relative entropy and apply the single-shot upper bound of \cref{lemma:single_shot_converse_relative_entropy} with number of gates $G = n^{k\ell}$ and number of measurement outcomes $K=2$:
		\begin{align}
			\underline{D}(\rho_n \fatpipe \sigma_n) &\simeq \flim_{\epsilon \to 0} \flim_{\ell \to \infty} \fliminf_{k\to\infty} \frac{1}{n^k}D_h^{\epsilon}(\rho_n^{\otimes n^k} \fatpipe \sigma_n^{\otimes n^k}; n^{k\ell}) \\
			\nonumber
			&\lesssim \flim_{\epsilon \to 0} \flim_{\ell \to \infty} \fliminf_{k\to\infty} \frac{1}{n^k}\frac{1}{1-\epsilon} \left[ D^{\mathbbm{C}}(\rho_n^{\otimes n^k} \fatpipe \sigma_n^{\otimes n^k}; n^{k\ell}; 2) + h_2(\epsilon)\right] \\
			\nonumber
			&\simeq \flim_{\epsilon\to 0} \frac{1}{1-\epsilon} \underline{D}^{\mathbbm{2}}(\rho_n \fatpipe \sigma_n) \\
			\nonumber
			&\simeq\underline{D}^{\mathbbm{2}}(\rho_n \fatpipe \sigma_n).
			\nonumber
		\end{align}
		Here, we have used  the fact that $\flim_{k\to\infty} n^{-k} h_2(\epsilon) \leq \flim_{k\to\infty} n^{-k} h_2(1/2) \simeq 0$ as per \cref{lemma:limit_of_inverse_polynomials_is_negligible}.
		
		We now come to the achievability statement $\underline{D}(\rho_n \fatpipe \sigma_n) \gtrsim \underline{D}^{\mathbbm{2}}(\rho_n \fatpipe \sigma_n)$. For that, we will use a double blocking approach and \cref{prop:computational_relative_entropy_poly_size_and_variance}. Let $k \in \bbN$, then by the additivity under polynomial tensor powers (\cref{item:computational_relative_entropy_tensor_additivity})
		\begin{align}
			\underline{D}(\rho_n \fatpipe q_n) &\simeq \frac{1}{n^k} \underline{D}(\rho_n^{\otimes n^k} \fatpipe \sigma_n^{\otimes n^k}).
		\end{align}
		Let now $\calM$ be the measurement channel that achieves optimality in the calculation of 
		\begin{align}
			D^{\bbC}(\rho_n^{\otimes n^k} \fatpipe \sigma_n^{\otimes n^k}; n^{k\ell},2) = D(\calM[\rho_n^{\otimes n^k}] \fatpipe \calM[\sigma_n^{\otimes n^k}]).
		\end{align}
		By construction, $\calM$ is a map with polynomial complexity and we can hence make use of data processing to conclude
		\begin{align}
			\underline{D}(\rho_n \fatpipe q_n) &\simeq \frac{1}{n^k} \underline{D}(\rho_n^{\otimes n^k} \fatpipe \sigma_n^{\otimes n^k}) \gtrsim \frac{1}{n^k} \underline{D}(\calM[\rho_n^{\otimes n^k}] \fatpipe \calM[\sigma_n^{\otimes n^k}]).
		\end{align}
		Now we note that $\calM$ is a two-outcome measurement, which means we can apply \cref{prop:computational_relative_entropy_poly_size_and_variance} under the assumption that 
		\begin{align}\label{eqn:steins_lemma_proof_condition_variance_poly_bounded}
			V(\calM[\rho_n^{\otimes n^k}] \fatpipe \calM[\sigma_n^{\otimes n^k}]) \leq O(\poly(n))
		\end{align}
		to conclude
		\begin{align}
			\underline{D}(\rho_n \fatpipe q_n) &\gtrsim \frac{1}{n^k} {D}(\calM[\rho_n^{\otimes n^k}] \fatpipe \calM[\sigma_n^{\otimes n^k}]) \simeq \frac{1}{n^k} D^{\bbC}(\rho_n^{\otimes n^k} \fatpipe \sigma_n^{\otimes n^k}; n^{k\ell}).
		\end{align}      
		As the above holds for all $k$ and $\ell$ we can take the appropriate limits to arrive at
		\begin{align}
			\underline{D}(\rho_n \fatpipe q_n) &\gtrsim \underline{D}^{\mathbbm{2}}(\rho_n \fatpipe \sigma_n)
		\end{align}
		as desired. 
		We are left to make sure that the assumption on the relative entropy variance is satisfied. We note that it only has to hold in the limit of $k, \ell \to \infty$. We use the fact that we can bound the relative entropy variance by the squared max-relative entropy, in this case for the classical post-measurement states:
		\begin{align}
			V(p \fatpipe q) &= \Tr[ p (\log p - \log q)^2] - D(p \fatpipe q)^2 \\
			&\leq  \Tr[ p (\log p - \log q)^2] \\
			&\leq \Tr[ p \bbI ] D_{\max}(p \fatpipe q)^2 \\
			&=D_{\max}(p \fatpipe q)^2.
		\end{align}
		Applying this to our setting under the appropriate limits, we require that
		\begin{align}
			\underline{D}_{\max}^{\mathbbm{2}}(\rho_n \fatpipe \sigma_n) \lesssim O(\poly(n)),
		\end{align}  
		which is exactly the requirement we added to the theorem statement.
	\end{proof}
	
	\section{Computational Rényi relative entropies}\label{section:computational_renyi_relative_entropies}
	Rényi relative entropies 
	\begin{align}
		D_{\alpha}(p \fatpipe q) = \frac{1}{\alpha - 1}\Tr[ p^{\alpha }q^{1-\alpha}]
	\end{align}
	play a core role in contemporary quantum information theory as they can serve as information-theoretic bounds for the purely operational hypothesis testing relative entropy. The objective of this section is to relate the computational relative entropy to the computational two-outcome measured Rényi relative entropies.
	
	We first present a converse bound. 
	\begin{lemma}[Single-shot converse (Rényi relative entropy)]\label{lemma:computational_stein_single_shot_converse_renyi}
		Consider two states $\rho, \sigma \in \calH$ and $0<\epsilon<1$. Then, for  any $\alpha > 1$ and $K\geq 2$ we have 
		\begin{align}
			D_h^{\epsilon}(\rho \fatpipe \sigma; G) \leq D_{\alpha}^{\bbC}(\rho \fatpipe \sigma; G; K) + \frac{\alpha}{\alpha - 1} \log \frac{1}{1-\epsilon}.
		\end{align}
	\end{lemma}
	\begin{proof}
		We closely follow the proof of Proposition 7.71 of Ref.~\cite{khatri2024principles}. 
		We consider the $G$-$K$-complexity measured Rényi relative entropy for $G$ gates and $K$ measurement outcomes. 
		Let now 
		\begin{align}
			\calM_{\Lambda}[\tau] = \begin{pmatrix}
				\Tr[ \Lambda \tau] & 0 \\
				0 & \Tr[ (\bbI - \Lambda) \tau]
			\end{pmatrix}
		\end{align}
		be the measurement channel associated to the POVM that achieves optimality in the calculation of $D_h^{\epsilon}(\rho \fatpipe \sigma; G)$. For brevity, we denote $\beta \coloneqq \beta_h^{\epsilon}(\rho \fatpipe \sigma; G)=\exp(-D_h^\epsilon(\rho\fatpipe \sigma; G)$. With this notation, we have
		\begin{align}
			\calM_{\Lambda}[\rho] &= \begin{pmatrix}
				1-\epsilon' & 0 \\
				0 & \epsilon'
			\end{pmatrix} ,\\
			\calM_{\Lambda}[\sigma] &= \begin{pmatrix}
				\beta & 0 \\
				0 & 1- \beta
			\end{pmatrix},
		\end{align}
		where $\epsilon' \leq \epsilon$. We furthermore know that the complexity of the measurement is bounded as $C(\calM_{\Lambda}) \leq G$ with 2 outcomes, 
		by construction. Hence, $\calM_{\Lambda}$ is a candidate in the optimization in ${D}_{\alpha}^{\bbC}(\rho\fatpipe \sigma; G; K)$. We therefore have that
		\begin{align}
			{D}_{\alpha}^{\bbC}(\rho \fatpipe \sigma; G; K) 
			&\geq D_{\alpha}(\calM_{\Lambda}[\rho] \fatpipe \calM_{\Lambda}[\sigma]  ) \\
			\nonumber
			&= D_{\alpha}\left(\left. \begin{pmatrix}
				1-\epsilon' & 0 \\
				0 & \epsilon'
			\end{pmatrix} \, \right\| \,\begin{pmatrix}
				\beta & 0 \\
				0 & 1- \beta
			\end{pmatrix} \right) \\
			\nonumber
			&= \frac{1}{\alpha - 1} \log\left( 
			(1- \epsilon')^{\alpha} \beta^{1-\alpha} + (\epsilon')^{\alpha} (1-\beta)^{1-\alpha}\right) \\
			\nonumber
			&\geq \frac{1}{\alpha - 1} \log\left( 
			(1- \epsilon')^{\alpha} \beta^{1-\alpha}\right) \\
			\nonumber
			&= \frac{\alpha}{\alpha - 1}\log (1-\epsilon') - \log \beta \\
			&\geq D_h^{\epsilon}(\rho \fatpipe \sigma; G) + \frac{\alpha}{\alpha - 1}\log (1-\epsilon).
			\nonumber
		\end{align}
		The statement of the Lemma follows from rearranging.
	\end{proof}
	
	We also find the following achievability statement using the $G$-$K$-complexity measured Rényi relative entropy.
	\begin{lemma}[Single-shot achievability]\label{lemma:computational_stein_single_shot_achievability_renyi}
		Consider two states $\rho, \sigma \in \calH$ and $0<\epsilon<1$. Then, for any $0 < \alpha < 1$ we have 
		\begin{align}
			D_h^{\epsilon}(\rho \fatpipe \sigma; G + \poly(K)) \geq D_{\alpha}^{\bbC}(\rho \fatpipe \sigma; G; K) - \frac{\alpha}{1-\alpha} \log \frac{1}{\epsilon}.
		\end{align}
	\end{lemma}
	\begin{proof}
		We closely follow the proof of Proposition 7.72 of Ref.~\cite{khatri2024principles}. 
		Let us consider the $G$-$K$-complexity measured Rényi relative entropy for $0 < \alpha < 1$ and let $\calM$ be the measurement that achieves optimality, i.e.,
		\begin{align}
			D_{\alpha}^{\bbC}(\rho \fatpipe \sigma; G;K) = D_{\alpha}(\calM[\rho] \fatpipe \calM[\sigma]).
		\end{align}
		For brevity, we denote the classical post-measurement states as $p = \calM[\rho]$, $q = \calM[\sigma]$. These states have dimension $K$ by construction of $\calM$. 
		We now wish to implement the optimal symmetric distinguishing measurement between $p$ and $q$ with prior probabilities $\gamma$ and $1-\gamma$. This measurement, with $\Lambda$ being the POVM effect detecting $p$, is given by
		\begin{align}
			\Lambda = \{ \gamma p > (1-\gamma) q\}.
		\end{align}
		We can always implement this by measuring in the computational basis and then comparing with a lookup table of size $K$, analogous to \cref{proposition:efficient_likelihood_ratio_testing}. Therefore, the gate cost of this is at most $\poly(K)$. We now essentially replicate the argument that leads to Proposition 7.72 of Ref.~\cite{khatri2024principles}.
		We know from the quantum Chernoff bound that $\Lambda$ fulfills
		\begin{align}\label{eqn:chernoff_upper_bound_opt_hyp_test}
			\gamma \Tr[ (\bbI - \Lambda) p]  + (1-\gamma) \Tr[\Lambda q] \leq \gamma^{\alpha}(1-\gamma)^{1-\alpha} \Tr[ p^{\alpha} q^{1-\alpha}].
		\end{align}
		This especially means that
		\begin{align}
			\gamma \Tr[ (\bbI - \Lambda) p] \leq \gamma^{\alpha}(1-\gamma)^{1-\alpha} \Tr[ p^{\alpha} q^{1-\alpha}].
		\end{align}
		This means we can choose
		\begin{align}
			\gamma = \left(1 + \left[\frac{\epsilon}{\Tr[p^{\alpha} q^{1-\alpha}]} \right]^{\frac{1}{1-\alpha}} \right)^{-1}
		\end{align}
		to guarantee that
		\begin{align}
			\Tr[ (\bbI - \Lambda) p] \leq \epsilon.
		\end{align}
		This makes $\Lambda$ an admissible POVM effect in the optimization of 
		\begin{align}
			D_h^{\epsilon}(p \fatpipe q; \poly(K)).
		\end{align}
		Because \cref{eqn:chernoff_upper_bound_opt_hyp_test} also guarantees the upper bound
		\begin{align}
			(1-\gamma) \Tr[\Lambda q] \leq \gamma^{\alpha}(1-\gamma)^{1-\alpha} \Tr[ p^{\alpha} q^{1-\alpha}]
		\end{align}
		we can after some algebra deduce the guarantee
		\begin{align}
			D_h^{\epsilon}(p \fatpipe q; \poly(K)) \geq - \log \Tr[\Lambda q] \geq D_{\alpha}(p \fatpipe q) - \frac{\alpha}{1-\alpha}\log \frac{1}{\epsilon}.
		\end{align}
		We now use the fact that $p$ and $q$ are the product of applying the $G$-$K$-complexity measurement $\calM$. We can thus construct a complexity-limited measurement for $\rho$ and $\sigma$ by combining $\calM$ and the subsequent measurement of $\Lambda$ to obtain the chain of inequalities
		\begin{align}
			D_h^{\epsilon}(\rho \fatpipe \sigma; G + \poly(K)) \geq D_h^{\epsilon}(\calM[\rho] \fatpipe \calM[\sigma]; \poly(K)) \geq D_{\alpha}(\calM[\rho] \fatpipe \calM[\sigma]) - \frac{\alpha}{1-\alpha}\log \frac{1}{\epsilon}
			.
		\end{align}    
		We now use the fact that
		\begin{align}
			D_{\alpha}(\calM[\rho] \fatpipe \calM[\sigma]) = D_{\alpha}^{\bbC}(\rho \fatpipe \sigma; G;K)
		\end{align}
		to conclude the statement of the lemma.    
	\end{proof}
	
	The above single-shot relations with the Rényi relative entropy directly imply bounds on the computational relative entropy in terms of the computational two-outcome measured Rényi relative entropies.
	\begin{theorem}[Regularized Rényi bounds]\label{theorem:renyi_upper_and_lower_bounds_computational_relative_entropy}
		Consider two states $\rho_n, \sigma_n \in \calH_n$ and let $0 < \alpha_{-} < 1 < \alpha_{+}$ be independent of $n$. Then, 
		\begin{align}
			\underline{D}_{\alpha_{-}}^{\mathbbm{2}}(\rho_n \fatpipe \sigma_n)\lesssim
			\underline{D}(\rho_n \fatpipe \sigma_n) \lesssim \underline{D}_{\alpha_{+}}^{\mathbbm{2}}(\rho_n \fatpipe \sigma_n).
		\end{align}
	\end{theorem}
	\begin{proof}
		We begin with the converse direction $\underline{D}(\rho_n \fatpipe \sigma_n) \lesssim \underline{D}_{\alpha_{+}}^{\mathbbm{2}}(\rho_n \fatpipe \sigma_n)$. We expand the definition of the computational relative entropy and apply \cref{lemma:computational_stein_single_shot_converse_renyi} with $\alpha = \alpha_+ > 1$, $G = n^{k\ell}$ and $K = 2$:
		\begin{align}
			\underline{D}(\rho_n \fatpipe \sigma_n) &\simeq \flim_{\epsilon \to 0} \flim_{\ell \to \infty} \fliminf_{k\to\infty} \frac{1}{n^k}D_h^{\epsilon}(\rho_n^{\otimes n^k} \fatpipe \sigma_n^{\otimes n^k}; n^{k\ell}) \\
			\nonumber
			&\lesssim \flim_{\epsilon \to 0} \flim_{\ell \to \infty} \fliminf_{k\to\infty} \frac{1}{n^k} D_{\alpha_+}^{\bbC}(\rho_n^{\otimes n^k} \fatpipe \sigma_n^{\otimes n^k}; n^{k\ell}; 2) + \frac{1}{n^k}\frac{\alpha_+}{\alpha_+ - 1} \log \frac{1}{1-\epsilon} \\
			\nonumber
			&\simeq \underline{D}_{\alpha_+}^{\mathbbm{2}}(\rho_n \fatpipe \sigma_n),
		\end{align}
		where the last line follows from \cref{lemma:limit_of_inverse_polynomials_is_negligible}.
		
		The achievability statement $\underline{D}(\rho_n \fatpipe \sigma_n)\gtrsim \underline{D}_{\alpha_{-}}^{\mathbbm{2}}(\rho_n \fatpipe \sigma_n)$ follows similarly by applying \cref{lemma:computational_stein_single_shot_achievability_renyi} with $0 <\alpha = \alpha_- < 1$, $G = n^{k\ell}$ and $K = 2$:
		\begin{align}
			\underline{D}(\rho_n \fatpipe \sigma_n) &\simeq \flim_{\epsilon \to 0} \flim_{\ell \to \infty} \fliminf_{k\to\infty} \frac{1}{n^k}D_h^{\epsilon}(\rho_n^{\otimes n^k} \fatpipe \sigma_n^{\otimes n^k}; n^{k\ell}) \\
			\nonumber
			&\gtrsim \flim_{\epsilon \to 0} \flim_{\ell \to \infty} \fliminf_{k\to\infty} \frac{1}{n^k} D_{\alpha_-}^{\bbC}(\rho_n^{\otimes n^k} \fatpipe \sigma_n^{\otimes n^k}; n^{k\ell} - C; 2) - \frac{1}{n^k}\frac{\alpha_+}{1-\alpha_-} \log \frac{1}{\epsilon} \\
			\nonumber
			&\simeq \underline{D}_{\alpha_-}^{\mathbbm{2}}(\rho_n \fatpipe \sigma_n),
		\end{align}
		where $C$ denotes the constant $\poly(2)$ from \cref{lemma:computational_stein_single_shot_achievability_renyi}. Asymptotically, losing a constant number of gates does not matter, which gives us the computational two-outcome measured Rényi relative entropy. The other term disappears again because of \cref{lemma:limit_of_inverse_polynomials_is_negligible}.
	\end{proof}
	
	\section{Computational entropy and compression}\label{section:computational_entropy_and_compression}	
	The entropy of a quantum state can be obtained from the relative entropy via
	\begin{align}
		S(\rho) = -D(\rho \fatpipe \bbI).
	\end{align}
	We can use the same relation to define a computational entropy from the computational relative entropy. In this section, we will show that the so-obtained quantity can be operationally interpreted as the optimal rate of compressing a quantum state under computational constraints.
	\begin{definition}[Computational entropy]\label{definition:computational_entropy}
		Consider a quantum state $\rho_n \in \calH_n$ with dimension $d_n$ and let $\omega_n = d_n^{-1} \bbI_n$ be the maximally mixed state on $\calH_n$. Then, we define the \emph{computational entropy} of $\rho_n$ as
		\begin{align}
			\underline{S}(\rho_n) \colonsimeq -\underline{D}(\rho_n \fatpipe \bbI_n).
		\end{align}
	\end{definition}		
	Let us first establish some properties of the computational entropy which it inherits from the properties of the computational relative entropy, see \cref{proposition:properties_of_computational_relative_entropy}.
	\begin{proposition}[Properties of the computational entropy]\label{proposition:properties_of_computational_computational_entropy}
		The computational entropy has the following properties:
		\begin{propenum}
			\item \emph{Boundedness.} \label{item:computational_entropy_boundedness} We have
			\begin{align}
				\log d_n \gtrsim \underline{S}(\rho_n) \gtrsim S(\rho_n) \geq 0.
			\end{align}
			
			\item \emph{Additivity under tensor powers.} Let $m\leq O(\poly(n))$. Then,
			\begin{align}
				\underline{S}(\rho_n^{\otimes m}) \simeq m \underline{S}(\rho_n).
			\end{align}
			
			\item \emph{Subadditivity under tensor products.} We have
			\begin{align}
				\underline{S}(\rho_n \otimes \sigma_n) \lesssim \underline{S}(\rho_n) + \underline{S}(\sigma_n).
			\end{align}
			
			\item \emph{Data processing.} Let $\Phi_n$ be a unital channel of polynomial gate complexity. Then,
			\begin{align}
				\underline{S}(\Phi_n[\rho_n]) \gtrsim \underline{S}(\rho_n).
			\end{align}
			
			\item \emph{Unitary invariance.} Let $U_n$ be a unitary of polynomial gate complexity. Then,
			\begin{align}
				\underline{S}(U_n \rho_nU_n^\dagger) \simeq \underline{S}(\rho_n).
			\end{align}
			
			\item \emph{Computational smoothing.}\label{item:computational_entropy_computational_smoothing} We have
			\begin{align}
				\underline{S}(\rho_n) \simeq \sup_{\tilde\rho_n \approx_c \rho_n} \underline{S}(\tilde\rho_n).
			\end{align}
		\end{propenum}
	\end{proposition}
	We emphasize that \cref{item:computational_entropy_boundedness} together with \cref{item:computational_entropy_computational_smoothing} imply that the computational entropy is always larger or equal to the \emph{pseudoentropy} $\max_{\tilde{\rho}_n\approx_c \rho_n} S(\tilde\rho_n)$. The fact that oracle separations between pseudoentropy and compressibility exist~\cite{wee2004pseudoentropy} makes us believe that this inequality is strict. 
	\begin{proof}
		The facts (i) to (iii) directly follow from the definition and  \cref{proposition:properties_of_computational_relative_entropy}. The fact (iv) follows similarly, we just have to restrict $\Phi_n\colon \calH_n \to \calH_n$ to a unital map which ensures it maps the identity operator on $\calH_n$ to itself. Fact (v) follows from the unitary invariance of the computational relative entropy (implied by isometric invariance) and the fact that the identity matrix is also unitarily invariant.
		Fact (vi) follows directly from the definition and \cref{theorem:undetectability_of_computational_smoothing}.
	\end{proof}
	
	Next, we define a compression task for quantum states under complexity constraints that works analogously to the set of POVM effects we have defined in \cref{definition:complexity_limited_measurements}. We want to capture the idea that we can compress a quantum state by applying a complexity constrained unitary and then measuring some computational registers in the computational basis. The state on the remaining registers is then the compressed state. This compressed state could then be used, e.g.\ for communication. It is expanded again by adding back the registers that were measured out and applying the inverse of the complexity-constrained unitary.
	\begin{definition}[Complexity-limited compression and expansion]
		Consider a Hilbert space $\calH$. We define the set of \emph{$G$-complexity $M$-compressors} as 
		\begin{align}
			\calC(\calH; G; M) \coloneqq \left\{ \left.\langle 0 |_{\mathrm{anc}} \otimes (\bigotimes_i A_i^{\dagger})   )  U (X \otimes |0 \rangle\!\langle 0|_{\mathrm{anc}}) U^{\dagger} (|0\rangle_{\mathrm{anc}} \otimes (\bigotimes_i A_i) \, \right| \, C(U)\leq G, A_i \in \{ |0\rangle, \bbI \}, \prod_i \Tr[A_i] \leq M \right\},
		\end{align}
		where the elements of $\calC$ are maps from $\calH$ to $\bbC^M$. Analogously, we define the set of \emph{$G$-complexity $M$-expanders} as
		\begin{align}
			\calE(\calH; G; M) \coloneqq \left\{ \left. \langle 0 |_{\mathrm{anc}}  U (X \otimes \bigotimes_i |0\rangle\!\langle 0| \otimes |0\rangle\!\langle 0|_{\mathrm{anc}}) U^{\dagger} |0\rangle_{\mathrm{anc}} \, \right| \, C(U)\leq G\right\},
		\end{align}
		where the elements of $\calE$ are maps from $\bbC^M$ to $\calH$. 
	\end{definition}
	The compression task under computational constraints can now be formalized as follows and can be seen as analogous to the definition of Yao~\cite{10.5555/1382436.1382790}.
	\begin{definition}[Complexity-limited compression task]\label{definition:computational_compression_task}
		Consider a quantum state $\rho \in \calH$. Then, the best achievable compression under complexity constraints with $\epsilon$ error is defined as
		\begin{align}
			M^{\epsilon}(\rho; G) \coloneqq \min \{ M \pipe \text{there exist } C \in \calC(\calH; G; M), E \in \calE(\calH, G; M) \text{ such that } \tfrac{1}{2} \lVert \rho - (E \circ C)[\rho] \rVert_1 \leq \epsilon \}.
		\end{align}
	\end{definition}
	Under computational complexity limitations, the best compression rate we can hope for is the following.    
	\begin{definition}[Computational compression rate]\label{definition:computational_compression_rate}
		Consider a quantum state $\rho_n \in \calH_n$ with dimension $d_n$. Then, the \emph{computational compression rate} is defined as
		\begin{align}
			\underline{C}(\rho_n) \colonsimeq \flim_{\epsilon \to 0} \flim_{\ell\to \infty} \fliminf_{k \to \infty}
			\frac{1}{n^k}\log M^{\epsilon}(\rho_n^{\otimes n^k}; n^{k\ell}).
		\end{align}
	\end{definition}
	The following single-shot characterization will allow us to deduce the operational interpretation of the computational entropy.   
	\begin{lemma}[Single-shot compression]\label{lemma:single_shot_compression}
		Consider a quantum state $\rho \in \calH$ with dimension $d$ and let $\omega$ be the maximally mixed state on $\calH$. Then,
		\begin{align}
			- D_h^{2\epsilon}(\rho \fatpipe \bbI; G) \leq 
			\log M^{\epsilon}(\rho; G) \leq - D_h^{\epsilon^2/9}(\rho \fatpipe \bbI; G).
		\end{align}
	\end{lemma}
	\begin{proof}
		For the achievability statement, we expand the definition of the complexity-limited hypothesis testing relative entropy with type I error $\epsilon_0$.
		\begin{align}
			-D_h^{\epsilon_0}(\rho \fatpipe \bbI; G) &= \log \min_{\Lambda \in \calQ(\calH; G)} \big\{ \Tr[ \Lambda ]
			\bigpipe \Tr[\Lambda \rho] \geq 1- \epsilon_0 \big\} = \log \Tr[ \Lambda^{*} ],
		\end{align}
		where we let $\Lambda^{*}$ be the optimal solution of the optimization problem.
		Right from the definition of the set of complexity-limited POVMs, $\calQ(\calH; G)$, we observe that measuring $\Lambda^{*}$ and only retaining the subsystems where it acts trivially constitutes a $G$-complexity $\Tr[\Lambda^{*}]$-compressor. Equivalently, adding back the subsystems prepared in the non-trivial computational basis projectors of $\Lambda^{*}$ and applying the inverse of the associated unitary gives us a $G$-complexity $\Tr[\Lambda^{*}]$-expander. Let us now see how $\epsilon_0$ relates to the error of this combination of compressor and expander.
		
		By construction of our computational model, we can implement $\Lambda^*$ by appending auxiliary systems in the all-zero state, performing a $G$-complexity unitary and then projecting some of the systems onto the all-zero state $|0\rangle\!\langle 0|$. We write the state after the application of the unitary as
		\begin{align}
			\tilde\rho_{AB} &\coloneqq V\rho V^{\dagger} = U(\rho \otimes |0\rangle\!\langle 0|_{\mathrm{anc}})U^{\dagger}
		\end{align}
		and the POVM effect realizing $\Lambda^*$ after the transformation as
		\begin{align}
			\Lambda_{AB} \coloneqq |0\rangle\!\langle 0 |_A \otimes \bbI_B,
		\end{align}
		such that $\Tr[\Lambda_{AB} \tilde\rho_{AB}] \geq 1-\epsilon_0$ by construction. Measuring $\Lambda_{AB}$ then implements the map
		\begin{align}
			\tilde\rho_{AB} \mapsto \Tr[ \Lambda_{AB} \tilde\rho_{AB}] \rho_{AB}^{\Lambda} + (1-\Tr[\Lambda_{AB}\tilde\rho_{AB}]) \rho_{AB}^{\bbI - \Lambda}.
		\end{align}
		The gentle measurement Lemma~\cite{odonnell2021quantum} implies that
		\begin{align}
			\lVert \tilde\rho_{AB} - \rho_{AB}^{\Lambda}\rVert_1 \leq \sqrt{\epsilon_0}.
		\end{align}
		The associated expander amounts to adding $|0\rangle\!\langle 0|$ and applying the inverse of the isometry $V$. We have the performance guarantee
		\begin{align}
			\lVert \rho_{AB} - (E \circ D)[\rho_{AB}]\rVert_1 &= \lVert \tilde\rho_{AB} - \Tr[ \Lambda_{AB} \tilde\rho_{AB}] \rho_{AB}^{\Lambda} - (1-\Tr[\Lambda_{AB}\tilde\rho_{AB}]) \rho_{AB}^{\bbI - \Lambda}\rVert_1 \\
			&\leq \lVert \tilde\rho_{AB} - \rho_{AB}^{\Lambda}\rVert_1 + (1 - \Tr[\Lambda_{AB}\tilde\rho_{AB}]) (\lVert \rho_{AB}^{\Lambda} \rVert_1 + \lVert \rho_{AB}^{\bbI - \Lambda}\rVert_1) \nonumber \\
			&\leq \sqrt{\epsilon_0} + 4 \epsilon_0  \nonumber \\
			&\leq 6 \sqrt{\epsilon_0} \nonumber
		\end{align}
		To achieve the necessary guarantee for the compression task, we need the right-hand side to be smaller than $2\epsilon$. We achieve this by choosing $\epsilon_0 = \epsilon^2/9$, which finishes the achievability part.
		The converse direction is immediate, as the existence of a compressor guarantees that there exists a $G$-complexity unitary $U$ such that
		\begin{align}
			\frac{1}{2}\lVert U(\rho \otimes |0\rangle\!\langle 0|_{\mathrm{anc}}) U^{\dagger} - |0\rangle\!\langle 0|_A \otimes \tilde\rho_B \otimes |0\rangle\!\langle 0|_{\mathrm{anc}} \rVert_1 \leq \epsilon.
		\end{align}
		Hence, choosing the POVM effect 
		\begin{align}
			\Lambda = U^{\dagger} (|0\rangle\!\langle 0|_A \otimes \bbI_B \otimes |0\rangle\!\langle 0|_{\mathrm{anc}} )U
		\end{align}
		guarantees that $\Tr[\Lambda] \leq M^{\epsilon}(\rho;G)$ and
		\begin{align}
			\Tr[\rho \Lambda] &= \Tr[(\rho \otimes |0\rangle\!\langle 0|_{\mathrm{anc}})(\Lambda \otimes |0\rangle\!\langle 0|_{\mathrm{anc}})] \\
			\nonumber
			&=\Tr[U(\rho \otimes |0\rangle\!\langle 0|_{\mathrm{anc}})U^{\dagger} U(\Lambda \otimes |0\rangle\!\langle 0|_{\mathrm{anc}})U^{\dagger}] 
			\\
			\nonumber
			&\geq  \Tr[(|0\rangle\!\langle 0|_A \otimes \tilde\rho_B \otimes |0\rangle\!\langle 0|_{\mathrm{anc}}) (|0\rangle\!\langle 0|_A \otimes \bbI_B \otimes |0\rangle\!\langle 0|_{\mathrm{anc}} )] - 2\epsilon 
			\\
			\nonumber
			&= 1- 2\epsilon.
		\end{align}
		From this, we conclude the converse direction.
	\end{proof}
	
	We are now ready to show the main result of this section.
	\begin{theorem}[Computational entropy as compression rate]\label{theorem:computational_entropy_as_compression_rate}
		Consider a quantum state $\rho_n \in \calH_n$. Then,
		\begin{align}
			\underline{C}(\rho_n) \simeq \underline{S}(\rho_n).
		\end{align}
	\end{theorem}
	The computational entropy indeed has the same interpretation in the complexity-limited setting as the unbounded entropy in the unbounded complexity setting. This underscores the sensibility of using the computational relative entropy as a mother quantity to derive other computational information theoretic quantities from. 
	\begin{proof}
		This result is an immediate corollary of \cref{lemma:single_shot_compression}. Letting $d_n$ be the dimension of $\rho_n$ and regularizing over $n^k$ copies and setting $G = n^{k\ell}$ yields
		\begin{align}
			- \frac{1}{n^k}D_h^{2\epsilon}(\rho_n^{\otimes n^k} \fatpipe \bbI_n^{\otimes n^k}; n^{k\ell}) \leq \frac{1}{n^k}
			\log M^{\epsilon}(\rho_n^{\otimes n^k}; n^{k\ell}) \leq - \frac{1}{n^k}D_h^{\epsilon^2/9}(\rho_n^{\otimes n^k} \fatpipe \bbI_n^{\otimes n^k}; n^{k\ell}).
		\end{align}
		Taking the lim inf in $k$ and lim in $\ell$ gives the associated regularized quantities. The subsequent limit $\epsilon \to 0$ recovers the compression rate in the middle and the computational relative entropy on the left and the right hand side, as $\epsilon \to 0$ implies $\epsilon^2/9\to 0$ as well.
	\end{proof}
	
	We can build on the computational Stein's lemma of \cref{theorem:computational_steins_lemma} to connect the computational entropy to \emph{computational two-outcome measured entropies}. We can define those analogously to how we have defined the computational two-outcome measured divergences in \cref{section:computationally_measured_divergences}.
	\begin{definition}[$G$-$K$-complexity measured quantum entropy]\label{def:g_k_complexity_measured_quantum_entropy}
		Consider a classical divergence measure between probability distributions $\mathbf{D}(p\fatpipe q) \geq 0$. The associated $G$-$K$-complexity measured entropy of a state $\rho\in\calH$ is defined as
		\begin{align}
			\mathbf{S}^{\bbC}(\rho; G; K) \coloneqq -\mathbf{D}^{\bbC}(\rho \fatpipe \bbI;G;K),
		\end{align}
		where $\bbI$ is the identity operator on $\calH$.
	\end{definition}
	Taking the polynomial regularization over computationally efficient two-outcome measurements yields the associated definition of a \emph{computationally measured two-outcome quantum entropy}.
	\begin{definition}[Computational two-outcome measured entropy]\label{def:computational_two_outcoe_measured_entropy}
		Consider a classical divergence measure $\mathbf{D}(p\fatpipe q) \geq 0$. The \emph{computationally measured two-outcome quantum entropy} associated to $\mathbf{S}$ for a state $\rho_n \in \calH_n$ is defined as
		\begin{align}
			\underline{\mathbf{S}}^{\mathbbm{2}}(\rho_n) \colonsimeq -\underline{\mathbf{D}}^{\mathbbm{2}}(\rho_n \fatpipe \bbI_n) &\simeq \flim_{\ell \to \infty} \fliminf_{k\to\infty} \frac{1}{n^k} \mathbf{S}^{\bbC}(\rho_n^{\otimes n^k}; n^{k\ell}; 2) 
		\end{align}
	\end{definition}
	We can establish the following properties.
	\begin{proposition}[Properties of computationally measured two-outcome quantum entropy]
		The computationally measured two-outcome quantum entropy has the following properties:
		\begin{propenum}				
			\item \emph{Additivity under tensor products.} Let $m\leq O(\poly(n))$. Then,
			\begin{align}
				\underline{\mathbf{S}}^{\mathbbm{2}}(\rho_n^{\otimes m}) \simeq m \underline{\mathbf{S}}^{\mathbbm{2}}(\rho_n).
			\end{align}
			
			\item \emph{Data processing.} Let $\Phi_n$ be a unital quantum channel of polynomial gate complexity. Then,
			\begin{align}
				\underline{\mathbf{S}}^{\mathbbm{2}}(\Phi_n[\rho_n]) \gtrsim \underline{\mathbf{S}}^{\mathbbm{2}}(\rho_n).
			\end{align}
		\end{propenum}
	\end{proposition}
	\begin{proof}
		Properties (i) and (ii) follow directly from the definition and \cref{proposition:properties_of_computational_two_outcome_divergence}. We only have to restrict the quantum channels in fact (ii) to unital channels to preserve the identity matrix in the second argument of the measured divergence. 
	\end{proof}
	
	We can build on the results of the previous sections, namely the computational Stein's Lemma as well as the relations between computational relative entropy and the computationally measured Rényi relative entropies. We start with the result implied by the computational Stein's Lemma.
	\begin{theorem}[Computational entropy as measured entropy]
		Consider a state $\rho_n \in \calH_n$ with $d_n = \dim(\calH_n)$ such that $\log d_n \leq O(\poly(n))$. Then,
		\begin{align}
			\underline{S}(\rho_n) \simeq \underline{S}^{\mathbbm{2}}(\rho_n)
		\end{align}
	\end{theorem}
	\begin{proof}
		This result directly follows from \cref{theorem:computational_steins_lemma}. We just have to make sure that the condition of polynomial logarithmic dimension implies the states bound on the computationally measured two-outcome max relative entropy. We have
		\begin{align}
			\underline{D}_{\max}^{\mathbbm{2}}(\rho_n \fatpipe \omega_n) &\simeq \flim_{\ell \to \infty} \fliminf_{k \to \infty} \max_{\calM \in \calM(\calH_n^{\otimes n^k}; n^{k\ell})} \frac{1}{n^k} D_{\max}(\calM[\rho_n^{\otimes n^k}] \fatpipe \calM[\omega_n^{\otimes n^k}]) \\
			&\lesssim \flim_{\ell \to \infty} \fliminf_{k \to \infty} \frac{1}{n^k} D_{\max}(\rho_n^{\otimes n^k} \fatpipe \omega_n^{\otimes n^k}) \nonumber \\
			&= D_{\max}(\rho_n \fatpipe \omega_n) \nonumber \\
			&= \log d_n - S_{\min}(\rho_n) \nonumber \\
			&\leq \log d_n, \nonumber
		\end{align}
		which is polynomially bounded by assumption.
	\end{proof}
	We continue with the result connecting the computational relative entropy to the computational measured Rényi relative entropies.
	\begin{theorem}[Computational entropy bounds by Rényi entropies]
		Consider a state $\rho_n \in \calH_n$ and let $0 < \alpha_{-} < 1 < \alpha_{+}$ be independent of $n$. Then, 
		\begin{align}
			\underline{S}^{\mathbbm{2}}_{\alpha_+}(\rho_n)\lesssim\underline{S}(\rho_n) \lesssim \underline{S}^{\mathbbm{2}}_{\alpha_-}(\rho_n)    
		\end{align}
	\end{theorem}
	\begin{proof}
		This is a direct consequence of \cref{theorem:renyi_upper_and_lower_bounds_computational_relative_entropy}.
	\end{proof}
	
	\subsection*{Separations between computational and unbounded entropy}
	
	We can use techniques similar to the ones employed in Ref.~\cite{Bansal_2025} to establish an unconditional incompressibility result which corresponds to a separation between the computational and unbounded entropy.
	\begin{theorem}[Unconditionally incompressible states with low entropy]\label{theorem:hilbert_schmidt_ensemble_states_vs_maximally_mixed}
		With no computational assumptions, there exists a quantum state $\rho_n \in (\bbC^2)^{\otimes n}$ which is computationally incompressible, i.e.,  $\underline{S}(\rho_n) \simeq n \log 2$, but whose unbounded entropy grows as $S(\rho_n) \leq \log^{1+\delta} (n) \log2$ for any $\delta > 0$.
	\end{theorem}
	\begin{proof}
		We consider the Hilbert-Schmidt ensemble of quantum states, which is constructed by sampling Haar-random pure states on $n + m$ qubits and tracing out the last $m$ qubits. Formally, we define
		\begin{align}
			\rho \sim \haar(n,m) \ \Leftrightarrow \
			\rho = \Tr_m[\psi], \psi \sim \haar(n+m),
		\end{align}
		where $\haar(k)$ denotes the Haar measure over $(\mathbb{C}^2)^{\otimes k}$. It can be shown that~\cite{Bansal_2025}
		\begin{align}\label{eqn:hs_ensemble_mean_is_close_to_maximally_mixed}
			\left\|\operatornamewithlimits{\mathbb{E}}_{\rho\sim \haar(n,m)}\rho^{\otimes k}-\omega^{\otimes k}\right\|_{1}\le O\left(\frac{k^2}{2^m}\right)\,,
		\end{align}
		and thus, by choosing $m = \omega(\log n)$, the expected value of $k$ powers of the distribution $\haar(n,m)$ becomes statistically indistinguishable from $k$ powers of the maximally mixed state.
		
		The rest of the argument is analogous to the proof of \cref{lemma:quantum-hard-exponential-complexity}, only that the POVM effect $\Lambda$ can only act on $n$ of the $n+m$ total qubits in question. After the analysis, the error is dominated by \cref{eqn:hs_ensemble_mean_is_close_to_maximally_mixed} where we get a negligible decay for $m = \omega(\log n)$, which is sufficient for our purposes but not as fast as $d^{-1/4}$ we get from the rest of the analysis. Nevertheless, we can conclude that there exists a state $\rho_n \approx_c \omega_n$. 
		
		This means that the computational entropy is maximal due to computational smoothing (\cref{item:computational_entropy_computational_smoothing}), i.e., $\underline{S}(\rho_n) \simeq \underline{S}(\omega_n) \simeq n \log 2$. On the other hand, as the state $\rho_n$ is obtained by tracing out $m$ qubits of a pure state, we know that $S(\rho_n) \leq m \log 2 = \omega(\log n) \log 2$, concluding the argument.
	\end{proof}
	
	\section{Complexity-limited entanglement theory}
	Entanglement is a purely quantum phenomenon, the manipulation of which is of pivotal importance in quantum applications \cite{PhysRevLett.69.2881,PhysRevLett.70.1895,Pirandola_2020}. Therefore, quantifying the limits of entanglement manipulation lies at the core of many applications of quantum information theory. As such, it is an important proving ground for computational quantum information theory. In this section, we will define computational analogues of distillable entanglement and entanglement cost and show that we can reproduce existing results from the literature \cite{leone2025entanglementtheorylimitedcomputational,arnon2023computationalentanglementtheory} and derive new ones. In particular, we give a computational analogue of the Rains bound \cite{rains1999bound} on distillable entanglement and relate the entanglement cost to the computational entropy.
	
	\subsection*{Computational distillable entanglement}
	
	To define a computational version of distillable entanglement, we follow the same strategy of first defining a complexity-constrained version of single-shot distillable entanglement and then performing a polynomial regularization. In the following definition, $O(\calH\otimes \calK)$ is a set of operations on a bipartite Hilbert space, like \emph{local operations and classical communication}  (LOCC) or \emph{positive-partial transpose preserving}  (PPT), and $O(\calH \otimes \calK; G)$ denotes the subset of operations in $O$ that have gate complexity at most $G$. In the following definition, we use the fidelity
	\begin{align}
		F(\rho, \sigma) \coloneqq \lVert \rho^{1/2} \sigma^{1/2}\rVert_1
	\end{align}
	as a figure of merit for the protocols.
	\begin{definition}[$G$-complexity distillable entanglement]\label{def:g_complexity_distillable_entanglement}
		Consider a bipartite state $\rho^{AB} \in \calH \otimes \calK$. We define the $G$-complexity $\epsilon$-distillable entanglement under the class of operations $O$ as 
		\begin{align}
			E_{D,O}^{\epsilon}(\rho^{AB}; G) &\coloneqq \log2 \times \max \big\{ m \in \bbN \pipe \text{there exists } \Phi \in O(\calH\otimes\calK; G) \text{ such that } F(\Phi[\rho^{AB}], \phi^{\otimes m}) \geq 1-\epsilon \big\},
		\end{align}
		where $\phi$ is the two-qubit maximally entangled state and the factor $\log 2$ converts from bits to nats.
	\end{definition}   
	The maximum above is guaranteed to exist, as $m = 0$ is always admissible (no distillable entanglement) and $m \leq \log_2 \operatorname{dim}(\rho^{AB})$ as we cannot distill more than this number of ebits.
	If we do not specify the class of operations $O$ explicitly, we mean the class $\mathrm{LOCC}$, which is a subset of $\mathrm{PPT}$. 
	We note that in general settings, the choice of maximally entangled reference state is arbitrary, as the maximally entangled pure target states can be freely converted into each other. This is not necessarily the case in a complexity constrained setting. We state our definition with qubit maximally entangled states, also known as \emph{ebits}, as they have low complexity and are the most common reference state.
	Based on this definition, we can define the \emph{computational distillable entanglement rate} by performing a polynomial regularization. Note that we go about the definitions in a slightly different way than Ref.~\cite{arnon2023computationalentanglementtheory}, but that does not qualitatively change the resulting statements.      
	\begin{definition}[Computational distillable entanglement rate]\label{definition:computational_distillable_entanglement_rate}
		Consider a bipartite state $\rho_n^{A_nB_n} \in \calH_n \otimes \calK_n$. Then,
		\begin{align}
			\underline{E}_{D, O}(\rho_n^{A_n B_n}) \colonsimeq \flim_{\epsilon \to 0}\flim_{\ell\to\infty}\fliminf_{k\to\infty} \frac{1}{n^k} E_{D, O}^{\epsilon}( (\rho_n^{A_n B_n})^{\otimes n^k}; n^{k\ell}).
		\end{align}    
	\end{definition}
	Again, if the class of operations is not specified, we assume it is $\mathrm{LOCC}$.
	
	We note the following achievability statement proven by Leone et al.~\cite{leone2025entanglementtheorylimitedcomputational} which we present in the language of our work. It relates a computational quantity to a purely entropic quantity.
	\begin{theorem}[Distilling the min entropy~\cite{leone2025entanglementtheorylimitedcomputational}]\label{theorem:computational_distillable_entanglement_lower_bound_min_entropy}
		
		Consider a bipartite pure state $\psi_n^{A_n B_n} \in \calH_n \otimes \calK_n$. The min entropy of the reduced state is defined as
		\begin{align}
			S_{\min}(\psi_n^{A_n}) \coloneqq -\log \lVert \psi_n^{A_n}\rVert_{\infty}.
		\end{align}
		There exists a computationally efficient, state-agnostic protocol that distills entanglement at the rate of the min entropy,
		\begin{align}
			\underline{E}_D(\psi_n^{A_n B_n}) \gtrsim \frac{1}{20} S_{\min}(\psi_n^{A_n}),
		\end{align}            
		as long as $S_{\min}(\psi_n^{A_n}) \leq O(\log n)$.
	\end{theorem}
	\begin{proof}
		Ref.~\cite[Theorem 18]{leone2025entanglementtheorylimitedcomputational} shows a lower bound on the computationally distillable entanglement from $\ell$ copies as $\ell \times \min\{ S_{\min}(\psi_n^{A_n}), \log \ell\}$. As we take $\ell = n^k$ and let $k \to \infty$, we will sooner or later have that the term $\log \ell = k \log n$ exceeds $S_{\min}(\psi_n^{A_n}) \leq q \log n$ for some $q$, which is why it does not appear in the lower bound.
	\end{proof}
	
	We now establish a novel result, the computational Rains bound. To do so, we build on the proof strategy of Ref.~\cite{fang2019non-asymptotic} to give an upper bound on the $G$-complexity distillable entanglement for the class PPT, which in turn will allow us to give an upper bound for the class LOCC. Polynomial regularization of this single-shot bound then yields the computational Rains bound.		
	
	Let $\Phi\colon \calH \otimes \calK \to \calH' \otimes \calK'$ be a quantum channel from $O(\calH \otimes \calK;G)$ and $\ket{\Omega} \in (\calH \otimes \calK)^{\otimes 2}$ be the canonical maximally entangled state. We then have the Choi state of the channel
	\begin{align}
		J_{\Phi} \coloneqq (\bbI \otimes \Phi)\big[|\Omega\rangle\!\langle \Omega |\big] \in (\calH \otimes \calK) \otimes(\calH' \otimes \calK').
	\end{align}
	We identify the two constituent systems as $AB$ and $A'B'$ such that $J_{\Phi} = J_{\Phi}^{ABA'B'}$.
	Given the Choi state $J_{\Phi}$, we can naturally cast the definition of distillable entanglement as an optimization over Choi states instead of channels. We use this to prove the following converse bound.
	\begin{proposition}[$G$-complexity PPT distillable entanglement converse]\label{proposition:computational_ppt_distillable_entanglement_single_shot_converse}
		Consider a bipartite state $\rho^{AB} \in \calH \otimes \calK$ with $d = \min\{ \dim(\calH), \dim(\calK)\}$. 
		We have that
		\begin{align}
			E_{D,\mathrm{PPT}}^{\epsilon}(\rho^{AB}; G) &\leq \bigg\lfloor\min_{\substack{P^{AB} = (P^{AB})^{\dagger}\\\lVert (P^{AB})^{T_B} \rVert_1 \leq 1}} D_h^{\epsilon}(\rho^{AB} \fatpipe P^{AB}; G + O(\poly(\log d)\log \tfrac{1}{\kappa}) ) - \log (1+\kappa) \bigg\rfloor
		\end{align}
		for all $\kappa$.
	\end{proposition}
	\begin{proof}
		We closely follow the proofs of Ref.~\cite[Proposition 3 and Theorem 4]{fang2019non-asymptotic}. 
		We first define a restricted version of $\mathrm{PPT}(\calH \otimes \calK; G)$ as
		\begin{align}\label{eqn:ppt_bar_G_definition}
			\begin{split}        
				&\overline{\mathrm{PPT}}(\calH \otimes \calK; G) \coloneqq \\
				&\qquad\quad \big\{ \Phi \in \mathrm{PPT}(\calH \otimes \calK; G) \pipe \text{exists } m\in\bbN \text{ such that }J_\Phi = W^{AB}\otimes (\phi^{\otimes m})^{A'B'} + Q^{AB}\otimes (\bbI - \phi^{\otimes m})^{A'B'}\big\},
			\end{split}
		\end{align}
		where $J_\Phi$ is the Choi state associated to the operation $\Phi$ and $\phi$ is the canonical maximally entangled state on two qubits. This corresponds to operations that are invariant under a specific twirling on the output that leaves $\phi^{\otimes m}$ invariant. In the unrestricted setting, it is sufficient to optimize over such operations.
		We know from the definition of the Choi state that $\Tr^{A'B'}[J_{\Phi}] = \bbI^{AB}$, which is a consequence of $\Phi$ being trace preserving. This enforces
		\begin{align}\label{eqn:tau_in_terms_of_sigma}
			W^{AB} + Q^{AB} (d'^2 - 1) = \bbI^{AB} \ \Leftrightarrow \ Q^{AB} = \frac{\bbI^{AB} - W^{AB}}{d'^2 - 1},
		\end{align}
		where $d' = 2^m$ is the dimension of systems $A'$ and $B'$.
		Moreover, since $J_{\Phi}$ is a valid quantum state, we must have that $0 \leq W^{AB} \leq \bbI$, indicating that $W^{AB}$ can be seen as a POVM effect. We thus have that $\Phi$ can be written in the following way
		\begin{align}
			\Phi[X^{AB}] = \Tr[(W^{AB})^T X^{AB}] (\phi^{\otimes m})^{A'B'} + \Tr[ (\bbI - W^{AB})^T X^{AB}] \frac{\bbI - \phi^{\otimes m}}{d'^2 - 1}.
		\end{align}
		This means we measure 
		the POVM $\{ (W^{AB})^T, (\bbI - W^{AB})^T\}$ and then either prepare $\phi^{\otimes m}$ or the maximally mixed state on the complement of $\phi^{\otimes m}$. 
		Let us now see how the fact that $\Phi$ must be a PPT operation constrains $W^{AB}$. This enforces that $J_\Phi$ has positive partial transpose with respect to the system $BB'$. We use the fact that
		\begin{align}
			(({\phi}^{\otimes m})^{A'B'})^{T_{B'}} &= \frac{1}{d'}\left(\Pi_{\mathrm{sym}}^{A'B'} - \Pi_{\mathrm{asym}}^{A'B'}\right),
		\end{align}
		where $\Pi_{\mathrm{sym}}^{A'B'}$ and $\Pi_{\mathrm{asym}}^{A'B'}$ are the symmetric and antisymmetric projectors relative to the bipartition $A'B'$.
		We can thus write
		\begin{align}
			J_{\Phi}^{T_{BB'}} &= (W^{AB})^{T_B} \otimes (\tilde{\phi}^{A'B'})^{T_{B'}} + (Q^{AB})^{T_B} \otimes ((\bbI - \tilde\phi)^{A'B'})^{T_{B'}} \\
			\nonumber
			&= \frac{1}{d'}(W^{AB})^{T_B} \otimes [\Pi_{\mathrm{sym}}^{A'B'} - \Pi_{\mathrm{asym}}^{A'B'}] + (Q^{AB})^{T_B} \otimes [(1 - \tfrac{1}{d'})\Pi_{\mathrm{sym}}^{A'B'} + (1 + \tfrac{1}{d'}) \Pi_{\mathrm{asym}}^{A'B'}] \\
			\nonumber
			&= \frac{1}{d'}\left((d'-1) (Q^{AB})^{T_B} +(W^{AB})^{T_B} \right) \otimes \Pi_{\mathrm{sym}}^{A'B'} + \frac{1}{d'}\left((d'+1) (Q^{AB})^{T_B}-(W^{AB})^{T_B} \right) \otimes \Pi_{\mathrm{asym}}^{A'B'},
			\nonumber
		\end{align}
		where we have additionally used that $\bbI^{A'B'} = \Pi_{\mathrm{sym}}^{A'B'} + \Pi_{\mathrm{asym}}^{A'B'}$. From the above it follows that $J_{\Phi}^{T_{BB'}} \geq 0$ if and only if
		\begin{align}
			(d'-1) (Q^{AB})^{T_B} +(W^{AB})^{T_B} &\geq 0 ,\\
			(d'+1) (Q^{AB})^{T_B}-(W^{AB})^{T_B}&\geq 0.
		\end{align}
		Inserting \cref{eqn:tau_in_terms_of_sigma} then yields
		\begin{align}
			\frac{(d'-1)}{d'^2-1}\left[ \bbI - (W^{AB})^{T_B}\right] +(W^{AB})^{T_B} 
			&\geq 0,\\
			\frac{(d'+1)}{d'^2-1}\left[ \bbI - (W^{AB})^{T_B}\right]-(W^{AB})^{T_B}&\geq 0.
		\end{align}
		Using that $d'^2-1=(d'-1)(d'+1)$ and rearranging leads to the final condition
		\begin{align}
			-\frac{\bbI}{d'} \leq (W^{AB})^{T_B} \leq \frac{\bbI}{d'} \ \Leftrightarrow \ \lVert (W^{AB})^{T_B} \rVert_{\infty} \leq \frac{1}{d'}.
		\end{align}
		Let us now analyze the performance of the above protocols for entanglement distillation, i.e.,  the quantity $E^{\epsilon}_{D, \overline{\mathrm{PPT}}}(\rho^{AB};G)$. We rewrite the optimization objective using the Choi state as
		\begin{align}
			F(\Phi[\rho^{AB}],(\phi^{\otimes m})^{A'B'}) = \Tr[ \Phi[\rho^{AB}] (\phi^{\otimes m})^{A'B'}] = \Tr[ J_\Phi((\rho^{AB})^T \otimes (\phi^{\otimes m})^{A'B'})]
			&= \Tr[ W^{AB} (\rho^{AB})^T],
		\end{align}
		which in turn implies we can write
		\begin{align}
			\begin{split}
				&E^{\epsilon}_{D, \overline{\mathrm{PPT}}}(\rho^{AB};G) \\
				&\quad = \left\lfloor \max \big\{ \log \frac{1}{\eta} \pipe 0 \leq W^{AB} \leq \bbI, \Tr[W^{AB} (\rho^{AB})^T] \geq 1- \epsilon, \lVert (W^{AB})^{T_B} \rVert_{\infty} \leq \eta, C(\Phi(W^{AB})) \leq G \big\} \right\rfloor,
			\end{split}
		\end{align}
		where we have defined $\Phi(W^{AB})$ to be the quantum channel whose Choi matrix is
		\begin{align}
			J_\Phi = W^{AB} \otimes (\phi^{\otimes m})^{A'B'} + \frac{1}{d'^2 - 1} (\bbI^{AB} - W^{AB}) \otimes (\bbI - \phi^{\otimes m})^{A'B'}.
		\end{align}
		Note that the factor $\log 2$ already disappeared because $d' = 2^m \Leftrightarrow \log d' = m\log 2$. 
		Let us denote with 
		\begin{align}
			\calW(G) \coloneqq \big\{ 0 \leq W \leq \bbI \pipe C(\Phi(W))\leq G\big\}
		\end{align}
		the set of POVM effects $W$ arising from $G$-complexity operations $\Phi$. We need to connect $\calW(G)$ to $\calQ(\calH\otimes \calK; G')$ for some value $G'$.			
		Then
		\begin{align}
			&E^{\epsilon}_{D, \overline{\mathrm{PPT}}}(\rho^{AB};G)
			\\
			\nonumber
			&= \bigg\lfloor -\log  \min_{W^{AB} \in \calW(G)} \big\{\lVert (W^{AB})^{T_B} \rVert_{\infty} \pipe \Tr[W^{AB} (\rho^{AB})^T] \geq 1- \epsilon \big\} \bigg\rfloor \\
			\nonumber
			&= \bigg\lfloor -\log  \min_{W^{AB} \in \calW(G)} \big\{\lVert (W^{AB})^{T_B} \rVert_{\infty} \pipe \Tr[W^{AB} \rho^{AB}] \geq 1- \epsilon \big\} \bigg\rfloor \\
			\nonumber
			&= \bigg\lfloor -\log  \min_{W^{AB} \in \calW(G)} \max_{\substack{P^{AB} = (P^{AB})^{\dagger}\\\lVert P^{AB} \rVert_1 \leq 1}}\big\{   \Tr[ (W^{AB})^{T_B} P^{AB} ] \pipe\Tr[W^{AB} \rho^{AB}] \geq 1- \epsilon \big\} \bigg\rfloor\\
			\nonumber
			&\leq \bigg\lfloor -\log  \max_{\substack{P^{AB} = (P^{AB})^{\dagger}\\\lVert P^{AB} \rVert_1 \leq 1}} \min_{W^{AB} \in \calW(G)} \big\{  \Tr[ (W^{AB})^{T_B} P^{AB} ] \pipe\Tr[W^{AB} \rho^{AB}] \geq 1- \epsilon \big\} \bigg\rfloor\\
			\nonumber
			&= \bigg\lfloor -\log  \max_{\substack{P^{AB} = (P^{AB})^{\dagger}\\\lVert (P^{AB})^{T_B} \rVert_1 \leq 1}}\min_{W^{AB} \in \calW(G)}  \big\{  \Tr[ W^{AB} P^{AB} ] \pipe\Tr[W^{AB} \rho^{AB}] \geq 1- \epsilon, \big\} \bigg\rfloor\\
			\nonumber
			&= \bigg\lfloor \min_{\substack{P^{AB} = (P^{AB})^{\dagger}\\\lVert (P^{AB})^{T_B} \rVert_1 \leq 1}}  -\log \min_{W^{AB} \in \calW(G)}   \big\{\Tr[ W^{AB} P^{AB} ] \pipe\Tr[W^{AB} \rho^{AB}] \geq 1- \epsilon \big\} \bigg\rfloor.
			\nonumber
		\end{align}
		The first equality identifies $\eta = \lVert (W^{AB})^{T_B}\rVert_{\infty}^{-1}$. The second equality replaces the optimization variable with $W^{AB} \mapsto (W^{AB})^{T}$. This maps $(W^{AB})^{T_B} \to (W^{AB})^{T_A}$, which we can undo to obtain the present form as the operator norm is invariant under transposition, which also guarantees that the POVM effect condition is not affected. The third equality uses the duality of operator and trace norm to expresses the operator norm of $(W^{AB})^{T_B}$ as an optimization over $P^{AB}$. The first inequality exchanges maximum and minimum using the minimax inequality $\min \max \geq \max \min$. The fourth equality replaces the optimization over $P^{AB}$ with an optimization over $(P^{AB})^{T_B}$. The last equality just pulls the maximization out of the logarithm, where the minus sign turns it into a minimization. 
		
		The above optimization over $W^{AB}$ is already very close to to the one that defines the hypothesis testing relative entropy. We are left to relate $\calW(G)$ to $\calQ(\calH\otimes\calK; G')$ for some $G'$. This means we have to understand how the gate complexity of $\Phi$, which by construction is $C(\Phi)\leq G$ relates to the gate complexity of $(W^{AB})^T$ seen as a POVM effect. We first note that if we 
		\begin{equation}
			W^{AB} = \langle 0|_{\mathrm{anc}} [U^{\dagger} (\bigotimes_i A_i) U]|0\rangle_{\mathrm{anc}} 
		\end{equation}
		(see \cref{definition:complexity_limited_measurements}), then $(W^{AB})^T = \langle 0|_{\mathrm{anc}} [U^{T} (\bigotimes_i A_i)^T U^{*}]|0\rangle_{\mathrm{anc}}$. In our circuit model, the operators $A_i$ are symmetric, i.e., $A_i = A_i^T$, which means that $(W^{AB})^T$ is implemented exactly as $W^{AB}$, only replacing $U$ with $U^{*}$. As $U^*$ can be implemented using the complex conjugates of all involved gates, it has the same gate complexity as $U$ and, therefore, 
		\begin{equation}
			C(W^{AB}) = C((W^{AB})^T).
		\end{equation}
		
		Next, we need the complexity of preparing the states $\phi^{\otimes m}$ and $(d'^2-1)^{-1}(\bbI - \phi^{\otimes m})$ from the all-zero state. We have that $C(\phi^{\otimes m}) = m$, because each copy can be prepared by a Hadamard gate followed by a CNOT gate, which can be absorbed into one two-qubit gate. 
		
		To understand the complexity of preparing the second state, we first establish the complexity of preparing the maximally mixed state on $2m$ qubits, $\omega_{2m}$. Clearly, we can implement a maximally mixed state on $2m$ qubits from the all-zero state by first applying a Hadamard gate to every qubit and then measuring in the computational basis and forgetting the result. We can also implement this computational-basis measurement on a qubit by applying a CNOT with a an auxiliary qubit which is then traced out. As we can absorb the CNOT and the Hadamard gate into a single two-local unitary and tracing out is considered a free operation in our computational model, this establishes that $C(\omega_{2m}) \leq 2m$.

		The state we actually want to prepare, $(d'^2-1)^{-1}(\bbI - \phi^{\otimes m})$, is nearly a maximally mixed state, we just need to \enquote{cut out} $\phi^{\otimes m}$ from the distribution. As we just established that $\phi^{\otimes m}$ can be easily prepared from the all-zero state with complexity at most $m$, we will instead look at the complexity of preparing $(d'^2-1)^{-1}(\bbI - |0\rangle\!\langle 0|^{\otimes 2m})$, mindful of the fact that their complexity can only differ by at most $m$.

		Our strategy to prepare the state $(d'^2-1)^{-1}(\bbI - |0\rangle\!\langle 0|^{\otimes 2m})$ follows the strategy to prepare the maximally mixed state, where we prepared an equal superposition over all states, measured in the computational basis and forgot the result. The only change that we make upon observing the outcome $\ket{0}^{\otimes 2m}$, we do not forget the result, but instead try again. As our system is anyways in the state $\ket{0}^{\otimes 2m}$, we apply our procedure for preparing the maximally mixed state again, measure again and then forget the result if the outcome was not $\ket{0}^{\otimes 2m}$. Imagine we perform this process $k$ times and at the last iteration we do not restart if we observe the outcome $\ket{0}^{\otimes 2m}$. Then, the output distribution is very close to the desired state with an error scaling as $2^{-2km}$, which quantifies the probability of observing $\ket{0}^{\otimes 2m}$ in every single one of the $k$ rounds. As such, we prepared the state
		\begin{align}
			(1-2^{-2km})\frac{\bbI - |0\rangle\!\langle 0|^{\otimes 2m}}{d'^2 -1} + 2^{-2km} \frac{\bbI}{d'^2},
		\end{align}
		which is good enough for our purposes.
		
		The computational complexity of the process we just outlined is certainly polynomial in $k$ and $m$, simply because the conditional re-running of the circuit can be implemented with (multi-)controlled Hadamard and Toffoli gates, which can be implemented with $O(\poly(mk))$ gates by combining~\cite[Lemma 7.9]{barenco1995elementary} and \cite[Theorem 2]{he2017decompositions}.
		
		We are left to understand how the imperfect nature of our implementation affects the distillable entanglement.
		This essentially amounts to a poisoning of the POVM effect $W^{AB}$ by $\bbI - W^{AB}$ proportional to $2^{-2(k+1)m}$, which reduces the overall amount of distillable entanglement, as we desire to minimize the operator norm of $W^{AB}$. Denote the such obtained approximate operator as $\tilde{W}^{AB}$, then
		\begin{align}
			\log \lVert \tilde{W}^{AB}\rVert_{\infty} \leq \log (\lVert W^{AB}\rVert_{\infty} + 2^{-2(k+1)m}).
		\end{align}
		Say we want that $2^{-2(k+1)m} \leq \kappa \lVert W^{AB} \rVert_{\infty}$, then $k = -\frac{1}{2m}\log_2 \kappa \lVert W^{AB} \rVert_{\infty} - 1 \geq -\frac{1}{2m}\log_2 \kappa 2^{-m} = 1 + \frac{1}{2m}\log \frac{1}{\kappa}$ repetitions are necessary. Using the operator $\tilde{W}^{AB}$ for distillation then gives
		\begin{align}
			\log \lVert \tilde{W}^{AB} \rVert_{\infty} &\leq \log (1+\kappa) \lVert {W}^{AB} \rVert_{\infty} = \log (1+\kappa) + \log  \lVert {W}^{AB} \rVert_{\infty}
		\end{align}
		with a circuit complexity of implementing the approximate operation $\tilde{\Phi}$ of $C(\tilde{\Phi}) \leq O(\poly(m\log\frac{1}{\kappa})) \leq O(\poly(\log d\log\frac{1}{\kappa}))$.
		We can thus conclude that the optimization over $W^{AB} \in \calW(G)$ can be replaced by $W^{AB} \in \calQ(\calH \otimes \calK; G + O(\poly(\log d\log \frac{1}{\kappa}))$ at the cost of distilling $\log (1+\kappa)$ less ebits. We thus have the inequality
		\begin{align}
			E^{\epsilon}_{D, \overline{\mathrm{PPT}}}(\rho^{AB};G)
			&\leq
			\bigg\lfloor \min_{\substack{P^{AB} = (P^{AB})^{\dagger}\\\lVert (P^{AB})^{T_B} \rVert_1 \leq 1}}  D_h^{\epsilon}(\rho^{AB}\fatpipe P^{AB}; G + O( \poly(\log d\log \frac{1}{\kappa})) - \log (1 + \kappa)\big\} \bigg\rfloor
		\end{align}
		for all $\kappa$.
		Having concluded our analysis of $E^{\epsilon}_{D,\overline{\mathrm{PPT}}}(\rho^{AB}; G)$, we are left to connect it back to our initial target, $E^{\epsilon}_{D,{\mathrm{PPT}}}(\rho^{AB}; G)$.
		To do so, we exploit the fact that $\phi^{\otimes m}$ is invariant under unitaries of the form $U \otimes U^{*}$. This means that we can perform a twirl under any 2-design of unitaries $\{ U_i \}_{i=1}^{N_2}$ without changing the outcome of the optimization. Let $C_2$ denote the minimum gate cost of performing such a 2-design twirl on $A'B'$, then the twirling channel
		\begin{align}
			\calT_2[X] \coloneqq \frac{1}{N_2} \sum_{i = 1}^{N_2} (U_i \otimes U_i^*) X (U_i \otimes U_i^*)^{\dagger}
		\end{align}
		has $C(\calT_2) = C_2$ by definition. 
		We have the identity
		\begin{align}
			\Tr[\Phi[\rho^{AB}] \phi^{\otimes m}] &= \Tr[\Phi[\rho^{AB}]\calT_2[\phi^{\otimes m}]] \\
			\nonumber
			&= \Tr[(\calT_2 \circ \Phi)[\rho^{AB}]\phi^{\otimes m}] \\
			&= \Tr[(\bbI \otimes \calT_2)[J_\Phi] ((\rho^{AB})^T \otimes \phi^{\otimes m})],
			\nonumber
		\end{align}
		where we have used  the fact that the twirling channel is self-adjoint. This means the performance of the operation $\Phi$ is the same as the performance of the operation $\calT_2 \circ \Phi$, whose Choi state is $(\bbI \otimes \calT_2)[J_\Phi]$ and whose gate complexity is at most $C(\calT_2 \circ \Phi) \leq C(\calT_2) + C(\Phi) \leq C_2 + G$. We now connect back to the operations in $\overline{\mathrm{PPT}}$, as any 2-design twirl of the form of $\calT_2$ acts as a projector onto the symmetric subspace and its complement, enforcing the very structure of the Choi state in \cref{eqn:ppt_bar_G_definition}.
		This argument shows that 
		\begin{align}
			E^{\epsilon}_{D,\overline{\mathrm{PPT}}}(\rho^{AB}; G) \leq
			E^{\epsilon}_{D,{\mathrm{PPT}}}(\rho^{AB}; G) \leq
			E^{\epsilon}_{D,\overline{\mathrm{PPT}}}(\rho^{AB}; G + C_2), 
		\end{align}
		as for any operation $\Phi \in \mathrm{PPT}(\calH\otimes\calK; G)$, there exists an operation $\calT_2 \circ \Phi\in\overline{\mathrm{PPT}}(\calH \otimes \calK; G + C_2)$ with the same performance. The lower bound follows because $\overline{\mathrm{PPT}}(\calH\otimes\calK;G) \subseteq \mathrm{PPT}(\calH\otimes\calK;G)$.
		
		Note now that we can implement a 2-design twirl by sampling a uniformly random Clifford operator and implementing it. The application of any local Clifford $U\otimes U^{*}$ to $\phi^{\otimes m}$ has gate complexity $O(m^2)$~\cite{gottesman1998theory,dehaene2003clifford}. Moreover, the channel implementing the convex combination $\calT_2$ has polynomial gate complexity as well: it is sufficient to sample a Clifford operation uniformly at random, which can be done efficiently (see,  e.g., Ref.~\cite{berg2021simplemethodsamplingrandom}), and then apply to the the state $\rho_{\tilde{\Lambda}}$. We thus have $C_2 \leq \poly(m)$. As the number of distillable ebits can never exceed $\log d$ for input dimension $d$, we finally obtain $C_2 \leq \poly(\log d)$.
		%
		%
		Combining the two results, we get the desired Proposition.
	\end{proof}

	As a consequence of the above result, we obtain the following theorem.
	\begin{theorem}[Converse bound on computational distillable entanglement]\label{theorem:converse_bound_on_distillable_entanglement}
		Consider a bipartite state $\rho_n^{A_n B_n} \in \calH_n \otimes \calK_n$ such that $\log d_n \coloneqq \min\{ \log \dim(\calH_n), \log \dim(\calK_n)\} \leq O(\poly(n))$. The computational distillable entanglement then fulfills
		\begin{align}
			\underline{E}_D(\rho_n^{A_n B_n}) \lesssim \underline{E}_{D, \mathrm{PPT}}(\rho_n^{A_n B_n}) \lesssim \inf_{\substack{P_n^{A_n B_n} = (P_n^{A_n B_n})^{\dagger}\\\lVert (P_n^{A_n B_n})^{T_{B_n}} \rVert_1 \leq 1}} \underline{D}(\rho_n^{A_n B_n} \fatpipe P_n^{A_n B_n}).
		\end{align}
	\end{theorem}
	\begin{proof}
		The first inequality is immediate from the fact that $\mathrm{LOCC}(\calH_n \otimes \calK_n;G) \subseteq \mathrm{PPT}(\calH_n \otimes \calK_n;G)$, as every $\mathrm{LOCC}$ operation is already $\mathrm{PPT}$.
		The second inequality is an immediate consequence of the single-shot result of \cref{proposition:computational_ppt_distillable_entanglement_single_shot_converse}, which we apply with $\kappa = 1$. We further obtain an upper bound by restricting to operators $G$ that are i.i.d.\ tensor powers and use the fact that the gate overhead $\poly(\log d_n) \leq n^q$ for some $q \in \bbN$ and sufficiently large $n$ by assumption, 
		\begin{align}
			&\underline{E}_{D, \mathrm{PPT}}(\rho_n^{A_n B_n}) \nonumber \\ 
			&\simeq \flim_{\epsilon \to 0} \flim_{\ell \to \infty}\fliminf_{k \to \infty} \frac{1}{n^k} E_{D, \mathrm{PPT}}^{\epsilon}((\rho_n^{A_n B_n})^{\otimes n^k}; n^{k\ell}) \\
			\nonumber
			&\lesssim \flim_{\epsilon \to 0} \flim_{\ell \to \infty}\fliminf_{k \to \infty} \frac{1}{n^k} \left[\min_{\substack{P_n^{A_n B_n} = (P_n^{A_n B_n})^{\dagger}\\\lVert (P_n^{A_n B_n})^{T_{B_n}} \rVert_1 \leq 1}}D_h^{\epsilon}((\rho_n^{A_n B_n})^{\otimes n^k} \fatpipe (P_n^{A_n B_n})^{\otimes n^k}; n^{k\ell} + O(\poly(\log d_n) ) - \log 2 \right]\\
			\nonumber
			&\lesssim \flim_{\ell \to \infty}\fliminf_{k \to \infty}\min_{\substack{P_n^{A_n B_n} = (P_n^{A_n B_n})^{\dagger}\\\lVert (P_n^{A_n B_n})^{T_{B_n}} \rVert_1 \leq 1}}\flim_{\epsilon \to 0}  \frac{1}{n^k} D_h^{\epsilon}((\rho_n^{A_n B_n})^{\otimes n^k} \fatpipe (P_n^{A_n B_n})^{\otimes n^k}; n^{k\ell} + n^q ) \\
			\nonumber
			&\simeq \inf_{\substack{P_n^{A_n B_n} = (P_n^{A_n B_n})^{\dagger}\\\lVert (P_n^{A_n B_n})^{T_{B_n}} \rVert_1 \leq 1}} \underline{D}(\rho_n^{A_n B_n} \fatpipe P_n^{A_n B_n}).
			\nonumber
		\end{align}
		The first equality is the definition of the computational PPT-distillable entanglement. The first inequality applies \cref{proposition:computational_ppt_distillable_entanglement_single_shot_converse} with $\kappa = 1$ and restricts the minimization to i.i.d.\ operators. The second inequality uses \cref{lemma:limit_of_inverse_polynomials_is_negligible} to show that the $\log 2$ factor becomes negligible and bounds the gate overhead $\poly(\log d_n)) \leq n^q$ which holds by assumption for some $q \in \bbN$ and sufficiently large $n$. The last equality follows from the definition of the computational relative entropy and the fact that there exists a sufficiently large $\ell$ to absorb the overhead $n^q$ in the limit. The minimum above exists for every $n$ as the set we optimize over is compact but could depend on the number of gates, hence we obtain the infimum in the sense of \cref{sec:limits_of_polynomial_resources}.
	\end{proof}
	
	The above theorem finally allows us to derive a computational analogue of the classical Rains bound on the distillable entanglement~\cite{rains2001semidefinite} by restricting the optimization over Hermitian operators $P_n^{A_n B_n}$ to quantum states.
	\begin{corollary}[Computational Rains bound]\label{corollary:computational_rains_bound}
		Consider a bipartite state $\rho_n^{A_n B_n} \in \calH_n \otimes \calK_n$ such that $\log d_n \coloneqq \min\{ \log \dim(\calH_n), \log \dim(\calK_n)\} \leq O(\poly(n))$. We define the \emph{computational Rains bound} as
		\begin{align}
			\underline{R}(\rho_n^{A_n B_n}) \colonsimeq \inf_{\substack{\sigma_n^{A_n B_n}\geq 0\\\lVert (\sigma_n^{A_n B_n})^{T_{B_n}} \rVert_1 \leq 1}} \underline{D}(\rho_n^{A_n B_n} \fatpipe \sigma_n^{A_n B_n}).
		\end{align}
		Then,
		\begin{align}
			\underline{E}_D(\rho_n^{A_n B_n}) \lesssim \underline{R}(\rho_n^{A_n B_n}).
		\end{align}
	\end{corollary}
	
	This quantity can be shown to be a proper entanglement monotone under local 
	complexity-limited quantum operations with classical communication, following the lines of thought of Ref.\ \cite{EisertWilde}.
	We note that the computational Rains bound also immediately inherits the computational smoothing property of the computational relative entropy \cref{theorem:undetectability_of_computational_smoothing}, which implies that
	\begin{align}
		\underline{E}_D(\rho_n^{A_n B_n}) \lesssim \inf_{\tilde{\rho}_n^{A_n B_n}\approx_c {\rho}_n^{A_n B_n}} \underline{R}(\tilde{\rho}_n^{A_n B_n}).
	\end{align}
	This in turn implies a statement similar to the vanishing of the computational relative entropy between computationally indistinguishable states of \cref{corollary:zero_rate_for_computationally_indistinguishable_states}.
	\begin{corollary}[Zero computational distillable entanglement via indistinguishability]\label{corollary:zero_distillable_entanglement_from_indistinguishability}
		Consider a bipartite state $\rho_n^{A_n B_n} \in \calH_n \otimes \calK_n$ such that $\log d_n \coloneqq \min\{ \log \dim(\calH_n), \log \dim(\calK_n)\} \leq O(\poly(n))$. Then, if there exists a quantum state $\sigma_n^{A_n B_n}$ with positive partial transpose that is computationally indistinguishable from $\rho_n^{A_n B_n}$, we have zero computational distillable entanglement. Formally,
		\begin{align}
			\rho_n^{A_n B_n} \approx_c \sigma_n^{A_n B_n}\text{ and } \lVert (\sigma_n^{A_n B_n})^{T_{B_n}} \rVert_1 \leq 1 \text{ imply } \underline{E}_D(\rho_n^{A_n B_n}) \simeq 0.
		\end{align}
	\end{corollary}
	\begin{proof}
		We combine the computational Rains bound with computational smoothing of the computational relative entropy (\cref{theorem:undetectability_of_computational_smoothing}) to obtain
		\begin{align}
			\underline{E}_D(\rho_n^{A_n B_n}) \lesssim \underline{R}(\rho_n^{A_n B_n}) \lesssim \underline{D}(\rho_n^{A_n B_n} \fatpipe \sigma_n^{A_n B_n}) \simeq \underline{D}(\sigma_n^{A_n B_n} \fatpipe \sigma_n^{A_n B_n}) \simeq 0.
		\end{align}
	\end{proof}
	
	We can also derive another interesting bound from \cref{theorem:converse_bound_on_distillable_entanglement}, namely one that gives a computational analog to the log-negativity bound on distillable entanglement.
	\begin{corollary}[Pseudo log-negativity bound]\label{corollary:pseudo_log_negativity_bound}
		Consider a bipartite state $\rho_n^{A_n B_n} \in \calH_n \otimes \calK_n$ such that $\log d_n \coloneqq \min\{ \log \dim(\calH_n), \log \dim(\calK_n)\} \leq O(\poly(n))$. We define its \emph{pseudo log-negativity} as
		\begin{align}
			E_{N,\mathrm{Pseudo}}(\rho_n^{A_n B_n}) \colonsimeq \inf_{\tilde\rho_n^{A_n B_n} \approx_c \rho_n^{A_n B_n}} \log \lVert (\tilde\rho_n^{A_n B_n})^{T_{B_n}} \rVert_1.
		\end{align}
		We have
		\begin{align}
			\underline{E}_{D}(\rho_n^{A_n B_n}) \lesssim        
			\underline{E}_{D,\mathrm{PPT}}(\rho_n^{A_n B_n}) \lesssim
			E_{N,\mathrm{Pseudo}}(\rho_n^{A_n B_n}).
		\end{align}
	\end{corollary}
	\begin{proof}
		We combine \cref{theorem:converse_bound_on_distillable_entanglement} with computational smoothing. Let $\tilde\rho_n^{A_n B_n}$ any state computationally indistinguishable from $\rho_n^{A_n B_n}$. 
		Then, we have the chain of inequalities
		\begin{align}
			\underline{E}_{D}(\rho_n^{A_n B_n}) &\lesssim        
			\underline{E}_{D,\mathrm{PPT}}(\rho_n^{A_n B_n}) \\
			&\lesssim \underline{D}(\rho_n^{A_n B_n} \fatpipe \tilde\rho_n^{A_n B_n} / \lVert (\tilde\rho_n^{A_n B_n})^{T_{B_n}} \rVert_1)\nonumber \\
			&\simeq \underline{D}(\tilde\rho_n^{A_n B_n} \fatpipe \tilde\rho_n^{A_n B_n} / \lVert (\tilde\rho_n^{A_n B_n})^{T_{B_n}} \rVert_1)\nonumber \\
			&\simeq \underline{D}(\tilde\rho_n^{A_n B_n} \fatpipe \tilde\rho_n^{A_n B_n}) + \log \lVert (\tilde\rho_n^{A_n B_n})^{T_{B_n}} \rVert_1 \nonumber \\
			&\simeq \log \lVert (\tilde\rho_n^{A_n B_n})^{T_{B_n}} \rVert_1. \nonumber
		\end{align}
		The first inequality follows by inclusion of LOCC in PPT. The second inequality is \cref{theorem:converse_bound_on_distillable_entanglement} with $P_n^{A_n B_N} = \tilde\rho_n^{A_n B_n} / \lVert (\tilde\rho_n^{A_n B_n})^{T_{B_n}} \rVert_1$. The first equality uses computational smoothing and the fact that $\rho_n^{A_n B_n} \approx_c \tilde\rho_n^{A_n B_n}$. The second equality uses \cref{item:computational_entropy_multiplication_of_second_argument}. The third equality uses $\underline{D}(\tilde\rho_n^{A_n B_n} \fatpipe \tilde\rho_n^{A_n B_n}) \simeq 0$. As the above holds for any $\tilde\rho_n^{A_nB_n}$ computationally indistinguishable from $\rho_n^{A_n B_n}$, it holds for any sequence approaching the infimum in the definition of $E_{N,\mathrm{Pseudo}}(\rho_n^{A_n B_n})$, from which we conclude the statement of the Corollary.
	\end{proof}

	\begin{corollary}[Pseudo entanglement entropy bound]\label{corollary:pseudo_entanglement_entropy_bound}
		Consider a bipartite pure state $\psi_n^{A_n B_n} \in \calH_n \otimes \calK_n$ such that $\log d_n \coloneqq \min\{ \log \dim(\calH_n), \log \dim(\calK_n)\} \leq O(\poly(n))$. We define its \emph{pseudo von Neumann entropy of entanglement} as
		\begin{align}
			S_{\mathrm{Pseudo}}(\psi_n^{A_n B_n}) \colonsimeq \inf_{\tilde\psi_n^{A_n B_n} \approx_c \psi_n^{A_n B_n}} S(\tilde{\psi}_{n}^{A_n})
		\end{align}
		where $S$ is the von Neumann entropy and $\psi_n^{A_n}\coloneqq \Tr_{B_n}\psi_n^{A_nB_n}$. We have
		\begin{align}
			\underline{E}_{D}(\psi_n^{A_n B_n}) \lesssim        
			\underline{E}_{D,\mathrm{PPT}}(\psi_n^{A_n B_n}) \lesssim
			S_{\mathrm{Pseudo}}(\psi_n^{A_n B_n}).
		\end{align}
	\end{corollary}
	\begin{proof}
		Let us show that there exists always a pure state achieving the minimum. Let $\tilde{\rho}_{n}^{A_nB_n}$ be any state computationally indistinguishable from $\psi_n^{A_n B_n}$. We notice that if  $\tilde{\rho}_{n}^{A_nB_n}\approx_c \psi_n^{A_nB_n}$, then $\tilde{\rho}_{n}^{A_nB_n}$ must necessarily obey to $\|\tilde{\rho}_{n}^{A_nB_n}\|_{\infty}\simeq 1$ as there exists a efficient quantum algorithm that measures $\|\tilde{\rho}_{n}^{A_nB_n}\|_{\infty}$ up to polynomial precision~\cite{odonnell2015efficientquantumtomography}. Let $\tilde{\psi}_n^{A_nB_n}$ the (unique) eigenvector of $\tilde{\rho}_{n}^{A_nB_n}$ with maximum eigenvalue. We have $\|\tilde{\psi}_n^{A_nB_n}-\tilde{\rho}_{n}^{A_nB_n}\|_1\le 2\sqrt{1-\Tr(\tilde{\psi}_n^{A_nB_n}\tilde{\rho}_{n}^{A_nB_n})}\simeq 0$. We can therefore apply Fannes inequality:
		\begin{align}
			|S(\tilde{\rho}_{n}^{A_n})-S(\tilde{\psi}_{n}^{A_n})|\le \frac{1}{2}\log(d_n-1)\|\tilde{\rho}_{n}^{A_n}-\tilde{\psi}_{n}^{A_n}\|_1+h(\|\tilde{\rho}_{n}^{A_n}-\tilde{\psi}_{n}^{A_n}\|_1/2)
		\end{align}
		where $h(\cdot)$ is the binary entropy. By data processing inequality it follows that $\|\tilde{\rho}_{n}^{A_n}-\tilde{\psi}_{n}^{A_n}\|_1\le \|\tilde{\rho}_{n}^{A_nB_n}-\tilde{\psi}_{n}^{A_nB_n}\|_1\simeq 0$. Since $\log (d_n-1)=O(\operatorname{poly}n)$, we have  $S(\tilde{\rho}_{n}^{A_n})\simeq S(\tilde{\psi}_{n}^{A_nB_n})$. 
		Let therefore $\tilde{\psi}_n^{A_nB_n}$ be a pure state computationally indistinguishable from $\psi_n^{A_nB_n}$ and let $\ket{\tilde{\psi}_n^{A_nB_n}}$ its state vector. By Schmidt decomposition, we can write $\ket{\tilde{\psi}_n^{A_nB_n}}=\sum_{i}\sqrt{\lambda_i}\ket{i_A}\ket{i_B}$, from which we can define $\sigma_{n}^{A_nB_n}\coloneqq \sum_{i}\lambda_i\ket{i_Ai_B}\bra{i_Ai_B}$. Note that $\|(\sigma_n^{A_nB_n})^{T_{B_n}}\|_1=1$ and, as such, can be used in the optimization in \cref{theorem:converse_bound_on_distillable_entanglement}. We have the following chain of inequalities
		\begin{align}
			\underline{E}_{D}(\psi_n^{A_n B_n}) &\lesssim        
			\underline{E}_{D,\mathrm{PPT}}(\psi_n^{A_n B_n}) \\
			\nonumber
			&\lesssim \underline{D}(\psi_n^{A_nB_n}\| \sigma_{n}^{A_nB_n})\\
			\nonumber
			&\simeq \underline{D}(\tilde{\psi}_n^{A_nB_n}\| \sigma_{n}^{A_nB_n})\\
			\nonumber
			&\lesssim D(\tilde{\psi}_{n}^{A_nB_n}\|\sigma_n^{A_nB_n})\\
			\nonumber
			&= S(\tilde{\psi}_{n}^{A_n}).
			\nonumber
		\end{align}
		The first inequality follows by inclusion of LOCC in PPT. The second inequality is \cref{theorem:converse_bound_on_distillable_entanglement} with $P_n^{A_n B_n} = \sigma_{n}^{A_nB_n}$. The first equality uses computational smoothing and the fact that $\psi_n^{A_n B_n} \approx_c \tilde\psi_n^{A_n B_n}$. The second equality uses \cref{item:computational_relative_entropy_boundedness}. The third equality follows by the  simple calculation
		\begin{align}
			D(\psi_n^{A_nB_n}\|\sigma_n^{A_nB_n})=\Tr[\psi_n^{A_nB_n}\log \psi_n^{A_nB_n}]-\Tr[\psi_n^{A_nB_n}\log\sigma_n^{A_nB_n}] =-\Tr[\sigma_n^{A_nB_n}\log\sigma_n^{A_nB_n}]=S(\psi_n^{A_n}),
		\end{align}
		where the first equality follows by definition; the second from the fact that $\psi_n^{A_nB_n}$ is a pure state and that $\sigma_n^{A_nB_n}$ corresponds to the dephased version of $\psi_n^{A_nB_n}$, the last equality follows by a direct identity.
		We conclude by definition of the pseudo von Neumann entropy of entanglement, noting that the argument above works for any computationally indistinguishable pure state and hence for any sequence that approaches the infimum.
	\end{proof}
	
	We can immediately apply \cref{corollary:pseudo_entanglement_entropy_bound} to re-derive Theorem 3 of Ref.\ \cite{leone2025entanglementtheorylimitedcomputational} as a simple corollary. 
	
	\begin{corollary}[Pseudo entropy bound]\label{corollary:pseudo_entanglement_entropy_bound1}
		There exists a family of pure bipartite states $\psi_{n}^{A_nB_n}$ indexed by the number of qubits $n$ such that $\underline{E}_D(\psi_{n}^{A_nB_n})\lesssim S_{\min}(\psi_{n}^{A_n})+o(1)$, regardless the value of distillable entanglement, which can be maximal up to log factors.
	\end{corollary}
	\begin{proof}
		We consider the families of states (parametrized by $n$ the number of qubits) $\mathcal{E}_{\mathrm{Haar},\eta,S_{\min}}$ and $\mathcal{E}_{m,\eta,S_{\min}}$ introduced in Lemma 19 of Ref.~\cite{leone2025entanglementtheorylimitedcomputational}. Theorems 23 and 24 of Ref.~\cite{leone2025entanglementtheorylimitedcomputational} show that, for each $n$, there exist a state $\psi_{n}^{A_nB_n}\in \mathcal{E}_{\mathrm{Haar},\eta,S_{\min}}$ such that $\psi_{n}^{A_nB_n}\approx_c \tilde{\psi_{n}}^{A_nB_n}$ with $\tilde{\psi_{n}}^{A_nB_n}\in \mathcal{E}_{m,\eta,S_{\min}}$. According to Lemma 20 of Ref.~\cite{leone2025entanglementtheorylimitedcomputational} it holds that $S(\tilde{\psi}_{n}^{A_n})\le S_{\min}(\tilde{\psi}_n^{A_n})+o(1)$. We conclude by directly applying \cref{corollary:pseudo_entanglement_entropy_bound}. 
	\end{proof}

	\subsection*{Computational entanglement dilution}
	The converse task to entanglement distillation is entanglement dilution, where we are tasked to find the minimum number of ebits necessary to prepare a target state across the bipartition with a restricted set of operations $O$. This number is the \emph{entanglement cost}. In this section, we consider a computational version of the entanglement cost. Again, we follow the blueprint of giving a complexity-limited single-shot definition that is then polynomially regularized to yield the corresponding computational quantity.
	\begin{definition}[$G$-complexity entanglement cost]\label{def:g_complexity_entanglement_cost}
		Consider a bipartite state $\rho^{AB} \in \calH \otimes \calK$. We define the $G$-complexity $\epsilon$-entanglement cost under the class of operations $O$ as 
		\begin{align}
			E_{C,O}^{\epsilon}(\rho^{AB}; G) &\coloneqq  \log 2 \times \min \big\{ m \in \bbN \pipe \text{there exists } \Phi \in O(\bbC_2^{\otimes m} \otimes \bbC_2^{\otimes m}; G) \text{ such that } F(\Phi[\phi^{\otimes m}],\rho^{AB}) \geq 1 - \epsilon \big\},
		\end{align}
		where $\phi$ is the two-qubit maximally entangled state and the factor $\log 2$ converts from bits to nats. We let the entanglement cost be $\infty$ if no feasible dilution protocol exists.
	\end{definition} 
	Polynomial regularization then allows us to define the computational entanglement cost as follows.
	
	\begin{definition}[Computational entanglement cost rate]\label{definition:computational_entanglement_cost_rate}
		Consider a bipartite state $\rho_n^{A_nB_n} \in \calH_n \otimes \calK_n$. Then,
		\begin{align}
			\underline{E}_{C, O}(\rho_n^{A_n B_n}) \colonsimeq \flim_{\epsilon \to 0}\flim_{\ell\to\infty}\fliminf_{k\to\infty} \frac{1}{n^k} E_{C, O}^{\epsilon}( (\rho_n^{A_n B_n})^{\otimes n^k}; n^{k\ell}).
		\end{align}            
	\end{definition}      
	
	Importantly, the unbounded entanglement cost is always larger than the distillable entanglement. We show that the same holds in the computational setting.
	\begin{proposition}\label{proposition:computational_entanglement_cost_larger_than_distillable_entanglement}
		Consider a bipartite state $\rho_n^{A_n B_n} \in \calH_n \otimes \calK_n$. Then,
		\begin{align}
			\underline{E}_{C}(\rho_n^{A_n B_n}) \gtrsim \underline{E}_{D}(\rho_n^{A_n B_n}).
		\end{align}
	\end{proposition}
	\begin{proof}
		We follow a similar reasoning as Ref.~\cite{arnon2023computationalentanglementtheory}. In our case, the result follows from Ref.~\cite[Theorem 1]{wilde2021second}, which shows that in the unbounded setting the following single-shot result holds:
		\begin{align}
			E_{D}^{\epsilon_1}(\rho_n^{A_n B_n}) \leq E_{C}^{\epsilon_2}(\rho_n^{A_n B_n}) + \log \frac{1}{1-(\sqrt{\epsilon_1} + \sqrt{\epsilon_2})^2},
		\end{align}
		as long as $\sqrt{\epsilon_1} + \sqrt{\epsilon_2} \leq 1$. In our case, this means we have 
		\begin{align}
			\underline{E}_{C}(\rho_n^{A_n B_n})&\simeq \flim_{\epsilon \to 0}\flim_{\ell\to\infty}\fliminf_{k\to\infty} \frac{1}{n^k} E_{C}^{\epsilon}( (\rho_n^{A_n B_n})^{\otimes n^k}; n^{k\ell}) \\
			&\gtrsim  \flim_{\epsilon \to 0}\flim_{\ell\to\infty}\fliminf_{k\to\infty} \frac{1}{n^k} E_{C}^{\epsilon}( (\rho_n^{A_n B_n})^{\otimes n^k}) \nonumber \\
			&\gtrsim \flim_{\epsilon \to 0}\flim_{\ell\to\infty}\fliminf_{k\to\infty} \frac{1}{n^k} E_{D}^{\epsilon}( (\rho_n^{A_n B_n})^{\otimes n^k}) - \frac{1}{n^k} \log \frac{1}{1 - 4 \epsilon} \nonumber \\
			&\gtrsim \flim_{\epsilon \to 0}\flim_{\ell\to\infty}\fliminf_{k\to\infty} \frac{1}{n^k} E_{D}^{\epsilon}( (\rho_n^{A_n B_n})^{\otimes n^k}; n^{k\ell}) - \frac{1}{n^k} \log \frac{1}{1 - 4 \epsilon} \nonumber \\
			&\simeq \underline{E}_D(\rho_n^{A_n B_n}).\nonumber 
		\end{align}
		Here, we have used  the fact that $\epsilon \to 0$ to ensure that $4\epsilon \leq 1$ and concluded by \cref{lemma:limit_of_inverse_polynomials_is_negligible}.
	\end{proof}

	In the unbounded setting, the entanglement cost of a pure quantum state is fully characterized by the \emph{entanglement entropy} of the quantum state, which is the entropy of the subsystems which is guaranteed to be equal. We can show a corresponding upper bound in the computational setting that holds for arbitrary bipartite states. Interestingly, the computational entropy of the two subsystems need not be identical, which is why the minimal computational entropy of the subsystems appears in the bound. The non-identity can be easily seen by considering a pure product state where one system is in a computational basis state but the other system is in a Haar random state, clearly the computational entropies are widely different while the unbounded entropies are the same.
	\begin{theorem}[Computational entanglement cost by compression]\label{theorem:computational_entanglement_cost_via_compression}
		Consider a bipartite state $\rho_n^{A_n B_n} \in \calH_n \otimes \calK_n$ of polynomial complexity. Then, 
		\begin{align}
			\underline{E}_{C}(\rho_n^{A_n B_n}) \lesssim \min\{ \underline{S}(\rho_n^{A_n}),  \underline{S}(\rho_n^{B_n})\}.
		\end{align}
	\end{theorem}
	\begin{proof}
		We let, without loss of generality, $\underline{S}(\rho_n^{A_n})\lesssim \underline{S}(\rho_n^{B_n})$. 
		We recall the definition of the computational entanglement cost rate:
		\begin{align}
			\underline{E}_{C}(\rho_n^{A_n B_n}) \simeq \flim_{\epsilon \to 0}\flim_{\ell\to\infty}\fliminf_{k\to\infty} \frac{1}{n^k} E_{C}^{\epsilon}( (\rho_n^{A_n B_n})^{\otimes n^k}; n^{k\ell}).
		\end{align}
		The protocol is then to let Bob prepare the state $(\rho_n^{A_n B_n})^{\otimes n^k}$ on his system, which is a polynomial complexity operation because we assumed $\rho_n^{A_n B_n}$ to be of polynomial complexity. He then performs the optimal compression of the $n^k$ subsystems $A_n$ to total variation distance $\epsilon/2$ using $G_c$ many gates. This guarantees that the resulting state is compressed into
		\begin{align}
			\log M^{\epsilon/2}((\rho_n^{A_n})^{\otimes n^k}; G_c) / \log 2
		\end{align}
		many qubits. The guarantees for compression are in terms of the total variation distance, but we can use the Fuchs-van-de-Graaf inequality to relate it back to the fidelity. A trace distance of at most $\epsilon/2$ guarantees a fidelity of at least $1-\epsilon$.
		He then teleports these qubits to Alice, using as many ebits in the process. Note that teleportation is an efficient transformation with complexity polynomial in the number of ebits used. As the complexity of the state is polynomial, it cannot cost more than polynomially many ebits to send it. Let the complexity of the teleportation part be $G_T \leq n^{q}$ for sufficiently large $n$. Alice then decompresses the state on her side at a gate cost of $G_c$ as well according to \cref{definition:computational_compression_task}. After executing this protocol, Alice and Bob hold the joint state $(\sigma_n^{A_n B_n})^{\otimes n^k}$ which fulfills
		\begin{align}
			F((\rho_n^{A_n B_n})^{\otimes n^k}, (\sigma_n^{A_n B_n})^{\otimes n^k}) &= F((\rho_n^{A_n B_n})^{\otimes n^k}, ([E\circ D]^{A_n^{\otimes n^k}} \otimes \bbI^{B_n^{\otimes n^k}})(\rho_n^{A_n B_n})^{\otimes n^k}) \\
			&\geq F((\rho_n^{A_n})^{\otimes n^k}, (E\circ D)[(\rho_n^{A_n})^{\otimes n^k}]) \nonumber\\
			&\geq 1-\epsilon.\nonumber
		\end{align}
		The first equality follows from the fact that teleportation only changes the location of the subsystems and not the state itself. The first inequality is data processing. The second inequality is the guarantee of the compression scheme.
		This protocol is hence an entanglement dilution protocol with fidelity $1-\epsilon$, that uses $\log M^{\epsilon/2}((\rho_n^{A_n})^{\otimes n^k}; G_c)/\log 2$ ebits at a gate cost of $n^k C(\rho_n^{A_n B_n}) + 2 G_c + G_T$ many gates. We can now let $G_c = n^{k\ell}/3$ and wait until $\ell$ gets large enough such that $G_c \geq G_T + n^{k} C(\rho_n^{A_n B_n})$. In this case we have that
		\begin{align}
			\underline{E}_{C}(\rho_n^{A_n B_n}) &\simeq \flim_{\epsilon \to 0}\flim_{\ell\to\infty}\fliminf_{k\to\infty} \frac{1}{n^k} E_{C}^{\epsilon}( (\rho_n^{A_n B_n})^{\otimes n^k}; n^{k\ell}) \\
			&\lesssim \flim_{\epsilon \to 0}\flim_{\ell\to\infty}\fliminf_{k\to\infty} \frac{1}{n^k} \log M^{\epsilon/2}((\rho_n^{A_n})^{\otimes n^k}; n^{k\ell}/3) \nonumber\\
			&\simeq \underline{C}(\rho_n^{A_n})\nonumber\\
			&\simeq \underline{S}(\rho_n^{A_n}).\nonumber
		\end{align}
		Here, we have used the fact that the computational entropy quantifies the computational compression task of \cref{theorem:computational_entropy_as_compression_rate} and reasoned that the factor $1/3$ does not change the limit.
	\end{proof}

\end{document}